\documentclass[11pt]{article}
\pdfoutput=1
\usepackage[utf8]{inputenc}
\usepackage{multirow}
\usepackage{amsmath, amsfonts, amssymb}
\usepackage{comment}
\usepackage{graphicx}
\usepackage{psfrag}
\usepackage{amsthm}
\usepackage[usenames,dvipsnames,svgnames,table]{xcolor}
\usepackage{enumerate}
\usepackage{arydshln}
\usepackage{soul}
 \usepackage{slashed}
\usepackage{subcaption}
\usepackage{adamsmacros}
\usepackage{mathrsfs}
 \usepackage{a4wide}
  \usepackage{tikz}
  \usepackage{tikz-cd}
  \usetikzlibrary{shapes.geometric}
  \usepackage{tcolorbox}
  \usepackage{quiver}

  \usepackage[usenames, dvipsnames]{xcolor}
  \definecolor{dark-gray}{gray}{0.20}
  \definecolor{gray}{gray}{0.30}
  \definecolor{light-gray}{gray}{0.80}
  \definecolor{dark-red}{rgb}{0.7,0,0}
  \definecolor{dark-green}{rgb}{0.1,0.4,0}
  \definecolor{dark-blue}{rgb}{0.3,0.3,0.7}
  \definecolor{light-blue}{rgb}{0.8,0.8,1}
      \definecolor{swamp}{RGB}{240, 199, 197}

  \usepackage{pifont}

\usepackage{newunicodechar} 
\newunicodechar{ỳ}{\`y}

\usepackage{setspace}

\newcommand{\be}{\begin{equation}}
\newcommand{\ee}{\end{equation}}
\newcommand{\eq}[1]{(\ref{#1})}

\def\be{\begin{equation}}
\def\ee{\end{equation}}
\def\bea{\begin{eqnarray}}
\def\eea{\end{eqnarray}}

\newcommand{\shortexact}[3]{
\begin{tikzcd}[ampersand replacement=\&]
	0 \& #1
	\&  #2
	\& #3
	\& 0
	\arrow[from=1-1, to=1-2]
	\arrow[from=1-2, to=1-3]
	\arrow[from=1-3, to=1-4]
	\arrow[from=1-4, to=1-5]
\end{tikzcd}
}

\newenvironment{eqaed}
    {\begin{equation}
    \begin{aligned}
    }
    { 
    \end{aligned}
    \end{equation}
    \ignorespacesafterend
    }

\newcommand{\C}{\mathbb{C}}
\newcommand{\R}{\mathbb{R}}

\newcommand{\Z}{\mathbb{Z}}
\AtBeginDocument{
\renewcommand{\H}{\mathbb{H}}
\renewcommand{\O}{\mathit O} 
}
\newcommand{\abs}[1]{\lvert #1 \rvert}
\newcommand{\ang}[1]{\langle #1 \rangle}

\newcommand{\Spin}{\mathit{Spin}}

\newcommand{\Sq}{\mathrm{Sq}}

\newcommand{\USp}{\mathit{USp}}
\newcommand{\SU}{\mathit{SU}}
\newcommand{\U}{\mathit U}
\newcommand{\Sp}{\mathit{Sp}}
\newcommand{\SO}{\mathit{SO}}
\newcommand{\Ext}{\mathrm{Ext}}

\newcommand{\tmf}{\mathit{tmf}}
\newcommand{\ko}{\mathit{ko}}
\newcommand{\pt}{\mathrm{pt}}

\newcommand{\String}{\mathrm{String}}

\newcommand{\cA}{\mathcal A}
\newcommand{\cP}{\mathcal P}

\newcommand{\RP}{\mathbb{RP}}

\newcommand{\HP}{\mathbb{HP}}

\newcommand{\Hom}{\mathrm{Hom}}

\newcommand{\MTSpin}{\mathit{MTSpin}}

\newcommand{\G}{\mathbb{G}}
\newcommand{\cQ}{\mathcal Q}

\def\simleq{\; \raise0.3ex\hbox{$<$\kern-0.75em
      \raise-1.1ex\hbox{$\sim$}}\; }
   \def\simgeq{\; \raise0.3ex\hbox{$>$\kern-0.75em
      \raise-1.1ex\hbox{$\sim$}}\; }

\numberwithin{equation}{section}

\usepackage{jheppub}
\usepackage{hyperref}
\usepackage{cleveref}

\hypersetup{
	colorlinks=true,
	linkcolor=dark-blue,
	citecolor=dark-red,
	urlcolor=dark-green,
	linktoc=page
}
\newtheorem{thm}[equation]{Theorem}
\newtheorem{prop}[equation]{Proposition}
\newtheorem{lem}[equation]{Lemma}
\newtheorem{cor}[equation]{Corollary}
\theoremstyle{definition}
\newtheorem{defn}[equation]{Definition}
\theoremstyle{remark}
\newtheorem{rem}[equation]{Remark}
\crefname{appendix}{Appendix}{Appendices}

\title{\centering Global anomalies \& bordism \\ of  non-supersymmetric strings}

\author{Ivano Basile$^1$,} \affiliation{$^1$ Arnold Sommerfeld Center for Theoretical Physics • Ludwig-Maximilians-Universität (LMU), Theresienstraße 37, 80333 München, Germany}
\author{Arun Debray$^2$,}\affiliation{$^2$ Purdue University, West Lafayette, IN 47907, United States}
\author{Matilda Delgado$^3$,} \affiliation{$^3$ Instituto de F\'{i}sica Te\'{o}rica IFT-UAM/CSIC,
C/ Nicol\'{a}s Cabrera 13-15, Campus de Cantoblanco, 28049 Madrid, Spain}
\author{and Miguel Montero$^3$.}

\emailAdd{ivano.basile@lmu.de}
\emailAdd{adebray@purdue.edu}
\emailAdd{matilda.delgado@uam.es}
\emailAdd{miguel.montero@csic.es}

\abstract{The three tachyon-free non-supersymmetric string theories in ten dimensions provide a handle on quantum gravity away from the supersymmetric lamppost. However, they have not been shown to be fully consistent; although local anomalies cancel due to versions of the Green-Schwarz mechanism, there could be global anomalies, not cancelled by the Green-Schwarz mechanism, that could become fatal pathologies. We compute the twisted string bordism groups that control these anomalies via the Adams spectral sequence, showing that they vanish completely in two out of three cases  (Sugimoto and $SO(16)^2$) and showing a partial vanishing also in the third (Sagnotti 0'B model). We also compute lower-dimensional bordism groups of the non-supersymmetric string theories, which are of interest to the classification of branes in these theories via the Cobordism Conjecture. We propose a worldvolume content based on anomaly inflow for the $SO(16)^2$ NS5-brane, and discuss subtleties related to the torsion part of the Bianchi identity. As a byproduct of our techniques and analysis, we also reprove that the outer $\Z_2$ automorphism swapping the two $E_8$ factors in the supersymmetric heterotic string is also non-anomalous.}

\begin{document}
\emergencystretch 3em
\hypersetup{pageanchor=false}
\makeatletter
\let\old@fpheader\@fpheader
\preprint{IFT-23-129}

\makeatother

\maketitle

\hypersetup{
    pdftitle={},
    pdfauthor={},
    pdfsubject={}
}

\newcommand{\remove}[1]{\textcolor{red}{\sout{#1}}}

\section{Introduction}
It is widely known that there are five supersymmetric string theories in ten dimensions \cite{Polchinski:1998rq}. It is slightly less known that there are several other non-supersymmetric string theories in ten dimensions, many of which have tachyons; and there are just three known models in ten dimensions which are both non-supersymmetric in spacetime and tachyon free: the $SO(16)\times SO(16)$ string \cite{Alvarez-Gaume:1986ghj, Dixon:1986iz}, the Sugimoto model \cite{Sugimoto:1999tx}, and the Sagnotti 0'B model \cite{Sagnotti:1995ga, Sagnotti:1996qj}. These non-supersymmetric models and their compactifications have been the subject of a renovated interest in the recent literature (see e.g. \cite{Abel:2016pwa,Abel:2018zyt,McGuigan:2019gdb,Itoyama:2020ifw,Angelantonj:2020pyr,Basile:2020mpt,Basile:2020xwi,Kaidi:2020jla,Faraggi:2020hpy,Cribiori:2020sct,Itoyama:2021fwc,Gonzalo:2021fma,Perez-Martinez:2021zjj,Basile:2021mkd,Basile:2021vxh,Sagnotti:2021mxb,Itoyama:2021itj,Cribiori:2021txm,Percival:2022uvp,Raucci:2022bjw,Basile:2022ypo,Baykara:2022cwj,Koga:2022qch,Cervantes:2023wti,Avalos:2023mti,Raucci:2023xgx,Mourad:2023wjg,Avalos:2023ldc, Fraiman:2023cpa}), presumably because they constitute a promising arena to study quantum gravity away from the supersymmetric lamppost. 

In spite of this recent surge of work, we still know remarkably little about the three tachyon-free non-supersymmetric models in ten dimensions, particularly when compared with their supersymmetric counterparts. In particular, the spectrum of all three non-supersymmetric models is chiral in ten dimensions, and so it is potentially anomalous. Local anomalies have long been known to cancel via non-supersymmetric versions of the Green-Schwarz mechanism \cite{Alvarez-Gaume:1986ghj,Dixon:1986iz,Sagnotti:1995ga, Sagnotti:1996qj,Sugimoto:1999tx}. However, to our knowledge, except for an inconclusive analysis in \cite{osti_5515225} these models have not been shown to be free of global anomalies, as was done early on in \cite{Witten:1985xe,Witten:1985mj,Freed:2000ta} for the supersymmetric chiral theories.\footnote{Even in this case, only gravitational and global anomalies in the identity component of the gauge group have been considered.} A global anomaly could lead to an inconsistency of the model, and having to discard it, or to new topological couplings that cancel it \cite{Debray:2021vob}.

The purpose of this paper is twofold: On one hand, we compute the potential global anomaly of the three tachyon-free non-supersymmetric string theories in quantum gravity (albeit only a subclass of anomalies for the Sagnotti 0'B model), showing that it vanishes. We do this using cobordism theory, and computing the relevant twisted string bordism groups. This is a standard technique for proving anomaly cancellation results; see~\cite{Freed:2019sco, Tachikawa:2021mby, Debray:2021vob, LY22, DY22, Tachikawa:2021mvw, Deb23, DOS23} for recent anomaly cancellation theorems in string and supergravity theories using this
technique.
 
 The second important result of our paper is precisely the calculation of these twisted string bordism groups (which have not appeared in the literature before), which are summarized in Table \ref{tbor} in the conclusions. The physics use of these bordism groups is that they can be used to predict new, singular configurations (branes) of the corresponding non-supersymmetric string theories, by means of the Cobordism Conjecture \cite{McNamara:2019rup} (see \cite{Dierigl:2022reg,Kaidi:2023tqo,Debray:2023yrs,Andriot:2022mri} for similar recent work in type II and supersymmetric heterotic string theories). Furthermore, the calculation of the bordism groups themselves by means of the Adams spectral sequence is interesting in its own right, and we expect that similar techniques can be used to compute string bordism groups of e.g. six-dimensional compactifications, and more generally, to study anomalies of any theory with a 2-group symmetry or a Green-Schwarz mechanism. 

Along the way, we will encounter and comment on issues such as whether the heterotic Bianchi identity can be taken to take values on the free part of cohomology or the torsion piece must be included, or the connections between anomaly cancellation in eleven-dimensional backgrounds and anomaly inflow on non-supersymmetric NS5 branes on these theories. We also include a quick introduction to the Green-Schwarz mechanism in the modern formalism of anomaly theory, providing for the first time a candidate for the worldvolume degrees of freedom for the NS5 brane in the $SO(16)\times SO(16)$ string.  We also study global anomalies in the $\Z_2$ outer automorphism swapping the two factors of the $SO(16)\times SO(16)$ string, showing that anomalies vanish.

The upshot of our paper is:
\begin{itemize}
    \item The bordism group controlling anomalies of the Sugimoto string, $\Omega_{11}^{\String\text{-}\Sp(16)}$, vanishes (\cref{sugimoto_thm}), and therefore the theory is anomaly-free\footnote{Modulo potential subtleties regarding the global structure of the gauge group, that we comment on in the Conclusions.}.
    \item The bordism group $\Omega_{11}^{\String\text{-}\SU(32)\ang{c_3}}$, controlling the anomaly of the Sagnotti 0'B model, is isomorphic to $0$ or $\Z_2$ (\cref{sag_thm}). We do not know whether the anomaly vanishes, although it does in all specific backgrounds we looked into.
    \item For the $ SO(16)\times SO(16)$ heterotic string, (where the identity component of the global form of the gauge group is actually $\Spin(16)\times\Spin(16)$, since the massless spectrum contains both spinors and vectors\footnote{See \cite{McI99} for an analysis detailing some possibilities for a global quotient in the $SO(16)^2$ gauge group. In this paper, we assume the simply-connected global form $\Spin(16)^2$ (so ``$SO(16)^2$'' is an abuse of notation); this has the advantage that all anomalies we find also exist for any other possibility (although with a nontrivial global quotient, there could be more anomalies than the ones that we study here). }), the bordism group $\Omega_{11}^{\String\text{-}\Spin(16)^2}$ controlling the anomaly vanishes (\cref{noswap}), and therefore this theory is anomaly-free.
    \item There is also a $\Z_2$ gauge symmetry swapping the two factors of $\Spin(16)$, whose anomaly we also studied. The bordism group controlling the anomaly has order $64$ --- but nevertheless (\cref{non_susy_swap_anomaly}), the anomaly vanishes.
\end{itemize}
As a consequence of our calculations, we also can cancel an anomaly in a supersymmetric string theory.
\begin{itemize}
    \item When one takes into account the $\Z_2$ symmetry of the $E_8\times E_8$ heterotic string swapping the two copies of $E_8$, the anomaly vanishes (\cref{susy_swap_anomaly}).
\end{itemize}
The cancellation of this anomaly is not a new result: it is a special case of the more general work of~\cite{Tachikawa:2021mby}. Our argument rests on different physical assumptions and is a different mathematical result; for example, we do not assume the Stolz-Teichner conjecture. Thus we answer a question of~\cite{Deb23}, who showed the bordism group controlling this anomaly has order $64$ but did not address the anomaly, and asked for a bordism-theoretic argument that the anomaly vanishes.

One key application of this $\Z_2$ symmetry of the $E_8\times E_8$ heterotic string is constructing the CHL string~\cite{CHL95}, a nine-dimensional string theory obtained by compactifying the $E_8\times E_8$ heterotic string on a circle, where the monodromy around the circle is the $\Z_2$ symmetry we discussed above. An anomaly in the $\Z_2$ symmetry would have implied an inconsistency in the CHL string. We find that the anomaly vanishes, in agreement with the results in \cite{Tachikawa:2021mby}, which showed this from a worldsheet perspective; by contrast, we approach the question from a pure spacetime perspective.

One can make an analogous construction for the $\SO(16)\times\SO(16)$ heterotic string, compactifying it on a circle whose monodromy exchanges the two bundles. The result is a nine-dimensional non-supersymmetric string theory whose gauge group is (perhaps a quotient of) $\Spin(16)$. Studying this theory would be an interesting extension of similar constructions in the $E_8 \times E_8$ case~\cite{Nak23} and in the $\SO(16)\times\SO(16)$ case without the monodromy~\cite{Fraiman:2023cpa}. Analogously to the CHL string, an anomaly in the $\Z_2$ symmetry of the $\SO(16)\times\SO(16)$ heterotic string would lead to an inconsistency of this new theory, and our anomaly cancellation result implies a consistency check for this theory on backgrounds where the gauge group is $\Spin(16)$. It would be interesting to study this theory on more general backgrounds.

We have also identified a plethora of non-trivial bordism classes on these theories. It is a natural direction to explore the nature and physics of the bordism defects associated to these branes \cite{Dierigl:2022reg,Kaidi:2023tqo,Debray:2023yrs}, a task we will not pursue in this paper. Furthermore, representatives of the bordism classes we encountered provide natural examples of interesting compactification manifolds for these non-supersymmetric strings to various dimensions; studying these, finding out whether moduli are stabilized (including SUSY-breaking stringy corrections to the potential) etc.\ is another important open direction to study.

This paper is organized as follows: In Section \ref{sec:local} we provide a lightning review of modern methods to study anomalies and how these cancel via the Green-Schwarz mechanism, as well as a detailed description of how this happens for each of the three tachyon-free, non-supersymmetric string theories. We believe this is the first time these important results are collected together in a single reference, and with a unified notation. Section \ref{sec:global_anomalies} contains our main result -- the calculation of bordism groups for these theories using the Adams spectral sequence -- together with a discussion of the natural cohomology theory for the Bianchi identity to take values in. We also study in detail the relationship of higher-dimensional anomaly cancellation to the worldvolume theory of magnetic NS branes. Section \ref{sec:swapping} extends the anomaly calculation to the $SO(16)\times SO(16)$ string including the (gauged) automorphism swapping the two $\Z_2$ factors. This multiplies the number of interesting bordism classes, but anomalies still cancel. Finally, Section \ref{conclus} presents a table with our results, conclusions, and potential further directions, including a few comments on how these anomalies might be studied from a worldsheet point of view, in the line of \cite{Tachikawa:2021mby,Tachikawa:2021mvw}. 

\section{Local anomalies and the Green-Schwarz mechanism}\label{sec:local}

The word ``anomaly'' describes the breaking of a classical symmetry by quantum effects. In a Lagrangian theory, anomalies correspond to a lack of invariance of the path integral under a symmetry transformation. They can arise for both global and gauge symmetries in field theories. Anomalies in global symmetries only point to the fact that the symmetry cannot be gauged; they can lead to anomaly matching conditions that heavily constrain the RG-flow and strong coupling dynamics of the theory \cite{tHooft:1979rat}. In contrast, anomalies in gauge theories point to true inconsistencies: a gauge symmetry is by definition a redundancy of the theory and as such can never be broken. In this paper, we will only consider anomalies in gauge symmetries. 

The anomalies under consideration arise in field theories when they are coupled to gauge fields and dynamical gravity. They then correspond to a lack of invariance of the path integral under a gauge transformation/diffeomorphism (for a review, see \cite{Alvarez-Gaume:1984zlq}). When this transformation can be continuously connected to the identity, we speak of local anomalies. Their cancellation heavily constrains a theory; for example the gauge group of $\mathcal{N}=1$ supergravities in ten dimensions is constrained by anomaly cancellation to be either one of four gauge groups: (a quotient of) $\Spin(32)$, $E_8 \times E_8$, $U(1)^{496}$ or $U(1) ^{248} \times E_8$. The last two can be ruled out as low-energy EFTs of a consistent theory of quantum gravity by demanding the consistency of the worldvolume theory of brane probes in \cite{Kim:2019vuc} (see also \cite{Angelantonj:2020pyr} for developments in the context of orientifold models) and using more general arguments in \cite{Adams:2010zy}. 

When the anomalous symmetry transformation cannot be continuously connected to the identity, then we speak of global anomalies (not to be confused with anomalies in global symmetries!). Global anomalies and their cancellation will be at the heart of this paper. It only makes sense to study them once local anomalies cancel. We therefore review local anomaly cancellation in the remainder of this section, before discussing global anomalies in the next ones.

The most direct way to study local anomalies is to compute certain one-loop Feynman diagrams involving external gauge bosons and/or gravitons, and chiral fermions in the internal legs; for ten-dimensional theories, the relevant diagram has 6 external legs \cite{Green:2012pqa}. 

There is, however, a much more concise way of studying such anomalies, through what is called an anomaly polynomial \cite{Alvarez-Gaume:1984zlq}. This is a certain formal polynomial in the gauge-invariant quantities $\text{tr} F^m$ and $\text{tr} R^m$, which are certain contractions of Riemann and gauge field strength tensors that do not involve the metric. If the theory we wish to study lives in $d$ spacetime dimensions, the anomaly polynomial is of degree $(d+2)$. Although the anomaly polynomial is often discussed in the physics literature directly in terms of $\text{tr} F^m$ and $\text{tr} R^m$, we find it more natural and convenient to write it down in terms of (the free part of) Chern and Pontryagin characteristic classes, more common in the mathematical literature. These can be written as linear combinations of $\text{tr} F^m$ and $\text{tr} R^m$ via Chern-Weil theory as follows. The $i$-th Chern class $c_{\mathbf{r},i}$ is associated to a complex vector bundle, in some representation $\mathbf{r}$ of the gauge group. Via Chern-Weil theory, they are represented in cohomology by the following characteristic polynomial of the field strength (here, $t$ is just a dummy variable)\footnote{Representing characteristic classes with differential forms misses any torsional components of integral cohomology, which is the more natural domain of characteristic classes. This subtlety will play an important role when discussing certain global anomalies in later parts of this paper, but it is immaterial in the present discussion.}:
\begin{equation}\label{eq:cherndef}
    \sum_i c_{\mathbf{r},i}\, t^i = \text{det}\left( \frac{i F}{2 \pi} t + 1\right),\end{equation}
    or, expanding the determinant,
    \begin{equation}
    \sum_i c_{\mathbf{r},i}\, t^i = 1 + \frac{i\, \text{tr}_{\mathbf{r}}(F)}{2 \pi}t  + \frac{\text{tr}_{\mathbf{r}}(F^2) - \text{tr}_{\mathbf{r}} (F)^2}{8 \pi^2}t^2 + \cdots
\end{equation}
The traces are over the gauge indices and as such the $i$-th Chern class $c_{\mathbf{r},i}$ is a $2i$-form. The Pontryagin classes are characteristic classes associated to a real vector bundle, which we will always take to be the spacetime tangent bundle. One way to define them is in terms of the Chern classes of the complexification of the vector bundle. The total Pontryagin class is the sum of the Pontryagin classes and its first few terms are as follows:
\begin{align} 
    p & = 1 + p_1 + p_2+ \cdots \\  p & = 1 - \frac{\text{tr} (R^2)}{8 \pi^2 } + \frac{\text{tr} (R^2)^2 - 2\, \text{tr}(R^4) }{128 \pi^4}\cdots 
\label{baddef}\end{align}
Similarly to the case of the Chern classes, the Pontryagin class $p_i$ is a $4i$-form. The reason we prefer these characteristic classes over the trace notation $\text{tr} F^m$ and $\text{tr} R^m$ is that, as will be clear later, the anomaly polynomial is a sum of Atiyah-Singer indices, and these indices are written in terms of these classes in the mathematical literature.

The precise relationship between local anomalies and the anomaly polynomial is as follows. The anomalous variation of the quantum effective action $\delta_\Lambda \Gamma$ can be related to the $(d+2)$-dimensional anomaly polynomial through what is called the Wess-Zumino descent procedure, which we briefly outline here. Since the characteristic classes are (locally) exact, the anomaly polynomial itself is also (locally) exact. We can therefore locally write it as \begin{equation}\label{eq:Polydef} P_{d+2}= d I_{d+1}\end{equation} where $I_{d+1}$ is called the (Lagrangian density of the) anomaly theory and locally satisfies the descent property $\delta_\Lambda I_{d+1} = dI_d$. It is related to an anomalous variation of the effective action $\delta_ \Lambda \Gamma$, which is local, by extending the spacetime manifold $X_d$ into a $(d+1)$-dimensional manifold $Y_{d+1}$ whose boundary is $X_d$. One finds
\begin{equation}\delta_ \Lambda \Gamma = \int_{X_d= \partial Y_{d+1}} I_d= \int_{Y_{d+1} } d I_{d} = \delta_ \Lambda \left[\int_{Y_{d+1} } I_{d+1}\right]\equiv \delta_\Lambda  \left[\alpha(Y_{d+1})\right],\end{equation}
where the anomaly $I_d$ is only defined up to a closed form and $\alpha(Y_{d+1})$ is called the anomaly theory. The anomaly corresponding to a given gauge transformation is then computed as the integral of $I_d$ over the spacetime manifold $X_d$. It follows that local anomalies vanish if and only if the anomaly polynomial vanishes.

The anomaly polynomial can be written as a sum of contributions from all of the fields in the theory. Each contribution is given by an index density in $(d+2)$-dimensions, whose integral over a compact manifold (with suitable structure, \emph{e.g.} spin for fermions) gives the index of the corresponding Dirac operator via the Atiyah-Singer index theorem. We now list some of these contributions that will be relevant in what follows. The index density associated to a left-handed Weyl fermion in the representation $\mathbf{r}$ of a gauge group with field strength $F$ is: 
\begin{align}\label{eq:indexfermionq}
    \mathcal{I}_{1/2}= \Big[ \hat A(R) \,\text{tr}_{\mathbf{r}}\, e^{i F/2\pi} \Big]_{d+2}\,,
\end{align} where the term $\text{tr}_{\mathbf{r}}\, e^{i F/2\pi}$ is sometimes referred to as the Chern character. The notation $[\cdots]_{d+2}$ means that one should select the $(d+2)$-form part of the enclosed expression. $\hat A(R)$ is called the \textit{A-roof polynomial}, and it can be expanded as: 
\begin{align}
    \hat A (R) = 1 - \frac{p_1}{24} + \frac{(7 p_1^2 - 4 p_2)}{5760} + \frac{-31 p_1^3 + 44 p_1 p_2 - 16 p_3}{967680} \cdots
\end{align}
Note that expression \eqref{eq:indexfermionq} can be easily applied to a fermion that is a singlet under the gauge group, in which case the Chern character reduces to $1$. The contribution to the anomaly polynomial corresponding to a left-handed fermion singlet is therefore simply:
\begin{align}\label{eq:indexfermions}
    \mathcal{I}_{\text{Dirac}}= \Big[ \hat A(R) \Big]_{d+2}.
\end{align} The index density associated to a left-handed Weyl gravitino in $(d+2)$ dimensions is:
\begin{align}\label{eq:indexgravitino}
     \mathcal{I}_{3/2}= \Big[ \hat A(R) \left(\text{tr}\, e^{i R/2\pi}-1\right)\text{tr}_{\mathbf{r}}\, e^{i F/2\pi} \Big]_{d+2}\,.
\end{align} Finally, a self-dual tensor gives a contribution: 
\begin{equation}\label{eq:indexSD}
    \mathcal{I}_{SD}= \Big[-\frac{1}{8} L(R)\Big]_{d+2} ,
\end{equation}
where $L(R)$ is called the \textit{L-polynomial} or \textit{Hirzebruch genus}, and it can be expanded as follows: 
\begin{align}
    L(R) = 1+ \frac{p_1}{3} + \frac{-p_1^2 +7 p_2}{45} + \frac{2 p_1^3 - 13 p_1 p_2 + 62 p_3}{945}+\cdots
\end{align}

Armed with these index densities, we can now review anomalies in ten-dimensional $\mathcal{N}=1$ supergravities. With this supersymmetry, the only multiplets are the gravity and vector multiplets \cite{Green:2012pqa}. There are pure gravitational anomalies coming from the chiral gravitini in the gravity multiplet as well as Majorana-Weyl spinors in both gravity and vector multiplets. There can also be gauge anomalies involving the fermions in the vector multiplet, which transform in the adjoint representation of the gauge group (they are gaugini). The anomaly polynomial of these theories is computed using \eqref{eq:indexfermionq}-\eqref{eq:indexSD} summing over all chiral fields (dilatino, gravitino, and gaugini), and reduces to the following expressions for the $\Spin(32)/\Z_2$ and the $E_8 \times E_8$ theories respectively: 
\begin{equation}\label{eq:anompolyso32}P_{12}^{\Spin(32)/\Z_2}= - \, \frac{p_1 + c_{\mathbf{32},2}}{2} \times \frac{1}{192} \left( 16 \, c_{\mathbf{32},2}^2 - 32 \, c_{\mathbf{32},4} + 4 \, c_{\mathbf{32},2}\, p_1 + 3 \, p_1^2 - 4 \, p_2 \right) \, ,\end{equation}
\begin{align}\begin{split}\label{eq:anompolye8e8} &P_{12}^{E_8\times E_8} = - \, \frac{p_1 + c_{\mathbf{16},2}^{(1)} + c_{\mathbf{16},2}^{(2)}}{2}\\ &\times \frac{1}{192} \left( 8 \, (c_{\mathbf{16},2}^{(1)})^2 + 8 \, (c_{\mathbf{16},2}^{(2)})^2 + 4 \, (c_{\mathbf{16},2}^{(1)} + c_{\mathbf{16},2}^{(2)}) \, p_1 - 8 \, c_{\mathbf{16},2}^{(1)} \, c_{\mathbf{16},2}^{(2)} + 3 \, p_1^2 - 4 \, p_2 \right) \,.
\end{split}\end{align}
Here, the Chern classes are expressed in the vector representations of $Spin(32)$ and the $SO(16)$ subgroups of each $E_8$ respectively, and the $(1)$ and $(2)$ indices differentiate between each of the two $E_8$ gauge groups. These expressions agree with the ones in \cite{Green:2012pqa} once \eqref{eq:cherndef}-\eqref{baddef} are used to put Chern and Pontryagin classes back into traces. The anomaly polynomial does not vanish, and at first sight this would mean that the theories are inconsistent. However, both anomaly polynomials factorize in the following schematic form:
\begin{equation}\label{eq:factorizedanom}P_{12}= X_4 X_8 \, ,\end{equation}
where the $X_4$ is a four-form and $X_8$ is an eight-form. This factorization property is the key to cancelling the anomaly, through what is known as the Green-Schwarz mechanism \cite{Green:1984sg,Green:2012pqa}. There is another field in these theories that can contribute to the aforementioned diagrams: the massless Kalb-Ramond B-field. Consistency of the 10-dimensional supergravity requires that the B-field not be invariant under gauge and gravitational interactions, and in fact it must satisfy an identity of the form
\begin{equation}
dH_ 3 = X_4 \, ,\label{bianch0}
\end{equation}
where $H_3\equiv dB_2 - \omega_\text{CS}$ is the gauge-invariant curvature of the B-field built from the Chern-Simons forms of gauge and spin connections that appears in the (super)gravity couplings. $X_4=d\omega_\text{CS}$ is a linear combination of characteristic classes of gauge and gravitational bundles. There is more to the Bianchi identity at the global level, a subject which we will discuss in Section \ref{subsec:torsi}. 

In this case, the anomaly corresponding to a factorized anomaly polynomial \eqref{eq:factorizedanom} can be cancelled by introducing the following term in the action: 
\begin{equation}\label{eq:GSterm10d} -\int B_2 \wedge X_8 \,.\end{equation} This term in the action is actually generated by string perturbation theory, as shown in \cite{Lerche:1987sg} in the heterotic case, and it contributes a term $ -dB_2 \wedge  X_8$ to the Lagrangian density of the anomaly theory. Using the Bianchi identity and \eqref{eq:Polydef}, we see that this adds a term $- X_4   X_8$ to the anomaly polynomial, exactly canceling the term in \eqref{eq:factorizedanom} so that the total local anomaly vanishes. This is the basic principle of the Green-Schwarz mechanism where the anomaly is cancelled by introducing an extra term in the action.

The main focus of this paper is however the three known\footnote{It was recently proven that there are no other examples in the heterotic context \cite{BoyleSmith:2023xkd}.} ten-dimensional non-tachyonic, non-supersymmetric string theories, which also 
 feature a B-field and a Green-Schwarz mechanism to cancel local anomalies. The diagram in Figure \ref{hexagon} sums up how these non-supersymmetric theories are related to eachother and to the supersymmetric ones via gaugings of various worldsheet symmetries. We will now briefly describe the matter content of these theories as well as their anomalies. 
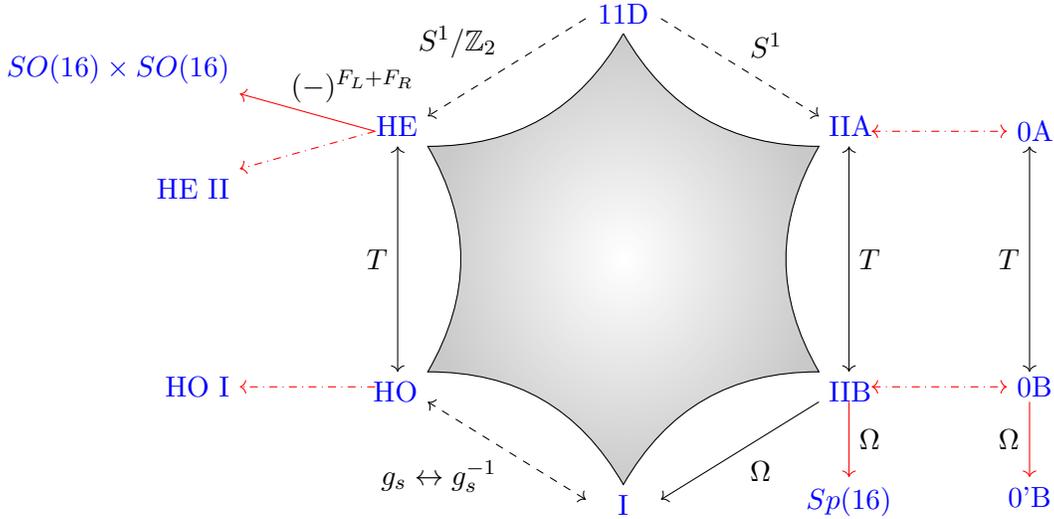
\begin{figure}
    \centering
   \begin{tikzpicture}
  \coordinate (a) at(0,3);
  \coordinate (b) at(2.6,1.5);
  \coordinate (c) at (2.6,-1.5);
  \coordinate (d) at (0,-3);
  \coordinate (e) at (-2.6,-1.5);
  \coordinate (f) at (-2.6,1.5);
  \filldraw [inner color=white, outer color=gray!30, opacity=1](a)node[above]{\color{blue}{11D}} to[bend right](b)node[above right]{\color{blue}{IIA}} to[bend right](c)node[below right]{\color{blue}{IIB}} to[bend right] (d)node[below]{\color{blue}{I}} to[bend right] (e)node[below left]{\color{blue}{HO}} to[bend right] (f)node[above left]{\color{blue}{HE}} to[bend right] (a);
  \draw[<->](3,1.5)--(3,0)node[right]{$T$}--(3,-1.5);
  \draw[<->](-3,1.5)--(-3,0)node[left]{$T$}--(-3,-1.5);
  \draw[dashed,->](.5,3.2)--(1.55,2.55)node[above right]{\color{black}{$S^1$}}--(2.6,1.9);
  \draw[dashed, ->](-.5,3.2)--(-1.55,2.55)node[above left]{\color{black}{$S^1/\Z_2$}}--(-2.6,1.9);
  \draw[->](2.6,-1.9)--(1.55,-2.55)node[below right]{$\Omega$}--(.5,-3.2);
  \draw[dashed, <->](-2.6,-1.9)--(-1.55,-2.55)node[below left]{\color{black}{$g_s\leftrightarrow g_s^{-1}$}}--(-.5,-3.2);
 
  \draw[red,dashdotted, <->](3.3,1.7)--(5.1,1.7)node[right]{\color{blue}{0A}};
  \draw[red,dashdotted, <->](3.3,-1.7)--(5.1,-1.7)node[right]{\color{blue}{0B}};
  \draw [<->] (5.4,1.5)--(5.4,0)node[left]{$T$}--(5.4,-1.5);
  \draw [->, red](5.4,-1.9)--(5.4,-2.4)node[left]{\color{black}{$\Omega$}}--(5.4,-2.9)node[below]{\color{blue}{0'B}};
  \draw [->, red](3,-1.9)--(3,-2.4)node[right]{\color{black}{$\Omega$}}--(3,-2.9)node[below]{\color{blue}{$Sp(16)$}};
  \draw[red,dashdotted, ->](-3.3,-1.7)--(-5.1,-1.7)node[left]{\color{blue}{HO I}};
  \draw[red,dashdotted, ->](-3.3,1.7)--(-5.1,1.2)node[below left]{\color{blue}{HE II}};
  \draw[red, ->](-3.3,1.7)--(-5.1,2.2)node[above left]{\color{blue}{$SO(16)\times SO(16)$}};
  \node at (-3.6,1.95)[above]{$(-)^{F_L+F_R}$};
\end{tikzpicture}
    \caption{A diagram showing how the three tachyon-free non-supersymmetric string theories relate to the supersymmetric ones and M-theory, via various worldsheet orbifolds. For instance, the $SO(16)^2$ is obtained from the $E_8^2$ heterotic by orbifolding spacetime and gauge group fermion number, $F=F_L+F_R$. We also list some tachyonic examples via red dot-dashed arrows, although we are not exhaustive. }
    \label{hexagon}
\end{figure}

\begin{itemize}
    \item \emph{The Sugimoto model}

 The Sugimoto string \cite{Sugimoto:1999tx} can be thought of as the non-supersymmetric sibling of the supersymmetric type I $\Spin(32)/\Z_2$ string. The main departure for us is that the gauge algebra is $\mathfrak{sp}(16)$ instead, and thus the Chern classes are taken in the fundamental representation of this group\footnote{Note that we use $\mathfrak{sp}(16)$ and  $\Sp(16)$ instead of the notation $ \mathit{USp}(32)$ that is often employed in the literature.}. This distinction arises from the different kind of orientifold projection of type IIB, which introduces anti-D9 branes and an O9 plane with positive Ramond-Ramond charge \emph{and} tension. The sign change in the reflection coefficients for unoriented strings scattering off the O9 is such that the Chan-Paton degeneracies reconstruct representations of the symplectic group $\Sp(16)$. 

 As in the type I case, the closed-string sector arranges into an $\mathcal{N} = 1$ supergravity multiplet, while the chiral fermions from the open-string sector arrange into the antisymmetric rank-two representation of the gauge group, leading to the same anomaly polynomial formally. This representation is however reducible and contains a singlet; this is nothing but the Goldstino that accompanies the breaking of supersymmetry \cite{Antoniadis:1999xk,Angelantonj:1999jh,Aldazabal:1999jr,Angelantonj:1999ms}. The low-energy interactions comply with the expected Volkov-Akulov structure of nonlinear supersymmetry \cite{Dudas:2000nv,Pradisi:2001yv,Kitazawa:2018zys}, although there is no tunable parameter that recovers a linear realization. All in all, since the anomaly polynomial is formally identical to the type I case, it factorizes as follows:
\begin{equation}\label{eq:sugimoto_ch_poly}
   P_{12}^{\text{Sugimoto}}= - \, \frac{p_1 + c_{\mathbf{32},2}}{2} \times \frac{1}{192} \left(16 \, c_{\mathbf{32},2}^2 - 32 \, c_{\mathbf{32},4} + 4 \, c_{\mathbf{32},2} \, p_1 + 3 \, p_1^2 - 4 \, p_2 \right) .
\end{equation}

    \item \emph{The Sagnotti model}

The type 0'B string~\cite{Sagnotti:1995ga, Sagnotti:1996qj} of Sagnotti is built from an orientifold projection of the tachyonic type 0B string, where the unique tachyon-free choice involves an O9 plane with zero tension. The resulting gauge group is $U(32)$, Ramond-Ramond $p$-form potentials with $p = 0 \, , 2 \, , 4$ (the latter having a self-dual curvature) survive the projection and they get anomalous Bianchi identities for the gauge-invariant curvatures,
\begin{eqaed}\label{eq:gen_Bianchi}
    dH_1 & = X_2 \, , \\
    dH_3 & = X_4 \, , \\
    dH_5 & = X_6
\end{eqaed}
where $X_2 = c_{\mathbf{32},1}$, $X_4$ is formally identical to the one in Sugimoto and type I strings, and $X_6$ is a polynomial in $p_1, \, c_{\mathbf{32},1},\, c_{\mathbf{32},2}$ and $c_{\mathbf{32},3}$. The Bianchi identity for $X_2$ tells us that the low-energy gauge group reduces to $SU(32)$, since $c_{\mathbf{32},1}$ is set to zero. Physically, the anomalous Bianchi identity for the RR axion induces the kinetic term $\abs{dC_0 + A}^2$, with $A$ the gauge field of the diagonal $\mathfrak{u}(1)$. This is just the St{\"u}ckelberg mass term for $A$.
All these RR fields with anomalous Bianchi identities play a crucial role in the cancellation of local anomalies via a more complicated Green-Schwarz mechanism involving a decomposition of the anomaly polynomial~\cite{Sagnotti:1996qj} of the form
\begin{equation}\label{eq:generalized_GS_mechanism}
   P_{12}^{\text{0'B}} = X_2 \, X_{10} + X_4 \, X_8 + X_6 \, X_6 \, .
\end{equation}

As we will see, setting $X_6$ to zero implies that $c_{\mathbf{32},3}$ is also trivial, and so we shall impose this condition when studying global anomalies of this theory.

    \item \emph{The heterotic model}

The case of the $SO(16)\times SO(16)$ string \cite{Alvarez-Gaume:1986ghj, Dixon:1986iz} is slightly different. Along with its two supersymmetric counterparts, it is the unique ten-dimensional heterotic model that is devoid of tachyons. It is built from a projection of either of the two heterotic models, most directly the $E_8 \times E_8$ one under the projector built from a combination of spacetime fermion number and an $E_8$ lattice symmetry. As a result, it does not have any chiral fields that are uncharged under the gauge symmetry, and in particular it does not have a gravitino. Its anomaly polynomial was derived in~\cite{Alvarez-Gaume:1986ghj, Dixon:1986iz} and factorizes as: 
\begin{equation}\begin{split}\label{eq:H16_ch_poly}
  P_{12}^{SO(16)^2}= &  - \, \frac{p_1 + c_{\mathbf{16},2}^{(1)} + c_{\mathbf{16},2}^{(2)}}{2}\\ & \times \frac{1}{24} \left( (c_{\mathbf{16},2}^{(1)})^2 + (c_{\mathbf{16},2}^{(2)})^2 + c_{\mathbf{16},2}^{(1)} \, c_{\mathbf{16},2}^{(2)} - 4 \, c_{\mathbf{16},4}^{(1)} - 4 \, c_{\mathbf{16},4}^{(2)} \right) \,,
\end{split}\end{equation}
where the $(1)$ and $(2)$ indices differentiate between each of the two $SO(16)$ gauge groups. The Green-Schwarz mechanism is carried by the Kalb-Ramond field, which survives the projection as befits a heterotic model \cite{Schellekens:1986xh}. \end{itemize}

All in all, local anomalies vanish for all three non-supersymmetric string theories, by the Green-Schwarz mechanism (or a more complicated version of it). This was already known in the literature, but leaves open the possibility for the presence of global anomalies. Global anomalies are those that arise in gauge/diffeomorphism transformations that cannot be continuously connected to the identity. These anomalies are not detected by the anomaly polynomial.  In the following section we detail how one can study these anomalies and we evaluate them for the case of the three non-supersymmetric tachyon-free string theories. 

\section{Global anomalies and bordism groups}\label{sec:global_anomalies}
In the previous Section, we have summarized prior results in the literature regarding  anomaly cancellation of ten-dimensional non-supersymmetric string theories via the Green-Schwarz mechanism. Importantly, the Green-Schwarz mechanism only guarantees cancellation of \emph{local} anomalies -- it guarantees that the (super)gravity path integral is gauge invariant as long as we only consider gauge transformations infinitesimally close to the identity. More generally, one also need discuss \emph{global} anomalies, namely anomalies in gauge transformations that cannot be continuously deformed to the identity. The archetypal example of such a global anomaly is Witten's $SU(2)$ anomaly~\cite{Witten:1982fp}. If one includes topology-changing transitions, one has even more general anomalies (dubbed Dai-Freed anomalies in \cite{Garcia-Etxebarria:2018ajm}), involving a combination of gauge transformations and spacetime topology change. In this paper, we will take the point of view that such anomalies should cancel in a consistent quantum theory of gravity, where spacetime topology is supposed to fluctuate.

The framework of anomaly theories introduced briefly in the previous section \eqref{eq:Polydef} can also be used to study global anomalies of Lagrangian theories such as the ones we are interested in. Given a $d$-dimensional quantum field theory, an anomaly on a manifold $X_d$ (possibly decorated with gauge field, spin structure, etc.) means that the partition function $Z(X_d)$ is not invariant under gauge transformations (or diffeomorphisms, for the case of a gravitational anomaly). In a modern understanding (see~\cite{Witten:2015aba,Yonekura:2016wuc,Garcia-Etxebarria:2018ajm,Davighi:2020vcm, Fre23} for detailed reviews, and also~\cite{Davighi:2023luh,Tachikawa:2023lwf} for a discussion in the context of the 6d Green-Schwarz mechanism), the anomaly can be captured by an invertible $(d+1)$-dimensional field theory $\alpha$ with the property that, when evaluated on a manifold with boundary $Y_{d+1}$ with $\partial Y_{d+1}=X_d$, the product
\begin{equation} Z(X_d) \cdot e^{- 2 \pi i \alpha(Y_{d+1})} \label{e332}\end{equation}
is invariant under gauge transformations. The $d$-dimensional QFT arises as a boundary mode of the $(d+1)$-dimensional invertible field theory $\alpha$, and the anomaly is re-encoded in the fact that \eq{e332} is not the partition function of a $d$-dimensional quantum field theory -- its value depends in general on $Y_{d+1}$ and the particular way on which the fields on $X_d$ are extended to $Y_{d+1}$.

In general, it may be very difficult to determine $\alpha$. However, in weakly coupled Lagrangian theories, we have a prescription to associate an anomaly theory to each of the chiral degrees of freedom involved. For instance, the anomaly theory for a Weyl fermion in $d$-dimensions ($d$ even) is given by the so-called eta invariant of a $(d+1)$-dimensional Dirac operator with the same quantum numbers as the fermion we started with~\cite{Dai:1994kq,Yonekura:2016wuc},
\begin{equation} e^{2\pi i{\alpha}_{\text{fermion}}(Y_{d+1})}=e^{2\pi i\eta_{d+1}(Y_{d+1})}.\end{equation}
If one has several fermions, the total anomaly theory is simply the product of these (so that the $\eta$ invariants add up). There are other topological couplings that can also contribute to the anomaly theory, as we will see below.

Two different open manifolds $Y_{d+1}$ and $Y'_{d+1}$, both having $X_d$ as a boundary, will yield values for the partition function \eq{e332} differing by a factor
\begin{equation} e^{2\pi i \alpha(Y_{d+1})}e^{-2\pi i \alpha(Y'_{d+1})}= e^{2\pi i \alpha(Y_{d+1}\cup \overline{Y'}_{d+1})}.\end{equation}
The manifold $Y_{d+1}\cup \overline{Y'}_{d+1}$ is just a general closed $(d+1)$-dimensional manifold. In an anomaly free-theory, the partition function in \eqref{e332} should not depend on the choice of extension; therefore, in this picture, anomaly cancellation is simply the statement that the anomaly theory $\alpha(\tilde{Y}_{d+1})$ be trivial when evaluated on a closed manifold $\tilde{Y}_{d+1}$. 

The particular case in which $\tilde{Y}_{d+1}$ itself is a boundary, $\tilde{Y}_{d+1}=\partial Z_{d+2}$,  corresponds to local anomalies, which allows us to connect the discussion to the preceding Section. The $\eta$ invariants introduced above, that give the anomalies for chiral fermions, can in this case be evaluated by means of the APS index theorem~\cite{atiyah_patodi_singer_1975},
\begin{equation} \label{eq:aps}\eta(\tilde{Y}_{d+1})=\text{Index}- \int_{Z_{d+2}} P_{d+2},\end{equation}
where $P_{d+2}$ is the anomaly polynomial of the previous subsection, and the ``Index'' is an integer. We thus recover the usual, perturbative, anomaly cancellation condition in terms of the anomaly theory. In theories where anomalies are cancelled via the Green-Schwarz mechanism, another ingredient is necessary. The ten-dimensional action has an extra Green-Schwarz term \eqref{eq:GSterm10d}, which is the boundary mode of an 11d invertible field theory
\begin{equation}\label{eq:GSanomtheory}\alpha_{\text{GS}}(Y_{11})= \int_{Y_{11}} H \wedge X_8.\end{equation}
The total anomaly theory is therefore the sum of the fermion anomaly and $\alpha_{\text{GS}}(Y_{11})$. On a manifold which is itself a boundary, $\tilde{Y}_{11} = \partial Z_{12}$,
\begin{equation} \int_{\tilde{Y}_{11}} H \wedge X_8=\int_{Z_{12}} dH \wedge X_8= \int_{Z_{12}}  X_4 \wedge  X_8,\label{e34}\end{equation}
where in the last equality we used the constraint that we are restricting to twisted string manifolds satisfying the (anomalous) Bianchi identity $d H  = X_4$. Taking this last contribution into account, we see that the local anomaly coming from the GS term can cancel that of the fermions, provided that the anomaly polynomial factorizes as discussed in Section \ref{sec:local}. 

In the rest of this paper, we will assume that local anomalies cancel, and ask what is the value of the total anomaly theory,
\begin{equation} e^{2\pi i \alpha_{\text{tot}}}=e^{2\pi i \alpha_{\text{fermions}}}e^{2\pi i \alpha_{\text{GS}}}\label{totalan}\end{equation}
when evaluated on 11-dimensional closed manifolds which are \emph{not} boundaries. This task seems daunting at first, since, depending on the collection of background fields, there can be infinitely many such manifolds. Fortunately, one can prove\footnote{The proof is a straightforward application of the APS index theorem \eqref{eq:aps}, see \cite{Yonekura:2016wuc}.}~\cite{Yonekura:2016wuc} that the partition function of the anomaly theory $\alpha_{\text{tot}}$ (mod 1) is a \emph{bordism invariant},\footnote{There is also a more theoretical and more general proof that the partition function of the anomaly theory is a bordism invariant, due to Freed-Hopkins-Teleman~\cite{FHT10} and Freed-Hopkins~\cite{FH16}; they show that up to a deformation, which is irrelevant for anomaly calculations, the partition function of \emph{any} reflection-positive invertible field theory is a bordism invariant.}
\begin{equation} e^{2\pi i \alpha_{\text{tot}}(Y^{(1)}_{11})} = e^{2\pi i \alpha_{\text{tot}}(Y^{(2)}_{11})}\quad\text{if}\quad Y^{(1)}_{11}\cup \overline{Y^{(2)}_{11}}=\partial Z_{d+2}.\end{equation}
This reduces the problem significantly: since $\alpha_{\text{tot}}$ (mod 1) is a bordism invariant, one need only evaluate it on a single representative per bordism class. Furthermore, these classes form an abelian group, the bordism group (of manifolds suitably decorated with a twisted string structure and gauge bundle). Bordism groups have appeared prominently in the field theory and quantum gravity literature, and there are many techniques available for their computation (see~\cite{Debray:2023yrs} for a detailed introduction). Thus, to compute these anomalies, we will just compute the relevant bordism groups and evaluate the anomaly theory on generators. Notice that if it happens that the relevant bordism group $\Omega_{11}$ is $0$, there are no global anomalies to check! That this happens was in fact shown by Witten in~\cite{Witten:1985mj,Witten:1985xe} for the $E_8\times E_8$ string when one does not take into account the $\Z_2$ symmetry switching the two $E_8$ gauge fields.\footnote{If one does want to take this $\Z_2$ symmetry into account, for example to study the CHL string, the relevant $\Omega_{11}$ is nonzero~\cite[Theorem 2.62]{Deb23}, and it was not known whether the global anomaly cancels. We will show that it does cancel in this paper, in Section \ref{sec:swapping}.} See~\cite{Freed:2000ta} for an analysis of type I string.

More recently,~\cite{Tachikawa:2021mvw,Tachikawa:2021mby} used the Stolz-Teichner conjecture to analyze global anomalies in supersymmetric, heterotic string theory even in stringy backgrounds, lacking a geometric description. In this paper we content ourselves with the target space treatment described above, which may miss anomalies of non-geometric backgrounds. In the following, we present the calculation and results for the three ten-dimensional non-supersymmetric string theories described in Section \ref{sec:local}. But before that, we will describe and justify more carefully the precise structure that will be assumed in our bordism calculations.  

\subsection{Bianchi identities and twisted string structures} \label{subsec:torsi}
As described in the previous Subsection, the computation of global anomalies can be organized in terms of a bordism calculation and an anomaly theory, which is just a homomorphism from the bordism group to $U(1)$. The precise bordism group to be used (i.e.\ the particular structure that our manifolds are required to have) depends on the theory we are interested in. For instance, all heterotic string theories under consideration include fermions, so we will consider only manifolds (and bordism between them) carrying a spin structure; the anomaly theory is related to the $\eta$ invariant for a certain Dirac operator on this manifold. This means that the second Stiefel-Whitney class of the allowed manifolds where the anomaly theory is to be evaluated will vanish,
\begin{equation}w_2=0.\end{equation}
In heterotic string theories, also the Bianchi identity \eq{bianch0} needs to be taken into account. Equation \eq{e34}  illustrates that cancellation of perturbative anomalies requires us to assume $dH=X_4$ even off-shell\footnote{If we insist on keeping the $B$-field as a background; see the discussion at the end of this Section.}. Therefore, we will restrict our bordism groups to consist of 11-dimensional manifolds in which \eq{bianch0} is satisfied. In particular, we will set
\begin{equation} \int_{M_4} X_4=0\end{equation}
for any closed 4-manifold $M_4$. The precise expression of $X_4$ in terms of characteristic classes depends on the particular theory under study. The particular case
\begin{equation} X_4=\frac{p_1}{2}\end{equation}
has been studied in the mathematical literature, and receives the name of a \emph{string} structure. The $X_4$'s that appear in heterotic string theories are always of the form
\begin{equation} X_4=a\, \frac{p_1}{2}(\text{Tangent bundle})+b\, c_2(\text{Gauge bundle}),\quad a,b\in\mathbb{Z},\label{err3}\end{equation}
and we will refer to the data of a solution to this equation for chosen $a$ and $b$ as a \emph{twisted} string structure. This notion appeared in the mathematical literature in \cite[Definition 8.4]{Wan08}.

The bordism groups related to the three non-supersymmetric string theories we are going to consider are
\begin{equation}\Omega_{11}^{\text{String}-\mathit{Sp}(16)}, \quad \Omega_{11}^{\text{String}-\SU(32)\ang{c_3}},\quad \Omega_{11}^{\text{String}-\mathit{Spin}(16)^2},\end{equation}
for the Sugimoto, Sagnotti, and $\SO(16)\times\SO(16)$ heterotic theories, respectively; these are the bordism groups of twisted string manifolds where the particular choice of twisted string structure is spelled out by the Green-Schwarz mechanisms for these theories as we discussed in Section~\ref{sec:local}.
  
\begin{rem}
Before presenting the results for the bordism groups, we must discuss an important subtlety, which affects the bordism calculation. Up to this point in this paper, we have been cavalier when writing down characteristic classes such as ``$p_1$'' or ``$c_2$'', and defined these characteristic classes as closed differential forms (e.g.\ in \eq{baddef}) by way of Chern-Weil theory. However, these differential forms have quantized periods, as is the case for data coming out of any quantum theory, and a proper treatment of the Green-Schwarz mechanism should take this into account. There are two ways to do this.
\begin{enumerate}
    \item The simplest approach is to lift to $\Z$-valued cohomology: the quantized periods are a reminder that the de Rham classes of the Chern-Weil forms of $p_1$, $c_2$, etc., lift canonically to classes in $H^4(BG;\Z)$ for various Lie groups $G$, and on many manifolds $M$, these integer-cohomology lifts of these characteristic classes can be torsion! Thus it is natural to wonder whether the B-field should be an element of $H^3(\text{--};\Z)$ and the Bianchi identity~\eqref{err3} should take place in $H^4(\text{--}; \Z)$. The definitions of string structure and twisted string structure in mathematics assume this lift has taken place.
    \item Alternatively, one could lift to differential cohomology $\check H^4(\text{--};\Z)$, which amounts to observing that it is not just the $\Z$-cohomology lift which is natural, but also the data of the Chern-Weil form; differential cohomology is a toolbox for encoding both of these pieces of data. Indeed, for any compact Lie group $G$ and class $c\in H^*(BG;\Z)$, there is a canonical differential refinement $\check c\in\check H^*(B_\nabla G;\Z)$~\cite{CS85, BNV16}, where $B_\nabla G$ is the classifying stack of $G$-connections.\footnote{If you do not want to think about stacks, this statement is essentially equivalent to the notion that for a principal $G$-bundle $P\to M$ with connection $\Theta$, the differential characteristic class $\check c(P, \Theta)\in \check H^*(M;\Z)$ is natural in $(P, \Theta)$.} Thus we could instead ask: should we begin with a B-field in $\check H^3(\text{--};\Z)$ and ask for the Bianchi identity to take place in $\check H^4(\text{--};\Z)$? This combines the two other formalisms we considered, differential forms and integral cohomology.
\end{enumerate}
The answer in the mathematics literature is often the second option, beginning with Freed~\cite[\S 3]{Freed:2000ta} and continuing in, for example,~\cite{Fre02, FMS07, SSS09, FL10, Bun11, DFM11a, DFM11b, FSS12, SSS12, FSS13, Sch13, FSS14, GS18, MM18, MM19, FSS20, Hsieh:2020jpj, GS21, Deb23}. In particular, \cite{FSS12, SSS12} interpret the data entering into the Green-Schwarz mechanism as specifying a connection for a Lie $2$-group built as an extension of the gauge group by $B\U(1)$, providing an appealing physical interpretation of the lift to differential cohomology.

We are interested in classifying anomalies, and while there is an interesting differential refinement of the story of the bordism classification of anomalies due to Yamashita-Yonekura~\cite{YY23, Yam23a, Yam23b}, the deformation classification of anomalies ultimately can proceed without differential-cohomological information, because it boils down to studying bordism groups. Because of this, we will work with characteristic classes in integral cohomology, noting here that the correct setup of the Green-Schwarz mechanism taking torsion and Chern-Weil forms into account uses differential cohomology, and that for our computations it makes no difference. 

\end{rem}

Note that cancellation of perturbative anomalies around \eq{e34} only requires the free part of $X_4\in\, H^4(\text{--};\mathbb{Z})$ to be trivial in a compact manifold, and poses no obvious restriction on torsion. Reference \cite{Witten:1985mj} studies a particular example suggesting that this should be the case, but does not attempt to make a general argument. To ascertain whether the torsion piece of $X_4$ must also be trivialized or not, consider the physical origin of the Bianchi identity, which is itself a two-dimensional version of the Green-Schwarz mechanism described above (see e.g. \cite{Garcia-Etxebarria:2014txa,Kim:2019vuc}). Consider a worldsheet wrapped on a 2-manifold $\Sigma_2$ of the ambient ten-dimensional spacetime manifold $M_{10}$. The configuration should be invariant under spacetime diffeomorphisms, and gauge transformations, which are manifested as global symmetries of the worldsheet. However, in heterotic or type I theories, the worldvolume degrees of freedom are chiral, and anomalous under these transformations. The anomaly  theory, which we denote $\alpha^{\text{worldsheet}}$, is encoded by a three-dimensional  $\eta$ invariant. Applying the APS index theorem \eq{eq:aps}, we obtain 
\begin{equation} \exp(2\pi i \alpha^{\text{worldsheet}})=\exp \left(2\pi i \int X_4 \right),\label{anomws}\end{equation}
where $X_4$ is a certain differential form built out of characteristic classes, and which is precisely the $X_4$ appearing above (indeed, \eq{anomws} is usually taken to give the definition of $X_4$). As things stand, any configuration with an insertion of a fundamental string worldsheet on $\Sigma_2$ has a gravitational anomaly; however, the worldsheet also has an electric coupling to the B-field,
\begin{equation} \exp\left(2\pi i\int_{\Sigma_2} B_2\right),\end{equation}
whose anomaly theory is simply
\begin{equation}\alpha_B=\int H.\end{equation}
Now, the total worldsheet anomaly is
\begin{equation} \exp\left(2\pi i \alpha^{\text{worldsheet}}_{\text{total}}\right)= \exp\left(2\pi i \alpha^{\text{worldsheet}}\right) \exp \left(2\pi i \alpha_B\right).\end{equation}
 The physical consistency condition is that the total anomaly is trivial
\begin{equation}\exp\left(2\pi i \alpha^{\text{worldsheet}}_{\text{total}}\right)=1,\quad \text{for all $M_3$}.\label{anows}\end{equation}
When $M_3$ is a boundary, anomaly cancellation is achieved, as above, by setting $dH=X_4$, precisely the Bianchi identity described above. However, this is not all there is to \eq{anows}. Assuming that anomalies vanish when $M_3$ is a boundary,  $\exp(2\pi i \alpha^{\text{worldsheet}}_{\text{total}}(M_3))$ is actually only dependent on the integer homology class of $M_3$. In fact, since it is a map that assigns a phase to each 3-cycle in the ambient 10-dimensional manifold $M_{10}$, it can be regarded as an element of $H^3(M_{10};U(1))$, with $U(1)$ coefficients.  Using the long exact sequence in cohomology associated to the short exact sequence of groups $\mathbb{Z}\,\rightarrow\,\mathbb{R}\,\rightarrow U(1)$, we obtain that \cite{Davighi:2020vcm}
\begin{equation}H^3(M_{10};\mathbb{R})\,\rightarrow\, H^3(M_{10};U(1))\,\rightarrow\, H^4(M_{10};\mathbb{Z})\,\rightarrow H^4(M_{10};\mathbb{R}),\label{exe3}\end{equation}
where the third map is taking the free part of the integer cohomology class. In general,  $\exp(2\pi i \alpha^{\text{worldsheet}}_{\text{total}})$ will have pieces both in the image of the first map and in its cokernel. An example where the anomaly theory has a non-trivial piece in the image of the first map of \eq{exe3} can be obtained by compactifying heterotic string theory on a Bieberbach  3- manifold, a fixed-point free quotient of the torus $T^3$.\footnote{We thank Cumrun Vafa for pointing out this example to one of us.} Since $T^3$ is Riemann-flat, a quick analysis would suggest that the Bianchi identity is satisfied automatically with no gauge  bundle or B-field turned on. However, trying to implement this manifold directly in the worldsheet results in a theory which is not level-matched. The problem is that $\exp(2\pi i \alpha^{\text{worldsheet}}_{\text{total}})$ with no gauge bundle turned on is nontrivial for most Bieberbach manifolds, and so the anomaly theory is a nontrivial class in $H^3(M_{10};\mathbb{R})$. Cancelling this anomaly forces either a B-field (discrete torsion) to be turned on, or a non-trivial flat gauge bundle to be present.

The rest of the anomaly theory is in the image of the second map in \eq{exe3}, and can therefore be represented by a certain  torsion integer cohomology class in $H^4(M_{10};\mathbb{Z})$, whose free part vanishes. We will now show that this is in fact the torsion part of $X_4-dH$. Consider a torsion 3-cycle $M_3$ of order $k$, i.e. such that $k M_3$ is the boundary of some 4-manifold $N_4$. Let us see how to compute the anomaly in this case. First, the anomaly theory $\alpha^{\text{worldsheet}}_{\text{total}}$ is a linear combination of $\eta$ invariants, which in this particular case can be re-expressed as linear combinations of gravitational and gauge Chern-Simons numbers as discussed above. 
The Chern-Simons invariant is additive on disconnected sums, and so, we have
\begin{equation}\alpha^{\text{worldsheet}}_{\text{total}}(M_3)=\frac{1}{k}\alpha^{\text{worldsheet}}_{\text{total}}(kM_3).\end{equation}
Next, we can use the fact that $kM_3=\partial N_4$, to write (after exponentiation)\footnote{The anomaly theory is a linear combination of real-valued eta invariants, thus division by $k$ is well-defined and there is no phase ambiguity.} 
\begin{equation}\exp\left(2\pi i \alpha^{\text{worldsheet}}_{\text{total}}(M_3)\right)=\exp\left(\frac{2\pi i}{k}\left(\text{Index}_{N_4}+\int_{N_4}(X_4-dH)\right)\right).\end{equation}
This expression is not obviously independent of the choice of $N_4$, but when $(X_4-dH)$ is pure torsion, it actually is. The reason is that the quantity $\int_{N_4}(X_4-dH)$ may be rewritten as a linking pairing in homology \cite{birman1980seifert}. If we Poincaré dualize $(X_4-dH)_{\text{tor}}$ to a torsion 6-cycle $M_6$, the linking pairing between $M_6$ and $M_3$ is constructed by choosing a boundary $N_4$ for $k M_3$ and computing $\int_{N_4}(X_4-dH)$ modulo $k$. Importantly, the result does not depend on the choice of $N_4$ (see \cite{birman1980seifert} for a review and proof of these facts). 

In short, the full analog of the Bianchi identity is \eq{anows}. Unpacking this condition, we recover that:\begin{itemize}
\item There is the condition on any $3$-manifold $M_3$ that  
\begin{equation} \int_{M_3} H = \int_{M_3} CS^{X_4}_3,\end{equation}
where $CS^{X_4}_3$ is a (local) Chern-Simons form obeying $dCS^{X_4}_3=X_4$. This will force discrete B-fields to be turned on in certain situations, such as on Bieberbach manifolds (these were referred to as ``worldsheet discrete theta angles'' in \cite{Tachikawa:2023lwf}). 
\item As a consequence of the previous point, when $M_3$ is a boundary, we get that the Bianchi identity $dH=X_4$ must hold over the integers.
\end{itemize}

The general analysis we just carried out is somewhat abstract; in the next Subsection, we will verify its correctness by explicitly checking, in a variety of backgrounds, that anomalies in ten dimensions only cancel if the torsional part of the Bianchi identity holds.

Finally, we comment on another possible way in which the anomaly calculation could have been set, avoiding the calculation of string bordism groups altogether, as in \cite{Tachikawa:2023lwf}. Anomalies are always studied with respect to a choice of background fields. The approach we have followed here takes the metric $g$, the gauge field $A$, and the 2-form field $B$ as background fields, and imposes the Bianchi identity as a restriction on the allowed backgrounds. However, in a quantum theory of gravity, there are no global symmetries, and therefore, there are no background fields either. This is manifested in the fact that all three of $g,A,B$ are actually dynamical fields that we are supposed to path-integrate over. Treating these as backgrounds is justified if there is some sort of weak coupling limit in which the fields become frozen. This is automatically the case at low energies in any ten-dimensional string theory, since the couplings of all of $g,A,B$ are dimensionful and become irrelevant in the deep IR.  It is not the case e.g.\ in six dimensions, where antisymmetric tensor fields are often strongly coupled and cannot be treated perturbatively. In such cases, the only approach available is to explicitly perform the path integral over the tensor fields, compute their effective action, and verify that the resulting path integral indeed cancels against the contributions of other chiral fields. There is no meaningful analog of the notion of having a string structure, since no weak coupling notion is available. The anomaly theory (as a function of the metric and background gauge fields) can then studied on general spin manifolds (and not just string manifolds), and anomalies cancel in a standard way, because the B-field (which is integrated over) couples to background 4 and 8-forms $X_4$ and $X_8$, and has a mixed anomaly captured by the anomaly polynomial $\int X_4 X_8$, just what is needed to cancel the anomaly of the fermions. From a perturbative string worldsheet point of view, we feel it is more natural to keep $B$ as a background field; furthermore, the techniques we use in this paper can be extended to compute lower-dimensional string bordism groups, which control solitonic objects in these non-supersymmetric theories via the Cobordism Conjecture \cite{McNamara:2019rup}.

 \subsection{Evidence for torsional Bianchi identities}

In the previous Subsection, we gave an argument that the Bianchi identity holds at the level of torsion, too. The argument relies heavily on string perturbation theory, and one may worry e.g. that it does not capture strongly coupled situations. In this Section, we provide independent evidence, which does not rely on the worldsheet at all, that the Bianchi identity holds at the level of integer cohomology. 
We do so by computing Dai-Freed anomalies of supersymmetric and non-supersymmetric string theories on simple eleven-dimensional lens spaces. Lens spaces are quotients of spheres by $\Z_p$ groups; they are the simplest examples of manifolds whose cohomology is purely torsional (except in bottom and top degrees, as usual). In particular, their first Pontryagin classes are torsion; the upshot of the calculation in this Section is that spacetime anomalies on lens spaces seem to vanish if and only if the Bianchi identities are satisfied at the level of integral cohomology, including torsional classes.

Now we turn to the details of evaluating anomalies on lens spaces; we refer the reader to \cite{Debray:2021vob, Debray:2023yrs} for more on lens spaces and the corresponding expressions for eta invariants. (Eleven-dimensional) lens spaces are defined to be quotients of the form
\begin{equation}L_p^{11}=S^{11}/\Z_p,\end{equation}
where the $\Z_p$ action acts as scalar multiplication by $e^{2\pi i/p}$ on the six complex coordinates $\mathbb{C}^6$ and where we embed the covering $S^{11}$ as the unit sphere. An important property of these lens spaces is that the Green-Schwarz term,
 \begin{equation} H\wedge X_8,\end{equation}
will automatically vanish on lens spaces, since $H^3(S^{11}/\Z_n;\mathbb{Z})=0$. As a result, the calculation of the full anomalies of string theories on lens spaces reduces to determining the anomalies of the chiral fields. We will now evaluate the anomaly theory of the Type I and the Heterotic string theories (supersymmetric and non-supersymmetric) on certain eleven-dimensional lens spaces. 

\subsubsection{Type I and HO heterotic theories}
As the Green-Schwarz (GS) contribution to the anomaly theory vanishes on lens spaces, the remaining fermion anomaly theory of the type I and HO heterotic theory is given by
\begin{eqaed}\label{eq:type-I_anomaly}
    \alpha(L^{11}_p) = \eta^\text{RS}_0(L^{11}_p) - 3 \, \eta^\text{D}_0(L^{11}_p) + \eta^\text{D}_\mathbf{adj}(L^{11}_p) \, .
\end{eqaed}
The Rarita-Schwinger eta invariant $\eta^{\mathrm{RS}}$ arises from the anomaly theory of a ten-dimensional gravitino according to $\alpha_\text{gravitino} = \eta^{\text{RS}} - 2 \eta^\text{D}$ \cite{Debray:2021vob}. In order to evaluate this anomaly theory on $L^{11}_p$, one can derive the branching rules for the adjoint representation of $\Spin(32)$ in terms of the charge-$q$ irreducible $\Z_p$ representations $\mathcal{L}^q$. This branching depends on how $\Z_p$ is included in the gauge group. We choose a family of inclusions of the form
\begin{eqaed}
    \Z_p \quad \hookrightarrow \quad U(1) \quad \overset{k}{\hookrightarrow} \quad SU(N) \quad \hookrightarrow \quad \Spin(2N) \, ,
\end{eqaed}
according to which the (complexified) vector representation of $\Spin(2N)$ splits as
\begin{eqaed}
    \mathbf{2N} \quad \longrightarrow \quad \mathbf{N} \oplus \mathbf{N}^* \, .
\end{eqaed}
The parameter $k$ denotes an inclusion that places the $U(1)$ fundamental representation $\mathcal{L}$ in $k$ diagonal blocks, in pairs $L \equiv \mathcal{L} \oplus \mathcal{L}^{-1}$, and the rest in the trivial representation $\mathcal{L}^0$. Then, the vector representation of $Spin(2N)$ further splits into
\begin{eqaed}\label{eq:vector_branching}
    \mathbf{2N} \quad \longrightarrow \quad \mathbf{N} \oplus \mathbf{N}^* \quad \longrightarrow \quad [kL \oplus (N-2k)\mathcal{L}^0] \oplus [kL \oplus (N-2k)\mathcal{L}^0] \, .
\end{eqaed}
In order to find the branching rules for other representations, it is convenient to use Chern characters. Letting $x \equiv c_1(\mathcal{L})$, the Chern character of $\mathbf{N}$ (and $\mathbf{N}^*$) decomposes into
\begin{eqaed}\label{eq:ch_suN}
    \text{ch}(\mathbf{N})  \quad \longrightarrow \quad k \left(e^x + e^{-x}\right) + (N-2k) \, .
\end{eqaed}
Then we can build the characters for adjoint, symmetric and antisymmetric $SU(N)$ representations, from which we can reconstruct the characters for $\Spin(2N)$ representations of interest, such as the adjoint (antisymmetric) and spinorial. The resulting branching rules involve the representations $L^q \equiv \mathcal{L}^q \oplus \mathcal{L}^{-q}$. In particular, the adjoint of $\Spin(2N)$ branches according to
\begin{eqaed}
    \mathbf{adj} \longrightarrow k(2k-1) L^2 \oplus 4k(N-2k)L \oplus \left[N(2N-1) - 2k(2k-1) - 8k(N-2k)  \right] \mathcal{L}^0 \, ,
\end{eqaed}
which gives the corresponding eta invariant. Using the expressions
\begin{eqaed}
    \eta^\text{D}_q(L^{11}_p) & = \frac{2 p^6 + 21 p^4 + 168 p^2 - 191 - 42 p^4 q^2 + 210 p^2 q^4 - 630 p^2 q^2}{60480p} \\
    & + \frac{- 252 p q^5 + 1260 p q^3 - 1008 pq + 84 q^6 - 630 q^4 + 1008 q^2}{60480p} \, , \\
    \eta^\text{RS}_0(L^{11}_p) & = \frac{22p^6 - 273p^4 - 3192p^2 + 3443}{60480p} \, ,
\end{eqaed}
the anomaly simplifies to
\begin{eqaed}
    \alpha^{(k)}_{Spin(32)}(L^{11}_p) = \frac{(p^2-1)(p^4 + (11-5k)p^2 + 10(k-3)^2)}{60p} \, .
\end{eqaed}
In order to compare the cases in which $\alpha = 0 \; \text{mod } 1$ to the Bianchi identity, let us recall that the total Pontryagin class of $L^{2k-1}_p$ is $p(L_p^{2k-1})= (1 + y)^k$ with $y$ a generator of $H^4(L^{2k-1}_p; \mathbb{Z}) \cong \Z_p$. Thus $p_1(L_p^{11}) = 6y$, and one can show that the canonical choice of $\tfrac{p_1}{2}$ afforded by the spin structure is
$\tfrac{p_1}{2} = 3y$. On the other hand, according to the branching rule \eqref{eq:vector_branching}, the total Chern class of the associated vector bundle is $c = (1 - y)^{2k}$, and thus $c_2 = -2k \, y$. Therefore,
\begin{eqaed}
    \frac{p_1 + c_2}{2} = (3-k) \, y
\end{eqaed}
vanishes if and only if $k = 3 \; \text{mod } p$. Plugging in $k= 3 + m p$ with $m$ integer, the anomaly does vanish (mod 1), and it does not vanish otherwise.

\subsubsection{\texorpdfstring{$E_8 \times E_8$}{E8 x E8} theory}

The calculation for the $E_8 \times E_8$ theory is almost the same as in the preceding case. The anomaly theory has the same form of \eqref{eq:type-I_anomaly}, the only difference being the branching of the adjoint representation $\mathbf{adj} = (\mathbf{248}, \mathbf{1}) \oplus (\mathbf{1}, \mathbf{248})$. We employ the same construction as before, embedding $\Z_p$ into the $\Spin(16)$ subgroup of $E_8$. The general construction is thus specified by a pair $(k_1 , k_2)$ pertaining to the two $E_8$ factors. One then has to compute the branching for the $\mathbf{120}$ and the spinorial $\mathbf{128}$ of $\Spin(16)$ which compose the adjoint representation of $E_8$. The former has been presented in the preceding section, now with $N=8$, while the latter can be constructed computing Chern characters of antisymmetric representations of $SU(8)$ whose direct sum gives the branching of the spinorial representation:
\begin{eqaed}
    \mathbf{128}_+ \oplus \mathbf{128}_- \quad \longrightarrow \quad \bigoplus_{m=0}^8 \mathbf{{\binom{8}{m}}} \, .
\end{eqaed}
The Chern character for the various antisymmetric representations can be found by expanding the graded Chern character for the exterior algebra $\Lambda(V) = \oplus_n \Lambda^n(V)$
\begin{eqaed}
    \text{ch}(\Lambda(V)) \equiv \sum_n t^n \, \text{ch}(\Lambda^n(V)) \, ,
\end{eqaed}
which can be computed exploiting the property $\Lambda(U \oplus V) \simeq \Lambda(U) \otimes \Lambda(V)$ and that, for line bundles $\mathcal{L}$,
\begin{eqaed}
    \text{ch}(\Lambda(\mathcal{L})) = 1 + t \, e^{c_1(\mathcal{L})} \, .
\end{eqaed}
Thus, \eqref{eq:ch_suN} gives
\begin{eqaed}
    \text{ch}(\Lambda(\mathbf{N})) \quad \longrightarrow \quad (1+t \, e^x)^k \, (1 + t \, e^{-x})^k \, (1+t)^{N-2k} \, .
\end{eqaed}
For instance for $N=8$ and $k=1$, summing the even or odd rank characters leads to
\begin{eqaed}
    \text{ch}(\mathbf{128}) \quad \longrightarrow \quad  64 + 32 \left(e^x + e^{-x}\right) \, ,
\end{eqaed}
which means that the spinorial representations branch according to $\mathbf{128} \to 32L \oplus 64\mathcal{L}^0$. Analogously, $\mathbf{120} \to L^2 \oplus 24 L \oplus 70 \mathcal{L}^0$, so that all in all
\begin{eqaed}
    \mathbf{248} \quad \longrightarrow \quad  L^2 \oplus 56L \oplus 134 \mathcal{L}^0 \, .
\end{eqaed}
The anomaly for this particular choice $(k_1, k_2) = (1,0)$ then simplifies to
\begin{eqaed}
    \alpha^{(1,0)}_{E_8 \times E_8}(L^{11}_p) = \frac{p^6 + 5p^4 + 34p^2 - 40}{60p}
\end{eqaed}
which vanishes (mod 1) for $p=2$. Let us now look at the Bianchi identity. The Chern class of the adjoint $E_8 \times E_8$ associated bundle is
\begin{eqaed}
    c = (1-4y) \, (1-y)^{56} \, ,
\end{eqaed}
with $y$ a generator of $H^4(L^{11}_p ; \mathbb{Z})$, and thus $c_2 = - 60 y$. For $E_8 \times E_8$ we have to divide $\frac{c_2}{2}$ by 30 in the Bianchi identity, thus getting
\begin{eqaed}
    \frac{p_1}{2} + \frac{c_2}{60} = 2 \, y \, .
\end{eqaed}
This class only vanishes if $p=2$, which is the same value for which the anomaly vanishes! Similarly, for $(k_1 , k_2) = (1,1)$ the Bianchi class is $y$, which never vanishes (except for the trivial case $p=1$), and accordingly the anomaly never vanishes either.

One can carry on with more complicated embeddings computing the spinorial branching of $\mathbf{128}$: for $(k_1 , k_2) = (2,1)$ the Bianchi class vanishes, and indeed the anomaly turns out to always vanish mod 1. At first glance, the case $(k_1, k_2) = (2,0)$ appears to present an exception, since the Bianchi class is $y \neq 0$ but the anomaly vanishes for $p=5$. However, in order to find the relationship between torsional Bianchi identities and anomalies, for given torsion the anomaly should vanish for all allowed backgrounds, and the $(1,1)$ embedding has the same Bianchi class but nonvanishing anomaly for $p=5$.

The general expression for any $(k_1, k_2)$ for the $E_8 \times E_8$ theory is more involved due to how the spinorial representations branch, but the procedure to compute the anomaly is systematic.

\subsubsection{Non-supersymmetric theories}

Let us now address the non-supersymmetric cases. An immediate consequence of the above result for the supersymmetric heterotic theories is that the anomaly on lens spaces satisfying the torsional Bianchi identity also vanishes for the non-SUSY heterotic theory, since its chiral matter content is in the virtual difference of the corresponding representations \cite{Schellekens:1986xh}. This fact will turn out to be useful when discussing fivebrane anomaly inflow in section \ref{sec:so16so16_inflow}.

For the Sagnotti model, the anomaly theory can be written as \cite{Sagnotti:1995ga, Sagnotti:1996qj}
\begin{eqaed}\label{eq:sagnotti_lens_anomaly}
    \alpha_{0'\text{B}}(L_p^{11}) & = \alpha_\text{self-dual}(L_p^{11}) - \eta^\text{D}_{\text{\bf antisym}}(L_p^{11}) \\
    & = - \, \alpha^\text{RS}_0(L_p^{11}) + 3 \, \eta^\text{D}_0(L_p^{11}) - \, \eta^\text{D}_{\text{\bf antisym}}(L_p^{11})
\end{eqaed}
since it contains a four-form RR field with self-dual curvature, similarly to type IIB. Following the same procedure as before, now with the simpler inclusion $\Z_p \hookrightarrow U(1) \hookrightarrow SU(32)$, one can evaluate the fermionic anomalies; for the self-dual field, in the second line of \cref{eq:sagnotti_lens_anomaly} we have used anomaly cancellation in type IIB supergravity to recast its anomaly theory in terms of fermionic contributions, along the lines of \cite{Debray:2021vob}. Thus we obtain
\begin{eqaed}
    \alpha^{(k)}_{0'\text{B}}(L_p^{11}) & = - \, \frac{(p^2-1)(5k^2 - 5k (p^2+12) +2 (p^4 + 11p^2 + 90) )}{120p} \, .
\end{eqaed}
The Chern class of the associated fundamental bundle is now $c = (1 - y)^k$, so that $c_2 = -k \, y$ and the Bianchi class
\begin{eqaed}
    \frac{p_1 + c_2}{2} = \left(3 - \, \frac{k}{2} \right) y
\end{eqaed}
vanishes for $k = 6 \; \text{mod } p$. Notice that $c_3 = 0$ as well for these bundles, since the total Chern class only contains powers of $y \in H^4(L_p^{11}, \mathbb{Z})$. Substituting $k = 6 + m p$ for integer $m$, the anomaly vanishes as expected, but not otherwise.

The calculation for the Sugimoto model is essentially identical: the anomaly theory is simply $\alpha_\text{Sugimoto} = - \, \alpha_{0'\text{B}}$, since the antisymmetric fermion has now positive chirality and the gravitino and dilatino contribute the opposite of the self-dual tensor. The inclusion we employ is $\Z_p \hookrightarrow U(1) \hookrightarrow \Sp(1) \simeq SU(2) \overset{k}{\hookrightarrow} \Sp(16)$, under which the $\mathbf{32}$ representation branches according to
\begin{eqaed}
    \mathbf{32} \quad \longrightarrow \quad kL \oplus (32 - 2k)\mathcal{L}^0 \, .
\end{eqaed}
Since the resulting Bianchi class is also the same, one obtains the same result: the anomalies cancel on lens backgrounds which satisfy the Bianchi identity at the torsional level.

\subsection{Vanishing bordism classes}
We now turn to the main results of this paper -- the calculation of string bordism groups with twisted string structures corresponding to the non-supersymmetric strings, by means of homotopy theory. These sections cover in some detail the mathematical aspects of the calculation; a table summarizing the results  can be found in the Conclusions.
\subsubsection{\texorpdfstring{$\Sp(16)$}{Sp(16)}}
\label{ss:sugi}
At this point we make our first bordism computation: that every closed, spin $11$-manifold $M$ with a principal
$\Sp(16)$-bundle $P$ satisfying the Green-Schwarz identity $\tfrac 12 p_1(M) + c_2(P) = 0$ is the boundary of a compact spin
$12$-manifold on which the $\Sp(16)$-bundle and Green-Schwarz data extend. This implies that the anomalies we
study in this paper vanish for the Sugimoto string.

To make these computations, we use the Adams and Atiyah-Hirzebruch spectral sequences. By now these are standard
tools in the mathematical physics literature, so we point the reader to~\cite{BC18, Debray:2023yrs} for background and
many example computations written for mathematical physicists. The computations in this paper are a little more
elaborate: twisted string bordism rather than twisted spin bordism. There are fewer such calculations in the
literature, but we found the references~\cite{Hil07, Hil09, BR21, Deb23, DY23} helpful.

On to business. The data of a $G$-gauge field and a B-field satisfying a Bianchi identity is expressed mathematically as a principal bundle for a Lie $2$-group extension of $G$ by $BU(1)$. Such extensions are classified by $H^4(BG;\Z)$~\cite[Corollary 97]{SP11}. Let $\String(n)\text{-}\Sp(16)$ denote the Lie $2$-group which is the extension of $\Spin(n)\times
\Sp(16)$ by $BU(1)$ classified by $\tfrac 12 p_1 + c_2\in H^4(B(\Spin(n)\times \Sp(16));\Z)$, and %
$\String\text{-}\Sp(16)$ be the colimit as $n\to\infty$ as usual. A string-$\Sp(16)$-structure on a manifold $M$ is
data of a spin structure, a $\Sp(16)$-bundle $P$, and a trivialization of the Green-Schwarz term $\tfrac 12 p_1(M) + c_2(P)$: exactly what we
need for the Sugimoto string.

Though we are primarily interested in showing $\Omega_{11}^{\String\text{-}\Sp(16)} = 0$, the lower-dimensional
bordism groups are barely more work. %
\begin{thm}
\label{sugimoto_thm}
The low-dimensional $\String\text{-}\Sp(16)$ bordism groups are:
\begin{alignat*}{2}
	\Omega_0^{\String\text{-}\Sp(16)} &\cong\Z \qquad & \Omega_6^{\String\text{-}\Sp(16)} &\cong \Z_2\\
	\Omega_1^{\String\text{-}\Sp(16)} &\cong\Z_2 \qquad & \Omega_7^{\String\text{-}\Sp(16)} &\cong 0\\
	\Omega_2^{\String\text{-}\Sp(16)} &\cong\Z_2 \qquad & \Omega_8^{\String\text{-}\Sp(16)} &\cong \Z^{\oplus
		3}\\
	\Omega_3^{\String\text{-}\Sp(16)} &\cong 0 \qquad & \Omega_9^{\String\text{-}\Sp(16)} &\cong (\Z_2)^{\oplus
		2}\\
	\Omega_4^{\String\text{-}\Sp(16)} &\cong\Z \qquad & \Omega_{10}^{\String\text{-}\Sp(16)} &\cong
		(\Z_2)^{\oplus 3}\\
	\Omega_5^{\String\text{-}\Sp(16)} &\cong\Z_2 \qquad & \Omega_{11}^{\String\text{-}\Sp(16)} &\cong 0.
\end{alignat*}
\end{thm}
\begin{proof}
Let $V\to B\Sp(16)$ be the vector bundle associated to the defining representation; it is rank $64$ as a real vector bundle. Then by an argument analogous
to~\cite[\S 10.4]{Debray:2023yrs}, there is an isomorphism $\Omega_*^{\String\text{-}\Sp(16)}\cong
\Omega_*^\String((B\Sp(16))^{V-64})$, where $(B\Sp(16))^{V-64}$ is the \emph{Thom spectrum} of the virtual vector bundle $V - \underline\R^{64}\to B\Sp(16)$. The Thom spectrum $X^V$ of $V\to X$ is a homotopy-theoretic object whose homotopy groups can be expressed as certain kinds of bordism groups by the Pontrjagin-Thom theorem; the upshot is that string bordism groups of $X^V$ are isomorphic to $(X, V)$-twisted string bordism groups of a point. See~\cite[\S 10.4]{Debray:2023yrs} for more information and references.

If $\tmf$ denotes the spectrum of connective
topological modular forms, then it follows that the map $\Omega_*^\String(X)\to \tmf_*(X)$ is an isomorphism in degrees $15$ and below
whenever $X$ is a space or connective spectrum~\cite[Theorem 2.1]{Hil09} (the latter condition includes all Thom spectra we study in this paper). Therefore for the rest of the proof we focus on $\tmf_*((B\Sp(16))^{V-64})$. These are finitely generated
abelian groups, so we may work one prime at a time (see~\cite[\S 10.2]{Debray:2023yrs}).

As input, we will need the following calculation of Borel.
\begin{prop}[{Borel~\cite[\S 29]{Bor53}}]
\label{borel_Sp}
$H^*(B\Sp(16);\Z)\cong\Z[c_2, c_4, \dotsc, c_{32}]$, where $c_i$ is the pullback of the $i^{\mathrm{th}}$ Chern
class under the map $B\Sp(16)\to B\U(32)$.
\end{prop}
For large primes $p$ (i.e.\ $p\ge 5$), we want to show that $\tmf_*((B\Sp(16))^{V-64})$ lacks $p$-torsion in
degrees $11$ and below. This follows because when $p\ge 5$, the homotopy groups of the $p$-localization
$\tmf_{(p)}$  are free and concentrated in even degrees~\cite[\S 13.1]{douglas2014topological}, and the $\Z_{(p)}$ cohomology
of $B\Sp(16)$ (hence also of $(B\Sp(16))^{V-64}$, by the Thom isomorphism) is always free and concentrated in even degrees as a consequence of \cref{borel_Sp} and the universal coefficient theorem, so the Atiyah-Hirzebruch spectral sequence
computing $p$-localized $\tmf_*((B\Sp(16))^{V-64})$ collapses with only free summands on the
$E_\infty$-page, preventing $p$-torsion in $\tmf_*((B\Sp(16))^{V-64})$ in the range we care about.

For $p = 3$, the $3$-localized Atiyah-Hirzebruch spectral sequence does not immediately collapse, so we use the
Adams spectral sequence (and we will see that this Adams spectral sequence does immediately collapse). The Adams spectral sequence takes the form
\begin{equation}\label{big_adams}
    E_2^{s,t} = \Ext_\cA^{s,t}(H^*(X; \Z_p), \Z_p) \Longrightarrow \pi_{t-s}^s(X)_p^\wedge.
\end{equation}
Let us explain this notation. We pick a prime $p$; then $\cA$ is the $p$-primary Steenrod algebra, the $\Z_p$-algebra of all natural transformations $H^*(\text{--};\Z_p)\to H^{*+t}(\text{--};\Z_p)$ that commute with the suspension isomorphism. The mod $p$ cohomology of any space or spectrum $X$ is thus naturally a $\Z$-graded $\cA$-module, so we may apply $\Ext_\cA$, the derived functor of $\Hom_\cA$. This gives us two gradings: the original $\Z$-grading on cohomology is the one labeled $t$, and the grading arising from the derived functors is the one labeled $s$. On the right-hand side of~\eqref{big_adams}, $\pi_*^s$ denotes stable homotopy groups, and $(\text{--})_p^\wedge$ denotes $p$-completion. We will not need to worry in too much detail about $p$-completion: we will only ever $p$-complete finitely generated abelian groups $A$, for which the $p$-completion carries the same information as the free summands and the $p$-power torsion summands of $A$. Thus we will typically be implicit about $p$-completion --- in particular, $\Z_p$ always denotes the cyclic group of order $p$, not the $p$-adic integers.

We are interested in $\tmf$-homology (or really string bordism), rather than stable homotopy, which means replacing $X$ with $\tmf\wedge X$ in~\eqref{big_adams}; then the Adams spectral sequence converges to $\tmf_{t-s}(X)_p^\wedge$.

By work of Henriques and Hill (see~\cite{Hil07, douglas2014topological}), building on work of Behrens~\cite{Beh06} and unpublished work of Hopkins-Mahowald, there is a change-of-rings theorem for the $3$-primary Adams spectral sequence for $\tmf$ simplifying~\eqref{big_adams} to
\begin{equation}
	E_2^{s,t} = \Ext_{\cA^{\tmf}}^{s,t}(H^*(X;\Z_3), \Z_3) \Longrightarrow \tmf_*(X)_3^\wedge.
\end{equation}
Here $\cA^{\tmf}$ is the graded $\Z_3$-algebra
\begin{equation}
	\cA^{\tmf} = \Z_3\ang{\beta, \cP^1}/(\beta^2, (\cP^1)^3, \beta(\cP^1)^2\beta - (\beta\cP^1)^2 -
	(\cP^1\beta)^2),
\end{equation}
with $\abs\beta= 1$ and $\abs{\cP^1} = 4$. For the Adams $E_2$-page, $\cA^{\tmf}$ acts on $H^*(X;\Z_3)$ by sending
$\beta$ to the Bockstein for $0\to\Z_3\to\Z_9\to\Z_3\to 0$ and $\cP^1$ to the first mod $3$ Steenrod power. See~\cite{Hil07, Hil09, BR21, DY23} for more information and some example computations with this variant of the
Adams spectral sequence.

As input, we need to know how $\beta$ and $\cP^1$ act on $H^*((B\Sp(16))^{V-32};\Z_3)$. This is determined in~\cite[Corollary 2.37]{DY23} from the input data of the action of the images of $\beta$ and $\cP^1$ on the mod $3$ Steenrod algebra on $H^*(B\Sp(16); \Z_3)$.
As the cohomology of $B\Sp(16)$ is concentrated in even degrees, $\beta$ must
act trivially, and thus likewise for the Thom spectrum $(B\Sp(16))^{V-64}$. Shay~\cite{Sha77} computes the action of $\cP^1$ on mod $3$ Chern classes;\footnote{We also found Sugawara's explicit calculations of this formula in~\cite[\S 5]{Sug79} helpful.} the formula implies
that in $H^*(B\Sp(16);\Z_3)$, $\cP^1(c_2) = c_4 + c_2^2$ and $\cP^1(c_4) = c_4c_2$. For the Thom class\footnote{We thank Y. Tachikawa for pointing out a minus sign error in this expression (which does not impact our calculation)}, $\cP^1(U) =- Uc_2$~\cite[Theorem 2.28]{DY23}. Using the Cartan formula, we can compute the $\cA^{\tmf}$-module
structure on $H^*((B\Sp(16))^{V-64};\Z_3)$.
\begin{defn}
If $M$ is a $\Z$-graded module over a $\Z$-graded algebra $A$, we will let $\Sigma^k M$ denote the same underlying $A$-module with the grading shifted up by $k$, i.e.\ if $x\in M$ is homogeneous of degree $m$, then $x\in\Sigma^k M$ has degree $m+k$. We will write $\Sigma$ for $\Sigma^1$.
\end{defn}
The notation $\Sigma^k$ is inspired by the suspension of a topological space, which has the effect of increasing the degrees of elements in cohomology by $1$.
\begin{defn}
\label{N3_defn}
Let $N_3$ denote the nontrivial
$\cA^{\tmf}$-module extension of $C\nu$ by $\Sigma^8\Z_3$, where $C\nu$ is the $\cA^{\tmf}$-module defined in~\cite[\S 3.2]{DY23}.
\end{defn}
Then, there is an $\cA^{\tmf}$-module isomorphism
\begin{equation}
\label{sugimoto_mod_3_coh}
	H^*((B\Sp(16))^{V-64};\Z_3)\cong \textcolor{BrickRed}{N_3} \oplus \textcolor{MidnightBlue}{\Sigma^8 N_3}
	\oplus P,
\end{equation}
where $P$ is concentrated in degrees $12$ and above (so we can ignore it). We draw the decomposition~\eqref{sugimoto_mod_3_coh} in \cref{sugimoto_figure_3}, left.

We need to compute $\Ext_{\cA^{\tmf}}(N_3, \Z_3)$. To do so, we use the fact that the short exact sequence of $\cA^{\tmf}$-modules (which we draw in \cref{N2_fig}, top)
\begin{equation}
\label{N2_SES}
    \shortexact{\textcolor{RubineRed}{\Sigma^8 \Z_3}}{N_3}{\textcolor{Periwinkle}{C\nu}}{}
\end{equation}
induces a long exact sequence on Ext groups; traditionally one draws the Ext of the first and third terms of a short exact sequence in the same Adams chart, so that the boundary maps have the same degree as a $d_1$ differential. See Beaudry-Campbell~\cite[\S 4.6, \S 5]{BC18} for more information and some examples for modules over a different algebra $\cA(1)$, and~\cite[Figures 2, 3, and 5]{DY23} for some $\cA^{\tmf}$-module examples.

We will draw the long exact sequence in Ext corresponding to~\eqref{N2_SES} in \cref{N2_fig}. To do so, we need $\Ext_{\cA^{\tmf}}(\textcolor{RubineRed}{\Z_3})$, which is due to Henriques-Hill~\cite{Hil07, douglas2014topological}, and $\Ext_{\cA^{\tmf}}(\textcolor{Periwinkle}{C\nu})$, which is computed in topological degree $14$ and below in~\cite[Figure 2]{DY23}. Our notation for names of Ext classes follows~\cite[\S 3]{DY23}; $\Ext_{\cA^{\tmf}}(\textcolor{RubineRed}{\Sigma^8\Z_3})$ is a free $\Ext_{\cA^{\tmf}}(\Z_3)$-module on a single generator, so call that generator $z$.\footnote{There are two classes which generate $\Ext_{\cA^{\tmf}}(\textcolor{RubineRed}{\Sigma^8\Z_3})$ as an $\Ext_{\cA^{\tmf}}(\Z_3)$-module, and one is $-1$ times the other. For the purposes of this paper, it does not matter which one we call $z$ and which one we call $-z$.} Most boundary maps are nonzero for ``degree reasons,'' meaning that their domain or codomain is the zero group. For $t-s\le 14$, there are two exceptions: $\partial(z)$ could be $\pm\alpha y$ or $0$, and $\partial(\alpha z)$ could be $\pm \beta x$ or $0$. Since the boundary maps commute with the $\Ext_{\cA^{\tmf}}(\Z_3)$-action and $\alpha(\alpha y) = \beta x$,\footnote{The equation $\alpha(\alpha y) = \beta x$ is stated in~\cite[Remark 3.21]{DY23}, but not proven there. One way to prove it is to compare with the equivalent $\alpha$-action $\alpha y\mapsto \beta x$ in $\Ext_{\cA^{\tmf}}(N_1)$ in the long exact sequence in (\textit{ibid.}, Figure 5): because $\partial(\alpha y) = \pm \beta w$ and $\alpha\beta w\ne 0$, and because $\alpha(\partial(\text{--})) = \partial(\alpha\cdot\text{--})$, $\alpha(\alpha y)\ne 0$, hence must be $\pm \beta x$, and we can choose the generator $x$ so that we obtain $\beta x$ and not $-\beta x$. The calculation of $\partial(\alpha x)$ in (\textit{ibid.}, Lemma 3.24) does not use any information about $\alpha(\alpha y)$.} these two boundary maps are either both zero or both nonzero. To see that they are both nonzero, we use that $\Ext_{\cA^{\tmf}}^{0,8}(N_3)\cong \Hom_{\cA^{\tmf}}(N_3, \Sigma^8\Z_3) = 0$, so $z\in\Ext_{\cA^{\tmf}}^{0,8}(\textcolor{RubineRed}{\Sigma^8 \Z_3})$ is not the image of an Ext class for $N_3$, so $\partial(z)\ne 0$.

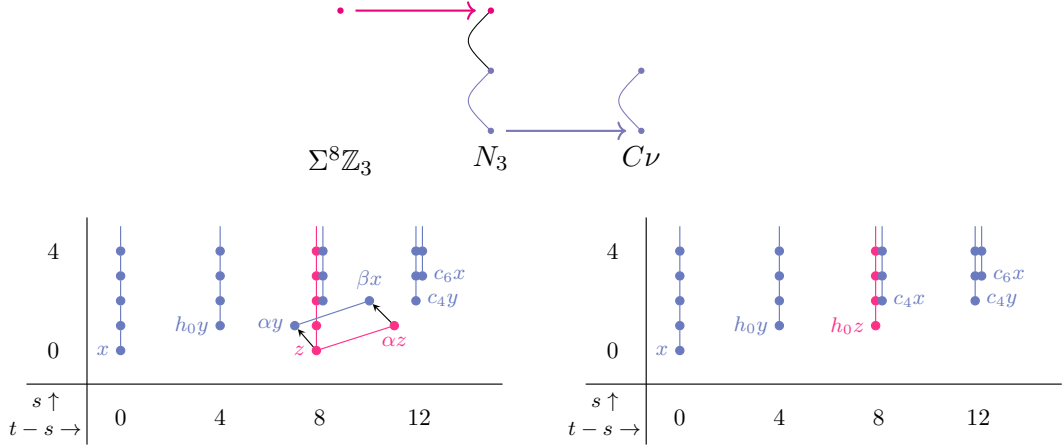
\begin{figure}[htbp!]
\centering
\begin{subfigure}[c]{0.4\textwidth}
\begin{tikzpicture}[scale=0.4]
        \PoneL(0, 2);
        \begin{scope}[RubineRed]
                \foreach \x in {-5, 0} {
                        \tikzpt{\x}{4}{}{};
                }
                \draw[->, thick] (-4.5, 4) -- (-0.5, 4);
        \end{scope}
        \begin{scope}[Periwinkle]
                \foreach \x in {0, 5} {
                        \tikzpt{\x}{0}{}{};
                        \tikzpt{\x}{2}{}{};
                        \PoneL(\x, 0);
                }
                \draw[->, thick] (0.5, 0) -- (4.5, 0);
        \end{scope}
        \node[below=2pt] at (-5, 0) {$\Sigma^8\Z_3$};
        \node[below=2pt] at (0, 0) {$N_3$};
        \node[below=2pt] at (5, 0) {$C\nu$};
\end{tikzpicture}
\end{subfigure}\\
\begin{subfigure}[c]{0.47\textwidth}
\includegraphics[width=\textwidth]{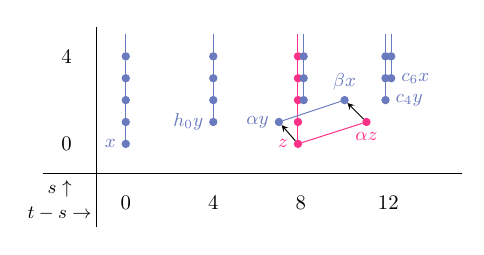}
\end{subfigure}
\begin{subfigure}[c]{0.47\textwidth}
\includegraphics[width=\textwidth]{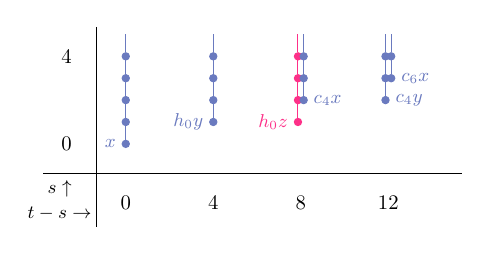}
\end{subfigure}
\caption{Top: the short exact sequence~\eqref{N2_SES} of $\cA^{\tmf}$-modules. Lower left: the induced long exact
sequence in Ext. Lower right: $\Ext_{\cA^{\tmf}}(N_3)$ as
computed by the long exact sequence.}
\label{N2_fig}
\end{figure}

With this Ext in hand, we can draw the $E_2$-page of the Adams spectral sequence in \cref{sugimoto_figure_3},
right. The spectral sequence collapses at $E_2$ in the range we study for degree reasons. The straight lines denote
actions by $h_0\in\Ext_{\cA^{\tmf}}^{1,1}(\Z_3, \Z_3)$, which lift to multiplication by $3$, so we see there is no
$3$-torsion in degrees $11$ and below.
\begin{figure}[htb!]
\centering
\begin{subfigure}[c]{0.2\textwidth}
  \begin{tikzpicture}[scale=0.6, every node/.style = {font=\tiny}]
    \foreach \y in {0, 4, ..., 16} {
      \node at (-2, \y/2) {$\y$};
    }
    \begin{scope}[BrickRed]
      \tikzpt{0}{0}{$U$}{};
      \tikzpt{0}{2}{}{};
	  \tikzpt{0}{4}{}{};
      \PoneL(0, 0);
      \PoneL(0, 2);
    \end{scope}
    \begin{scope}[MidnightBlue]
      \tikzpt{2}{4}{$Uc_2^2$}{};
      \tikzpt{2}{6}{}{};
      \tikzpt{2}{8}{}{};
      \PoneL(2, 4);
	  \PoneL(2, 6);
    \end{scope}
  \end{tikzpicture}
\end{subfigure}
\begin{subfigure}[htb!]{0.6\textwidth}
\includegraphics[width=\textwidth]{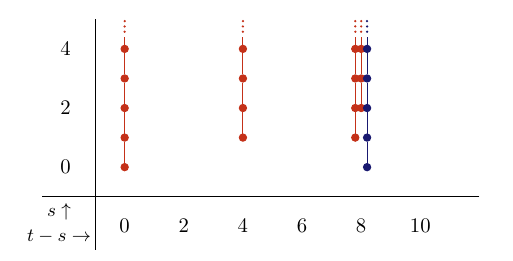}
\end{subfigure}
\caption{Left: the $\cA^{\tmf}$-module structure on $H^*((B\Sp(16))^{V - 64};\Z_3)$ in low degrees; the pictured submodule contains
all elements in degrees $11$ and below. Right: the $E_2$-page of the Adams spectral sequence computing
$\tmf_*((B\Sp(16))^{V-64})_3^\wedge$.}
\label{sugimoto_figure_3}
\end{figure}

Lastly, for $p = 2$, we use the Adams spectral sequence again; the outline of the proof is quite similar to the $p
= 3$ case, but the details are different. Specifically, we will once again use the Adams spectral sequence and a standard change-of-rings theorem to simplify the calculation of the $E_2$-page, but the algebra of cohomology operations is different.

Let $\cA(2)$ be the subalgebra of the mod $2$ Steenrod algebra generated by $\Sq^1$, $\Sq^2$, and $\Sq^4$. There is an isomorphism $H^*(\tmf;\Z_2)\cong \cA\otimes_{\cA(2)}\Z_2$~\cite{HM14, Mat16}, which by a standard argument simplifies the $E_2$-page of the $2$-primary Adams spectral sequence to
\begin{equation}
    E_2^{s,t} = \Ext_{\cA(2)}^{s,t}(H^*(X; \Z_2), \Z_2) \Longrightarrow \tmf_*(X)_2^\wedge.
\end{equation}
The next thing to do is to determine how $\cA(2)$ acts on $H^*((B\Sp(16))^{V-64};\Z_2)$. Since this cohomology ring vanishes in degrees not divisible by $4$, $\Sq^1$ and $\Sq^2$ act trivially. For $\Sq^4$,~\cite[\S 3.3]{BC18} says $\Sq^4(U) = Uw_4(V) = c_2$, and the Wu formula computes the Steenrod squares in $H^*(B\Sp(16);\Z_2)$, using that the mod $2$ reductions of the generators in \cref{borel_Sp} are Stiefel-Whitney classes. This allows us to completely describe the $\cA(2)$-action on $H^*((B\Sp(16))^{V-64};\Z_2)$ in the degrees we need: $\Sq^4(U) = Uc_2$, $\Sq^4(Uc_2^2) = Uc_2^3$, $\Sq^4(Uc_4) = Uc_6$, and all other actions by $\Sq^1$, $\Sq^2$, or $\Sq^4$ starting in degree $11$ or below vanish. Thus, if $M_4$ denotes the $\cA(2)$-module consisting of two $\Z_2$ summands in degrees $0$ and $4$ connected by a $\Sq^4$, there is an isomorphism
\begin{equation}
\label{BSp16_A2}
    H^*((B\Sp(16))^{V-64};\Z_2) \cong
        \textcolor{BrickRed}{M_4} \oplus
        \textcolor{Green}{\Sigma^8 M_4} \oplus
        \textcolor{MidnightBlue}{\Sigma^8 M_4} \oplus P,
\end{equation}
where $P$ contains no elements in degrees $11$ or below, and hence will be irrelevant to our calculations. We draw~\eqref{BSp16_A2} in \cref{sugimoto_figure_2}, left.

Bruner-Rognes~\cite[\S 4.4]{BR21} compute $\Ext_{\cA(2)}(M_4)$; using their result, we give the $E_2$-page of the Adams spectral sequence computing $\tmf_*((B\Sp(16))^{V-64})_2^\wedge$ in \cref{sugimoto_figure_2}, top right.

\begin{figure}[htb!]
\centering
\begin{subfigure}[c]{0.3\textwidth}
  \begin{tikzpicture}[scale=0.6, every node/.style = {font=\tiny}]
    \foreach \y in {0, 4, ..., 12} {
      \node at (-2, \y) {$\y$};
    }
    \begin{scope}[BrickRed]
      \tikzptR{0}{0}{$U$}{};
	  \tikzpt{0}{4}{}{};
      \sqfourL(0, 0);
    \end{scope}
    \begin{scope}[Green]
      \tikzptR{0}{8}{$Uc_2^2$}{};
	  \tikzpt{0}{12}{}{};
      \sqfourL(0, 8);
    \end{scope}
    \begin{scope}[MidnightBlue]
      \tikzptR{2}{8}{$Uc_4$}{};
	  \tikzpt{2}{12}{}{};
      \sqfourL(2, 8);
    \end{scope}
  \end{tikzpicture}
\end{subfigure}
\begin{subfigure}[htb!]{0.6\textwidth}
%
\includegraphics[width=\textwidth]{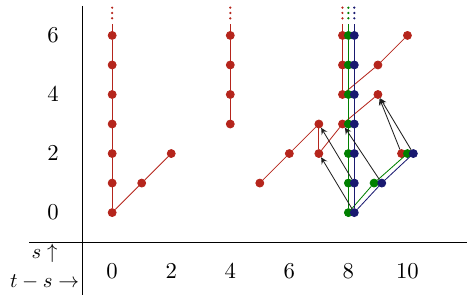}
\includegraphics[width=\textwidth]{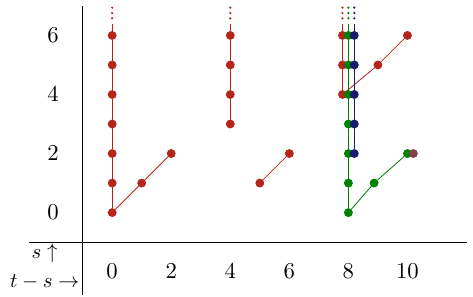}
\end{subfigure}
\caption{Left: the $\cA(2)$-module structure on $H^*((B\Sp(16))^{V - 64};\Z_2)$ in low degrees; the pictured submodule contains
all elements in degrees $11$ and below. Above right: the $E_2$-page of the Adams spectral sequence computing
$\tmf_*((B\Sp(16))^{V-64})_2^\wedge$. Below right: the $E_3 = E_\infty$ page.}
\label{sugimoto_figure_2}
\end{figure}

Looking at the $E_2$-page, most differentials are ruled out by degree considerations or the fact that they must commute with the action of $h_0$ or $h_1$. There is a $d_2$ out of $E_r^{0,8}$, and $d_2\colon E_2^{2,12}\to E_2^{4,13}$. There also four potentially non-trivial $d_2$ differentials out of $E_2^{0,8},\; E_2^{1,10},\; E_2^{1,9}$ and $E_2^{2,12}$, as shown in in \cref{sugimoto_figure_2}, top right.

These non-trivial differentials can kill classes in the $E_2$-page, as we will now see\footnote{These differentials were erroneously concluded to vanish in an earlier version of this manuscript. We are thankful to Shota Saito and Yuji Tachikawa for pointing this out to us and sharing a note with the corrected calculation.}. They can be understood as follows. Let us consider the generator of $E_2^{0,8}$ in the Adams spectral sequence. Classes $x\in E_2^{0,\bullet}$ of an Adams spectral sequence for $G$-bordism correspond (at prime 2) naturally to (a subset of) $\Z_2$-valued characteristic classes $c_x$ for manifolds with $G$-structure, and $x$ survives to the $E_\infty$-page if and only if there is a closed manifold $M$ with $G$-structure such that $\int_M c_x = 1$; see~\cite[\S 8.4]{Freed:2019sco}.

For the Adams spectral sequence for string-$\Sp(16)$ bordism at $p = 2$, the two classes corresponding to a basis of $E_2^{0,8}$ are the mod $2$ reductions of $c_2^2$ and $c_4$, since these generate the cohomology of $BSp(16)$ at prime 2 (see \cite[\S 8.4]{Freed:2019sco}). We now consider closed string-$\Sp(16)$ $8$-manifolds on which these classes do not vanish.
\begin{itemize}
    \item The quaternionic projective plane $\HP^2$ has a tautological principal $\Sp(1)$-bundle $P := S^{11}\to\HP^2$; let $P^\vee\to \HP^2$ be the same space with the quaternion-conjugate $\Sp(1)$-action, and let $Q\to \HP^2$ be the principal $\Sp(16)$-bundle induced from $P^\vee$ by the inclusion $i\colon \Sp(1)\to\Sp(16)$. Using the fact that $i$ pulls $c_2$ back to $c_2$ and Borel and Hirzebruch's calculation of the characteristic classes of $\HP^2$~\cite[\S 15.5, \S 15.6]{BH58} (see also~\cite[\S 5.2]{Freed:2019sco} for a good review), the reader can verify that $(\HP^2, Q)$ has a unique string-$\Sp(16)$ structure, meaning in particular that $c_2(Q) = -\tfrac 12 p_1(\HP^2)$, and that $\int_{\HP^2} c_2(Q)^2 = 1$.
    \item An 8-manifold with manifold with the smallest possible value of $c_4$,  $\int c_4=1$, was identified in  \cite{LY22}. However, this manifold does not admit a twisted-string structure. In fact, any $\text{Sp}(16)$-twisted string manifold in dimension eight has $\int c_4$ being a multiple of 12\footnote{We thank Y. Tachikawa for sharing a note on this with us.}, as we will now show. Consider an $Sp(16)$ bundle $P$ over a spin 8-manifold, let $V$ be the defining fundamental representation of $Sp(16)$ and let $\underline{V}$ be the corresponding virtual representation with dimension 0. Then the index of the Dirac operator on the spinor bundle tensored with $P \times _{Sp(16)} \underline V$ is 
    \begin{equation}
        \text{ind} D(\underline{V}) = \frac{1}{12}c_2(P)\left( c_2(P) + \frac{p_1(M)}{2}\right) + \frac{c_4(P)}{6},
    \end{equation} where the first term vanishes by definition of twisted string structure. Since $\underline{V}$ is a pseudorreal representation, the index is even, we can conclude that $\int_{S^8} c_4(P)$ is a multiple of 12 for any twisted-string manifold.
    A concrete example of a manifold with $\int c_4=\pm 12$ bis given by $S^8$ with principal $\Sp(16)$-bundle $P\to S^8$ classified by either generator of
    \begin{equation}
        [S^8, B\Sp(16)] = \pi_8(B\Sp(16))\overset\cong\to\pi_8(B\mathit{Sp}) = \pi_0(B\Sp) = \Z,
    \end{equation}
    using Bott periodicity. Since $H^4(S^8;\Z) = 0$, $c_2(P)$ and $\tfrac 12 p_1(S^8)$ vanish and therefore $(S^8, P)$ is string-$\Sp(16)$. We suspect, but did not prove, that this manifold attains the smallest possible value $\int c_4=\pm 12$. 
    \qedhere
\end{itemize}
\end{proof}
Since there is no mod 2 cohomology class that detects this second generator of twisted string bordism, the corresponding classes must be killed in the Adams spectral sequence. At prime 3, this appears automatically in the calculation of Ext groups leading to the $E_2$ page, which is why this subtlety did not appear in our discussion. At prime 2, however,  and this factor of 12 (or equivalently 4, modulo 2) needs to appear as nontrivial differentials. Therefore two non-zero differentials should kill the generators of $E_2^{0,8}$ and $E_2^{1,9}$ in the mod-2 Adams spectral sequence to give a factor of $2^2 = 4$, where the second on. There are two additional differentials coming from $E^{1,10}$ and $E^{2,12}$, since differentials commute with the action of $h_0$.

Taking a look now at the non-zero differential from $E^{1,10}$, we see that it cannot point to the green dot in $E^{3,11}_2$ as this will give another factor of 2, which goes against the existence of configurations with $\int c_4 = 12$, since we have to kill the two dots in $E_2^{0,8}$.

The differential out of the red dot in $E_2^{2,12}$ does not vanish --- to see this, consider the map
\begin{subequations}
\label{nonvan_diffs}
\begin{equation}
    f\colon \HP^1\longrightarrow\HP^\infty\simeq B\Sp(1)\longrightarrow B\Sp(16)
\end{equation}
and the induced map on Thom spectra
\begin{equation}
    f_*\colon \tmf_*((\HP^1)^{f^*V - 64})\longrightarrow\tmf_*((B\Sp(16))^{V - 64}).
\end{equation}
\end{subequations}
The map $f_*$ induces a pullback map on mod $2$ cohomology and on Adams spectral sequences; the map on mod $2$ cohomology is the quotient by all elements of degree greater than $4$, so the effect on Adams spectral sequences is to kill all summands in Ext except for the red summands. As $H^*((\HP^1)^{f^*V - 64};\Z_2)$ consists of two $\Z_2$ summands in degrees $0$ and $4$, joined by a $\Sq^4$, the Adams spectral calculating its $2$-completed $\tmf$-homology is worked out by Bruner-Rognes~\cite[Theorem 8.1]{BR21}, who show that $d_2\colon E_2^{2,12}\to E_2^{4,13}$ is an isomorphism. Thus this differential persists to the $\Sp(16)$ Adams spectral sequence.

Once we conclude that one differential out of $E_2^{2,12}$ is nonzero, and kills the only class in the target, the other one (coming from the blue dot) has no choice but to vanish. More precisely, the red and blue dots in $(10, 2)$ have the same $d_2$, so $d_2$ of their sum is $0$. Therefore their sum lives to $E_\infty$.

We now turn the page and obtain the $E_3$-page of the Adam's spectral sequence, in \cref{sugimoto_figure_2}, bottom right. Most differentials are ruled out by degree considerations or the fact that they must commute with the action of $h_0$ or $h_1$. By $h_0$- and $h_1$-linearity, all longer differentials vanish in the range depicted and therefore $E_3 = E_\infty$ in topological degrees $12$ and below.

\subsubsection{\texorpdfstring{$\U(32)$}{U(32)}}
\label{ss:sagnotti}
Now we discuss the Sagnotti string, whose gauge group is $\U(32)$. The Green-Schwarz mechanism for this theory
involves three classes in degrees $2$, $4$, and $6$ canceling $c_1$, $c_2$, and $c_3$ of the gauge bundle,
respectively.

We may impose the degree-$2$, $4$, and $6$ conditions on $B\U(32)$ in any order. Starting with $c_1$, we obtain
$B\SU(32)$; then, let $B\SU(32)\ang{c_3}$ denote the fiber of the map
\begin{equation}
	c_3\colon B\SU(32)\longrightarrow K(\Z, 6).
\end{equation}
A map $X\to B\SU(32)\ang{c_3}$ is equivalent data to a rank-$32$ complex vector bundle $V\to X$ with
$\SU$-structure and a trivialization of $c_3(V)$. There is a tautological such vector bundle $V_t\to
B\SU(32)\ang{c_3}$, which is the pullback of the tautological bundle over $B\SU(32)$.

Finally, the degree-$4$ condition for a $\U(32)$-bundle $V$ over a manifold $M$ asks for a trivialization of $\tfrac
12 p_1(M) + c_2(V)$. Thus, we ask for a $(B\SU(32)\ang{c_3}, V_t)$-twisted string structure on $M$, i.e.\ a map $f\colon M\to B\SU(32)\ang{c_3}$ and a string structure on $TM\oplus f^*V_t$; the Whitney sum formula for $\tfrac 12 p_1$ unwinds this into the usual Green-Schwarz condition. We will be a little casual with the notation and call a $(B\SU(32)\ang{c_3}, V_t)$-structure a
$\String\text{-}\SU(32)\ang{c_3}$-structure, even though we do not construct a Lie $2$-group
$\String\text{-}\SU(32)\ang{c_3}$ realizing this twisted string structure (and indeed, there is no guarantee one
exists).
\begin{thm}
\label{sag_thm}
$\Omega_{11}^{\String\text{-}\SU(32)\ang{c_3}}$ is isomorphic to either $0$ or $\Z_2$.
\end{thm}
The ambiguity is in a differential we were not able to resolve. Unfortunately, this means we were not able to use
bordism-theoretic methods alone to calculate the anomaly of the Sagnotti string. Our proof also yields partial
information on lower-dimensional bordism groups; there is ambiguity due to Adams spectral sequence differentials,
some of which we suspect are nonzero.
\begin{proof}
Before we start our analysis, we need to understand $H^*(B\SU(3)\ang{c_3}; A)$ for various coefficient rings $A$.
As $B\SU(32)\ang{c_3}$ is the fiber of $c_3\colon B\SU(32)\to K(\Z, 6)$, the fiber of $B\SU(32)\ang{c_3}\to
B\SU(32)$ is $\Omega K(\Z, 6) \simeq K(\Z, 5)$; moreover, this fibration pulls back from the universal fibration
with fiber $K(\Z, 5)$, namely the loop-space-path-space fibration for $K(\Z, 6)$:
\begin{equation}
\label{2serre}
\begin{tikzcd}
	{K(\Z, 5)} & {K(\Z, 5)} \\
	& {B\SU(32)\ang{c_3}} & {*} \\
	& {B\SU(32)} & {K(\Z, 6).}
	\arrow["{c_3}", from=3-2, to=3-3]
	\arrow[from=2-2, to=3-2]
	\arrow[from=2-3, to=3-3]
	\arrow[from=1-2, to=2-3]
	\arrow[from=1-1, to=2-2]
	\arrow["{=}", from=1-1, to=1-2]
	\arrow[from=2-2, to=2-3]
	\arrow["\lrcorner"{anchor=center, pos=0.125}, draw=none, from=2-2, to=3-3]
\end{tikzcd}\end{equation}
We will compute $H^*(B\SU(32)\ang{c_3}; A)$ for various $A$ using the Serre spectral sequence, along with some
information gained from the map of Serre spectral sequences induced by~\eqref{2serre}.

As for the Sugimoto string, we work one prime at a time.
\begin{lem}
\label{lp_sagnotti}
For $p\ge 5$, there is no $p$-torsion in $H^*(B\SU(32)\ang{c_3}; \Z)$ in degrees $12$ and below, and all free summands are
concentrated in even degrees.
\end{lem}
\begin{proof}
It suffices to work with cohomology valued in the ring $\Z[1/6]$ of rational numbers whose denominators in lowest terms are of the form $2^m3^n$, as tensoring with $\Z[1/6]$ preserves all $p$-power torsion for $p\ge 5$.

Cartan~\cite{Car54} and Serre~\cite{Ser52} computed $H^*(K(\Z, n);\Z[1/6])$; their formulas imply that when $p \ge 5$, $H^k(K(\Z, 5);\Z[1/6])$ is torsion-free for $k\le 12$, and vanishes apart from $H^5(K(\Z, 5);\Z[1/6])\cong\Z[1/6]$.

Now consider the Serre spectral sequence for the fibration on the left in~\eqref{2serre} using cohomology with $\Z[1/6]$ coefficients. The map of fibrations~\eqref{2serre} induces a map of Serre spectral sequences, and this map is an isomorphism on
$E_2^{0,\bullet}$. Since this map commutes with differentials, this means the fate of all classes in
$E_2^{0,\bullet}$ is determined by their preimages in the spectral sequence for $K(\Z, 5)\to * \to K(\Z, 6)$. For
example, we know thanks to Serre~\cite[\S 10]{Ser53} that in that spectral sequence, $E$ transgresses to the mod $2$
reduction of the tautological class $F$ of $K(\Z, 6)$. Therefore in the spectral sequence for $B\SU(32)\ang{c_2}$,
$E$ transgresses to the pullback of $F$, which is $c_3$. The Leibniz rule then tells us $d_6(xE) = xc_3$ for $x\in H^*(B\SU(32);\Z[1/6])$; since this cohomology ring is polynomial, $xc_3\ne 0$ as long as $x\ne 0$, so these differentials never vanish. Therefore the nonzero part of the $E_\infty$-page, at least in total degree $12$ and below, is a quotient of $E_2^{*, 0} = H^*(B\SU(32);\Z[1/6])$. Since $H^*(B\SU(32);\Z[1/6])$ is free and concentrated in even degrees. This implies the $E_\infty$-page is also free and concentrated in even degrees in total degree $12$ and below, which implies the lemma statement.
\end{proof}
\begin{cor}
For $p\ge 5$, $\Omega_k^{\String\text{-}\SU(32)\ang{c_3}}$ lacks $p$-torsion for $k\le 11$.
\end{cor}
\begin{proof}
We want to compute $\Omega_*^\String((B\SU(32)\ang{c_3})^{V_t -64})$, and as noted above, we may replace
$\Omega_*^\String$ with $\tmf$ for the degrees $k$ in the corollary statement. Because of \cref{lp_sagnotti} and
the fact that $\tmf_{(p)}$ has homotopy groups concentrated in even degrees and lacks $p$-torsion, the
Atiyah-Hirzebruch spectral sequence computing $\tmf_*((B\SU(32)\ang{c_3})^{V_t - 64})_p^\wedge$ collapses with no
$p$-torsion in the range $11$ and below.
\end{proof}
As usual, $p = 2$ and $p = 3$ are harder.
\begin{lem}\hfill
\begin{enumerate}
	\setcounter{enumi}{1}
	\item  $H^*(B\SU(32)\ang{c_3}; \Z_2)\cong\Z_2[c_2, G, c_4, H, J, K, c_6, L,\dotsc]/(\dots)$, with $\abs{c_i} =
	2i$, $\abs{G} = 7$, $\abs{H} = 8$, $\abs{J} =10$, $\abs{K} = 11$, and $\abs{L} =12$; all missing generators
	and relations are in degrees $13$ and above. In addition, we have the following Steenrod squares:
	\begin{itemize}
		\item $\Sq^1$ vanishes on the named generators except $\Sq^1(G) = H + \lambda_1 c_4$ and $\Sq^1(K) = L +
		\lambda_2 c_2c_4 + \lambda_3c_6$ for some $\lambda_1,\lambda_2,\lambda_3 \in\Z_2$.
		\item $\Sq^2$ vanishes on the named generators except for $\Sq^2(H) = J$ and possibly on $c_6$, $K$, and
		$L$.
		\item $\Sq^4(c_2) = c_2^2$, $\Sq^4(c_4) = c_2c_4 + c_6$, $\Sq^4(G) = K$, and $\Sq^4(H) = L$.
	\end{itemize}
	\item $H^*(B\SU(32)\ang{c_3};\Z_3)\cong\Z_3[c_2, c_4, J, c_6, \dotsc]/(\dots)$ with $\abs{c_i} = 2i$ and
	$\abs{J} =10$, and with all missing generators and relations in degrees $13$ and above; $c_i$ denotes the
	pullback of the mod $3$ reduction of the $i^{th}$ Chern class along $B\SU(32)\ang{c_3}\to B\SU(32)$, and
	$\cP^1(c_2) = c_2^2 + c_4$ and $\cP^1(c_4) = c_2c_4$.
\end{enumerate}
\end{lem}
\begin{proof}
This is a standard argument with the Serre spectral sequence for the fibration on the left in~\eqref{2serre}, so we
sketch the details.

For the mod $2$ cohomology, we need as input $H^*(B\SU(32);\Z_2)\cong\Z_2[c_2, c_3, \dotsc, c_{32}]$ with
$\abs{c_i} = 2i$~\cite[\S 29]{Bor53}: these are the mod $2$ reductions of the Chern classes. We also need
$H^*(K(\Z,5);\Z_2)$, which was computed by Serre~\cite[\S 10]{Ser53}. This is a polynomial ring on infinitely many
generators; the six in degrees below $13$ are $E\in H^5$, the mod $2$ reduction of the tautological class;
$G:=\Sq^2(E)$; $H:=\Sq^1(H)$; $I:=\Sq^4(E)$; $K:=\Sq^4(G)$; and $L:=\Sq^5(G)$.

In the Serre spectral sequence, the class $E$ transgresses to $c_3$, and the proof is the same as in the proof of \cref{lp_sagnotti}.
Similarly, we divine the fate of the other classes on the line $p = 0$:
\begin{itemize}
	\item In the Serre spectral sequence for the rightmost fibration in~\eqref{2serre}, $G$ transgresses via $d_8$ to
	$\Sq^2(F)$ by the Kudo transgression theorem~\cite{Kud56}, so in the leftmost fibration, $d_r(G) = 0$ for
	$r\le 7$, and $d_8(G) = \Sq^2(c_3) = 0$. Thus $G$ is a permanent cycle.
	\item In a similar way, $H$ transgressing to $\Sq^3(F)$ via $d_9$ pulls back to imply $d_9(H) = \Sq^3(c_3) =
	0$, so $H$ is also a permanent cycle. Likewise, $E^2$ is a permanent cycle, because in the fibration over
	$K(\Z, 6)$, it supports the transgressing $d_{11}(E^2) = \Sq^5(F)$, and $\Sq^5(c_3) = 0$, and similarly $K$ and
	$L$ are permanent cycles.
	\item $I = \Sq^4(E)$ transgresses to $\Sq^4(c_3) = c_5 + c_2c_3$.
\end{itemize}
The Leibniz rule then cleans up the rest of the spectral sequence in degrees $12$ and below. This gives us the ring
structure. For the Steenrod squares, we use the information of the $\cA$-action on $H^*(B\SU(32);\Z_2)$ coming from
the Wu formula, together with the $\cA$-actions we gave when describing $H^*(K(\Z, 5);\Z_2)$ above. There is
ambiguity in the Steenrod squares in $H^*(B\SU(32)\ang{c_3};\Z_2)$ coming from the loss of information passing
to the associated graded on the $E_\infty$-page, which is the source of $\lambda_1$, $\lambda_2$, and $\lambda_3$ in
the theorem statement. However, by pulling back to the analogous fibration over $B\SU(2)$, where the fiber bundle
admits a section (as $c_3$ of an $\SU(2)$-bundle is canonically trivial), so we can use the Künneth formula to
compute Steenrod squares. Pulling back to $\SU(2)$ loses all information about $c_i$ for $i > 2$, so this leaves
ambiguity in $c_4$ and $c_6$ as described in the theorem statement, but resolves the ambiguity involving $c_2$.
Some ambiguity can be erased by redefining generators, which is how we disambiguate $\Sq^4(H) = L$, but this still
leaves the choices listed in the theorem statement.

For $\Z_3$ cohomology, we begin with $H^*(K(\Z, 5);\Z_3)\cong\Z_3[E, \cP^1(E), \beta\cP^1(E), \cP^2(E), \dots]$,
with the remaining generators in degrees $15$ and above, where $E\in H^5(K(\Z, 5);\Z_3)$ is the mod $3$ reduction
of the tautological class. Just like for $\Z_2$ cohomology, $E$ transgresses via $d_5$ to $c_3$; then the Kudo
transgression theorem~\cite{Kud56} tells us
\begin{itemize}
	\item $\cP^1(E)$ transgresses via $d_{10}$ to $\cP^1(c_3) = c_2c_3 - c_5 = -c_5$ by the $E_{10}$-page (as
	$d_5(c_2E) = c_2c_3$), and
	\item $\beta\cP^1(E)$ is a permanent cycle (it transgresses to $\beta(-c_5) = 0$, which we know for degree
	reasons).
\end{itemize}
We obtain $\cP^1(c_i)$ from Sugawara's calculations~\cite[\S 5]{Sug79} of Shay's formula~\cite{Sha77}. With the fate of these classes known, the Leibniz rule cleans up the rest of the spectral sequence in total degrees $12$ and below to obtain the theorem statement.
\end{proof}
Now, just as in the proof of \cref{sugimoto_thm}, we run the Adams spectral sequences at $p = 3$ and $p = 2$. The twist
by $V_t$ twists the action of $\cP^1$ at $p = 3$, and the action of $\Sq^4$ at $p = 2$, in an analogous way.
Reusing names of $\cA^{\tmf}$-modules from \S\ref{ss:sugi}, we conclude that there is an $\cA^{\tmf}$-module
isomorphism
\begin{equation}
\label{sag_3_Atmf}
	H^*((B\SU(32)\ang{c_3})^{V_t - 64}; \Z_3) \cong
		\textcolor{BrickRed}{N_3} \oplus
		\textcolor{Green}{\Sigma^8 N_3} \oplus
		\textcolor{MidnightBlue}{\Sigma^{10} N_3} \oplus P,
\end{equation}
where $P$ is concentrated in degrees $12$ and above, so will not affect us. We draw~\eqref{sag_3_Atmf} in
\cref{sagnotti_figure_3}, left. We calculated
$\Ext_{\cA^{\tmf}}(\textcolor{BrickRed}{N_3})$ in \cref{N2_fig}; 
using this, we discover that, like for the Sugimoto string, in degrees $11$ and below, the
$E_2$-page consists only of $h_0$-towers in degrees $0$, $4$, and $8$, so there can be no $3$-torsion. See
\cref{sagnotti_figure_3}, right, for a picture of this Adams spectral sequence.
\begin{figure}[!htbp]
\centering
\begin{subfigure}[c]{0.3\textwidth}
  \begin{tikzpicture}[scale=0.6, every node/.style = {font=\tiny}]
    \foreach \y in {0, 4, ..., 16} {
      \node at (-2, \y/2) {$\y$};
    }
    \begin{scope}[BrickRed]
      \tikzpt{0}{0}{$U$}{};
      \tikzpt{0}{2}{}{};
	  \tikzpt{0}{4}{}{};
      \PoneL(0, 0);
      \PoneL(0, 2);
    \end{scope}
    \begin{scope}[Green]
      \tikzpt{2}{4}{$Uc_4$}{};
      \tikzpt{2}{6}{}{};
      \tikzpt{2}{8}{}{};
      \PoneL(2, 4);
      \PoneL(2, 6);
    \end{scope}
	\tikzpt{0}{5}{$U J$}{MidnightBlue};
        \tikzpt{0}{7}{}{};
        \tikzpt{0}{9}{}{};
        \PoneL(0, 5);
        \PoneL(0, 7);
  \end{tikzpicture}
\end{subfigure}
\begin{subfigure}[!htbp]{0.6\textwidth}
\includegraphics[width=\textwidth]{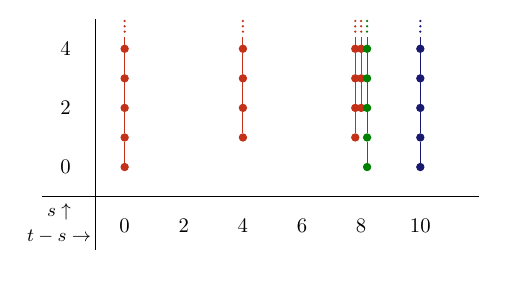}
\end{subfigure}
\caption{Left: the $\cA^{\tmf}$-module structure on $H^*((B\SU(32)\ang{c_3})^{V_t - 64};\Z_3)$ in low degrees; the pictured submodule contains all elements in degrees $11$ and below. Right: the $E_2$-page of the Adams spectral sequence computing
$\tmf_*((B\SU(32)\ang{c_3})^{V_t-64})_3^\wedge$. This figure is part of the proof of \cref{sag_thm}.}
\label{sagnotti_figure_3}
\end{figure}
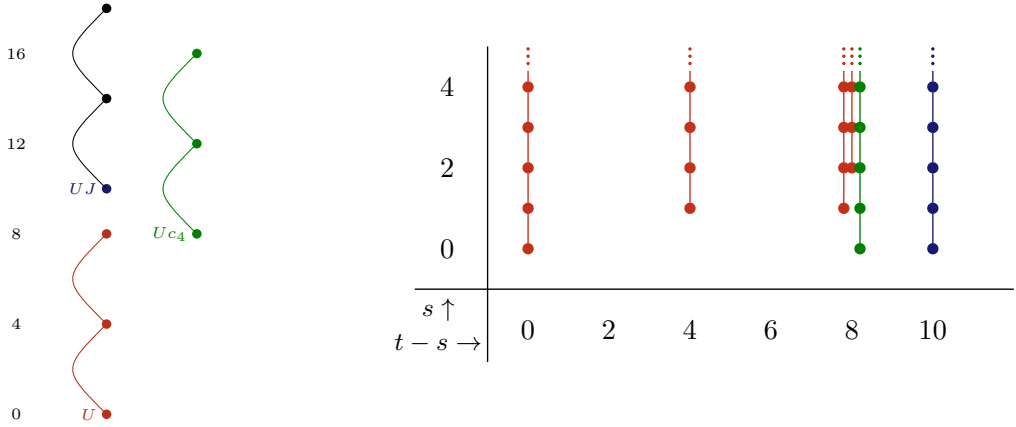

Last, $p = 2$. The ambiguity in the Steenrod actions is not severe enough to get in the way of the existence of an
isomorphism of $\cA(2)$-modules
\begin{equation}
\label{sagnotti_A2}
	H^*((B\SU(32)\ang{c_3})^{V_t - 64}; \Z_2) \cong
		\textcolor{BrickRed}{M_4} \oplus
		\textcolor{RedOrange}{\Sigma^7 N_1} \oplus
		\textcolor{Green}{\Sigma^8 M_4} \oplus
		\textcolor{MidnightBlue}{\Sigma^8 M_4} \oplus
		\textcolor{Fuchsia}{\Sigma^{11} N_2} \oplus P
\end{equation}
where $P$ is concentrated in degrees $12$ and above, so will be irrelevant for us, and:
\begin{itemize}
	\item $\textcolor{RedOrange}{N_1}$ is isomorphic to $\cA(2)\otimes_{\cA(1)} Q$ in degrees $6$ and below (i.e.\
	the quotients of these two $\cA(2)$-modules by their submodules of elements in degrees $7$ and above are
	isomorphic); and
	\item $\textcolor{Fuchsia}{N_2}$ is isomorphic to $\cA(2)\otimes_{\cA(1)} Q$ in degrees $3$ and below.
\end{itemize}
Here $Q$ is the ``question mark,'' the $\cA(1)$-module which has a $\Z_2$-vector space basis $\{x_0, x_1,
x_3\}$ with $\abs{x_i} = i$, and with $\Sq^1(x_0) = x_1$, $\Sq^2(x_1) = x_3$ (all other $\cA(1)$-actions are trivial
for degree reasons). The module $\textcolor{Fuchsia}{\Sigma^{11}N_2}$ is generated by $Uc_2G$. We draw the
decomposition~\eqref{sagnotti_A2} in \cref{sagnotti_figure_2}, left.

Bruner-Rognes~\cite[\S 4.44]{BR21} calculate $\Ext_{\cA(2)}(M_4)$. For $\textcolor{RedOrange}{N_1}$ and
$\textcolor{Fuchsia}{N_2}$, we use that an isomorphism of $\cA(2)$-modules in degrees $k$ and below implies the
existence of an isomorphism of Ext groups in topological degrees $k-1$ and below, so it is good enough to know
$\Ext_{\cA(2)}(\cA(2)\otimes_{\cA(1)}Q)$; then the change-of-rings theorem (see, e.g., \cite[\S 4.5]{BC18}) implies
\begin{equation}
	\Ext_{\cA(2)}(\cA(2)\otimes_{\cA(1)}Q)\cong \Ext_{\cA(1)}(Q),
\end{equation}
and Adams-Priddy~\cite[Table 3.11]{AP76} compute $\Ext_{\cA(1)}(Q)$. Putting all this together, we can draw the
$E_2$-page in \cref{sagnotti_figure_2}, right. In the range relevant to us, the $E_2$-page is generated as an $\Ext_{\cA(2)}(\Z_2)$-module by the following ten summands.
\begin{enumerate}
    \item Coming from $\Ext(\textcolor{BrickRed}{M_4})$: $a_1\in\Ext^{0,0}$, $a_2\in\Ext^{3,7}$, $a_3\in\Ext^{1,6}$, $a_4\in\Ext^{2,9}$, and $a_5\in\Ext^{2,12}$.
    \item Coming from $\Ext(\textcolor{RedOrange}{\Sigma^7 N_1})$: $b_1\in\Ext^{7,0}$ and $b_2\in\Ext^{2,12}$.
    \item Coming from $\Ext(\textcolor{Green}{\Sigma^8 M_4})$: $c\in\Ext^{0,8}$.
    \item Coming from $\Ext(\textcolor{MidnightBlue}{\Sigma^8 M_4})$: $d\in\Ext^{0,8}$.
    \item Coming from $\Ext(\textcolor{Fuchsia}{\Sigma^{11}N_2})$: $e\in\Ext^{0,11}$.
\end{enumerate}
These are subject to various relations: notably, if $x$ is any one of these generators, $h_2x = 0$.

\begin{figure}[!htbp]
\centering
\begin{subfigure}[c]{0.38\textwidth}
  \begin{tikzpicture}[scale=0.6, every node/.style = {font=\tiny}]
    \foreach \y in {0, 2, ..., 12} {
      \node at (-2, \y) {$\y$};
    }
    \begin{scope}[BrickRed]
      \tikzptR{0}{0}{$U$}{};
      \tikzpt{0}{4}{}{};
      \sqfourL(0, 0);
    \end{scope}
	\begin{scope}[RedOrange]
		\tikzptR{0}{7}{$UG$}{};
		\foreach \y in {8, 10, 11, 12} {
			\tikzpt{0}{\y}{}{};
		}
		\sqone(0, 7);
		\sqone(0, 11);
		\sqtwoL(0, 8);
		\sqfourL(0, 7);
		\sqfourR(0, 8);
	\end{scope}
	\begin{scope}[Green]
		\tikzptR{2}{8}{$Uc_2^2$}{};
		\tikzpt{2}{12}{}{};
		\sqfourL(2, 8);
	\end{scope}
	\begin{scope}[MidnightBlue]
		\tikzptR{4}{8}{$Uc_4$}{};
		\tikzpt{4}{12}{}{};
		\sqfourL(4, 8);
	\end{scope}
	\begin{scope}[Fuchsia]
		\tikzpt{5}{11}{$Uc_2G$}{};
		\tikzpt{5}{12}{}{};
		\sqone(5, 11);
	\end{scope}
  \end{tikzpicture}
\end{subfigure}
\begin{subfigure}[!htbp]{0.6\textwidth}
\includegraphics{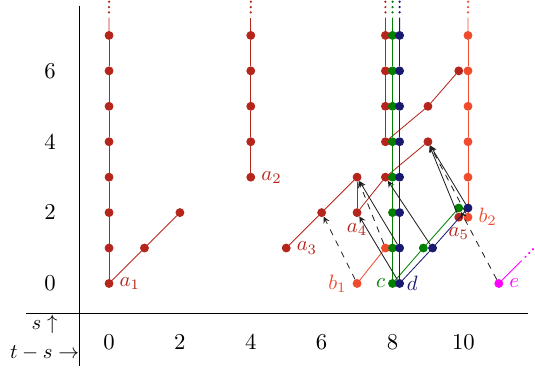}
\end{subfigure}
\caption{Left: the $\cA(2)$-module structure on the quotient of $H^*((B\SU(32)\ang{c_3})^{V_t - 64};\Z_2)$ by its submodule of elements in degrees greater than $12$; the pictured submodule contains all elements in degrees $11$ and below. Right: the $E_2$-page of the Adams spectral sequence computing
$\tmf_*((B\SU(32)\ang{c_3})^{V_t-64})_2^\wedge$. This figure is part of the proof of \cref{sag_thm}. Solid arrows indicate $d_2$ differentials that we deduce in \cref{sugimoto_to_sagnotti}; dashed arrows indicate differentials that we do not determine. No matter the values of the dashed differentials, the spectral sequence collapses at $E_3$ within this range, which follows from \cref{nod6}.}
\label{sagnotti_figure_2}
\end{figure}

Once we take into account the fact that differentials commute with $h_0$ and $h_1$, we still need to determine $d_2(b_1)$, $d_r(c)$, $d_r(d)$, $d_2(a_5)$, $d_2(b_2)$, and $d_r(e)$.
\begin{lem}
\label{sugimoto_to_sagnotti}
For all $r$, $d_r(c) = 0$ and $d_r(d) = a_4$; $d_2(a_5) = h_1^2a_4$.
\end{lem}
The values of these differentials are the same as for the corresponding classes in the Adams spectral sequence for the Sugimoto string, and the proofs are the same as we gave for them in \S\ref{ss:sugi}.
\end{proof}
Ultimately
we need to address $d_2\colon E_2^{0,11}\to E_2^{2,13}$. If this differential vanishes, there is also potential for
$d_6\colon E_6^{0,11}\to E_6^{6, 16}$ to be nonzero. The fate of these two differentials determines whether
$\Omega_{11}^{\String\text{-}\SU(32)\ang{c_3}}$ is nonzero, so it is unfortunate that the techniques we applied
were unable to resolve them.

We are able to obtain some partial information, though.
\begin{lem}\label{nod6}
If $d_2(e) = 0$, so that $d_6(e)$ is defined, then $d_6(e) = 0$.
\end{lem}
\begin{proof}
$d_6(e)\in E_6^{6,16}$. No other nonzero differentials have source or target $E_r^{6,16}$, so $E_6^{6,16}\cong E_2^{6,16}\cong (\Z_2)^{\oplus 2}$, spanned by the classes $w_1h_1^2 a_1$ and $h_0^4b_2$. Here $w_1\in\Ext_{\cA(2)}^{4,12}(\Z_2)$ is the class whose image in $\Ext_{\cA(1)}(\Z_2)$ is the Bott periodicity class. Thus there are $\lambda_1,\lambda_2\in\Z_2$ such that
\begin{equation}
    d_6(e) = \lambda_1 w_1h_1^2 a_1 + \lambda_2 h_0^4 b_2.
\end{equation}
Because $d_6$ commutes with $h_0$ and $h_0e = 0$, $\lambda_2 = 0$.

To show $\lambda_1 = 0$, consider the map of Adams spectral sequences induced by the map from the $\tmf$-homology to the $\ko$-homology of $(B\SU(32)\ang{c_3})^{V_t - 32}$. The map on $E_2$-pages is the map
\begin{equation}
\label{A2to1}
    \Ext_{\cA(2)}(H^*((B\SU(32)\ang{c_3})^{V_t - 32};\Z_2), \Z_2) \longrightarrow
    \Ext_{\cA(1)}(H^*((B\SU(32)\ang{c_3})^{V_t - 32};\Z_2), \Z_2)
\end{equation}
induced by the inclusion $\cA(1)\to \cA(2)$ of algebras. It is possible to compute the right-hand Ext groups using the decomposition~\eqref{sagnotti_A2} and the techniques in~\cite{BC18}; one learns that $e$ and $w_1h_1^2 a_1$ both remain nonzero after~\eqref{A2to1}, so it suffices to compute $d_6(e)$ in the $\ko$-homology Adams spectral sequence. There, though, the submodule $M_{a_1}$ of the Ext groups generated by $a_1$ splits off: because $V_t -32$ is spin, there is a Thom isomorphism $\ko_*((B\SU(32)\ang{c_3})^{V_t - 32})\cong\ko_*(B\SU(32)\ang{c_3})$, so $\ko_*(\pt)$ splits off; as this splitting lifts to the level of spectra, it also splits $M_{a_1}$ off of the Adams spectral sequence, so all differentials into $M_{a_1}$ from any other summand vanish. Thus $d_6(e)$ cannot be $h_1^2 w_1a_1$ in the $\ko$-homology Adams spectral sequence, so the same is true in the $\tmf$-homology Adams spectral sequence.
\end{proof}
Likewise, since $E_2^{2,12}\cong (\Z_2)^{\oplus 4}$, spanned by the classes $a_5$, $b_2$, $h_1^2 c$, and $h_1^2d$, then there are $\lambda_1,\dotsc,\lambda_4\in\Z_2$ such that
\begin{equation}
\label{d2_sagnotti}
    d_2(e) = \lambda_1 a_5 + \lambda_2 b_2 + \lambda_3 h_1^2 c + \lambda_4 h_1^2 d.
\end{equation}
\begin{lem}
In~\eqref{d2_sagnotti}, $\lambda_1 = 0$, $\lambda_2 = 0$, and $\lambda_4 = 0$.
\end{lem}
\begin{proof}
Because $h_0 b_2\ne 0$ but $h_0h_1^2 = 0$, if $\lambda_2 \ne 0$, then $h_0d_2(e)\ne 0$. However, since $d_2$ commutes with $h_0$-multiplication, and $h_0e = 0$, $\lambda_2$ must vanish and $d_2(e)\in\mathrm{span}(a_5, h_1^2c, h_1^2d)$.

By \cref{sugimoto_to_sagnotti}, $d_2(c) = 0$ and $d_2(d) = a_4$, so $d_2(h_1^2c) = 0$ and $d_2(h_1^2d) =h_1^2a_4$, and $d_2(a_5) = h_1^2a_4$. Therefore $d_2\colon\mathrm{span}(a_5, h_1^2c, h_1^2d)\to E_2^{4,13}$ is nonzero on a class $\mu_1 a_5 + \mu_2 h_1^2c + \mu_3 h_1^2d$ if and only if $\mu_1 \ne \mu_3$. Thus $\lambda_1 = \lambda_4$: otherwise $d_2(d_2(e))\ne 0$, and it is always true that $d_2\circ d_2 = 0$.

For $\lambda_4$, consider the map $r\colon B\SU(3)\ang{c_3}\to B\SU(32)\ang{c_3}$ and the map $r_*$ it induces of Adams spectral sequences. The pullback $r^*$ on cohomology kills $c_4$ but leaves $c_2$ and $G$ alone; therefore on Ext groups, $e\in\mathrm{Im}(r_*)$ (because $e$ is the filtration $0$ class corresponding to to $c_2G$), $h_1^2c\in\mathrm{Im}(r_*)$ (because $c$ is the filtration $0$ element corresponding to $c_2^2$), and $h_1^2 d\not\in\mathrm{Im}(r_*)$ (because $d$ corresponds to $c_4$). The map $r_*$ commutes with differentials, so $d_2(e)\in\mathrm{Im}(r_*)$, which is only consistent if $\lambda_4 = 0$. Thus, by the previous paragraph, $\lambda_1 = 0$ too.
\end{proof}
Determining whether $\lambda_3 = 0$ appears to be difficult. This would be a good problem to address because if $\lambda_3\ne 0$, so that $d_2(e)\ne 0$, then the bordism group controlling the anomaly of the Sagnotti string would vanish, and the anomaly would cancel, at least on the class of backgrounds we studied.

Because the class $e$ potentially causing a nonzero bordism group is in Adams filtration $0$, the corresponding bordism invariant is the integral of a modulo $2$ characteristic class, explicitly
\begin{equation}
    \int c_2G.
\end{equation}
The class $G\in H^7(B\SU(32)\ang{c_3};\Z_2)$ is a little mysterious, so we go into some more detail; it is an example of a \emph{secondary characteristic class} in the sense of Peterson-Stein~\cite{PS62}.

Recall that by a trivialization of a cohomology class $z\in H^k(X; A)$, where $A$ is an abelian group, we mean a null-homotopy of a map $f_z\colon X\to K(A, k)$ whose homotopy class represents $z$. There is a space of such trivializations, and a standard result in obstruction theory implies that its set of path components is a torsor over $H^{k-1}(X; A)$. In other words, given two trivializations of $f_z$, their difference is well-defined as an element of $H^{k-1}(X; A)$.

The Wu formula implies $\Sq^2(c_3) = 0$ in $H^8(B\SU(32);\Z_2)$, and in fact provides a canonical trivialization for $\Sq^2(c_3)$. Pulling back to $B\SU(32)\ang{c_3}$ trivializes $c_3$, and therefore provides a second trivialization of $\Sq^2(c_3)$. The difference between these two trivializations is the class $G\in H^7(B\SU(32)\ang{c_3};\Z_2)$.

As a final comment, if we knew of a manifold with a non-trivial integral of $c_2 G$, it would by definition be a generator of the bordism group. We could evaluate the anomaly theory on it in order to determine whether or not the anomaly vanishes. Regrettably, we do not know of such a manifold. Finding it would actually also settle the bordism question as well, since classes at the bottom row of the spectral sequence must be detected by cohomology classes.

\subsubsection{\texorpdfstring{$\Spin(16)\times\Spin(16)$}{Spin(16) x Spin(16)}}
\label{SO16_no_swap}
Next we discuss the symmetry type of the non-supersymmetric heterotic string with gauge Lie algebra
$\mathfrak{so}(16)\oplus\mathfrak{so}(16)$. Although usually called $SO(16)^2$, there are fields transforming in spinor representations in the massless spectrum of the theory which means that we should instead consider $\Spin(16)^2$. There is a further subtlety: according to~\cite{McI99}, the gauge group $G$ is the quotient
of $\Spin(16)\times\Spin(16)$ by the diagonal $\Z_2$ subgroup $\ang{(k, k)}$, where $k\in\Spin(16)$ is either
central element not equal to $\pm 1$.\footnote{Strictly speaking, the analysis of~\cite{McI99} does not take into account the full string spectrum. Therefore, \emph{a priori} the correct gauge group $G$ may differ from this particular quotient of $\Spin(16)^2$.}

As the computation of $H^*(BG)$ is complicated, we will make a simplifying assumption: only working with the double
cover $\Spin(16)\times\Spin(16)$, as we mentioned above. Thus our anomaly cancellation results are only partial information: if we found an anomaly for $\Spin(16)^2$, it would imply the existence of an anomaly for the actual gauge group $G$. However, we found that anomalies cancel for $\Spin(16)^2$, which is only partial information: there could be an anomaly of the theory which vanishes when restricted to gauge fields induced from a $\Spin(16)^2$ gauge field. 
It would be
interesting to address the more general question of the anomaly for $G$.\footnote{Like for any double cover, for any odd prime $p$, the quotient $B\Spin(16)\times B\Spin(16)\to BG$ is a $p$-primary equivalence, so the lack of $p$-primary torsion we establish for $\Spin(16)\times\Spin(16)$ remains valid for $G$.}

Let $\String\text{-}\Spin(16)^2$ be the Lie $2$-group which is the string cover of
$\Spin\times\Spin(16)\times\Spin(16)$ corresponding to the degree-$4$ cohomology class $\tfrac 12 p_1^{(1)} -
\tfrac 12 p_1^{(2)} - \tfrac 12 p_1^{(3)}$, where $c^{(i)}$ refers to the cohomology class $c$ coming from the
$i^{\mathrm{th}}$ factor of $B\Spin$ or $B\Spin(n)$.\footnote{Elsewhere in the paper we have referred to $\tfrac 12 p_1^L$ and $\tfrac 12 p_1^R$ as Chern classes, and indeed they are Chern classes of the representations that play a role in the Green-Schwarz mechanism for this string theory. However, the bordism computation we perform in this section only depends on the characteristic class, not the representation (this is the thesis of~\cite{DY23}), so to emphasize this independence, we use the more intrinsic name $\tfrac 12 p_1$, as this class is one-half of the first Pontrjagin class of the vector representation of $\Spin(n)$.}
Quotienting $\String\text{-}\Spin(16)^2$ by the $\Spin(16)^2$
factor produces a map to $\Spin$; composing with $\Spin\to\mathit O$ we obtain a tangential structure as usual.
\begin{thm}\label{noswap}
In degrees $11$ and below, the $\String\text{-}\Spin(16)^2$ bordism groups are:
\begin{alignat*}{2}
	\Omega_0^{\String\text{-}\Spin(16)^2} &\cong \Z \qquad & \Omega_6^{\String\text{-}\Spin(16)^2} &\cong 0\\
	\Omega_1^{\String\text{-}\Spin(16)^2} &\cong \Z_2 \qquad & \Omega_7^{\String\text{-}\Spin(16)^2} &\cong 0\\
	\Omega_2^{\String\text{-}\Spin(16)^2} &\cong \Z_2 \qquad & \Omega_8^{\String\text{-}\Spin(16)^2} &\cong \Z^2
		\oplus \textcolor{BrickRed}{\Z^3} \oplus \textcolor{MidnightBlue}{\Z}\\
	\Omega_3^{\String\text{-}\Spin(16)^2} &\cong 0 \qquad & \Omega_9^{\String\text{-}\Spin(16)^2} &\cong
	(\Z_2)^{\oplus 2} \oplus \textcolor{BrickRed}{(\Z_2)^{\oplus 2}} \oplus \textcolor{MidnightBlue}{\Z_2}\\
	\Omega_4^{\String\text{-}\Spin(16)^2} &\cong \Z \oplus\textcolor{BrickRed}{\Z}
		\qquad & \Omega_{10}^{\String\text{-}\Spin(16)^2} &\cong (\Z_2)^{\oplus 3} \oplus
		\textcolor{BrickRed}{(\Z_2)^{\oplus 3}} \oplus \textcolor{MidnightBlue}{\Z_2}\\
	\Omega_5^{\String\text{-}\Spin(16)^2} &\cong 0 \qquad & \Omega_{11}^{\String\text{-}\Spin(16)^2} &\cong 0.
\end{alignat*}
\end{thm}
The colors in the theorem statement will be explained below; they correspond to different summands in (an
approximation to) $\mathit{MT}(\String\text{-}\Spin(16)^2)$.
\begin{proof}
The inclusion $i\colon \Spin(16)\to\Spin$ induces a map $Bi\colon B\Spin(16)\to B\Spin$ which is $15$-connected,
because it is an isomorphism on cohomology in degrees $15$ and below. This map sends $\tfrac 12 p_1$ to $\tfrac 12
p_1$, so is compatible with the construction of $\String\text{-}\Spin(16)^2$ --- that is, if
$\String\text{-}\Spin^2$ is defined in the same way as $\String\text{-}\Spin(16)^2$ but using $\Spin$ instead
of $\Spin(16)$, then $i$ induces a map of tangential structures
\begin{equation}
	i_2\colon B(\String\text{-}\Spin(16)^2)\to B(\String\text{-}\Spin^2),
\end{equation}
as well as the analogous map on bordism groups. Because $i$ is $15$-connected, $i_2$ is also $15$-connected, so the
induced map of Thom spectra is also $15$-connected (e.g.\ check on cohomology, where it follows from
$15$-connectivity of $i_2$ via the Thom isomorphism). Therefore for $k\le 15$, the map
$\Omega_k^{\String\text{-}\Spin(16)^2}\to\Omega_k^{\String\text{-}\Spin^2}$ induced by $i$ is an isomorphism.
Therefore for the rest of this proof, we can work only with $\String\text{-}\Spin^2$ bordism without affecting the
results.

Concretely, a string-$\Spin^2$ structure on a vector bundle $E\to X$ is data of a spin structure on $E$ and two
virtual spin vector bundles $V^L,V^R\to X$ and a trivialization of $\tfrac 12 p_1(E) - \tfrac 12p_1(V^L) -\tfrac
12p_1(V^R)$. Since $\tfrac 12p_1$ is additive in direct sums~\cite[Lemma 1.6]{Deb23}, this is equivalent to a
trivialization of $\tfrac 12 p_1(E - V^L - V^R)$, meaning that a string-$\Spin^2$ structure is equivalent to the data
of $V^L$ and $V^R$ and a string structure on $W := E - V^L - V^R$.

The data $(E, V^L, V^R)$ and $(E, W, V^R)$ are equivalent, as $V^L = E - W - V^R$, and the spin structure on $V^L$
can be recovered from the spin structures on $E$, $W$, and $V^R$ by the two-out-of-three property (the string
structure on $W$ includes data of a spin structure). Therefore the data of a string-$\Spin^2$ structure on $E\to X$
is equivalent to the following data:
\begin{itemize}
	\item a spin structure on $E$,
	\item a virtual string vector bundle $W\to X$, and
	\item a virtual spin vector bundle $V^R\to X$.
\end{itemize}
Taking bordism groups, we learn
\begin{equation}
	\Omega_*^{\String\text{-}\Spin^2} \overset\cong\longrightarrow \Omega_*^\Spin(B\Spin\times B\String).
\end{equation}
For any spaces $A$ and $B$, the stable splitting $\Sigma_+^\infty(A)\simeq \Sigma^\infty A\vee \mathbb S$ and its analogue for $B$ together imply a stable splitting
\begin{subequations}
\begin{equation}
	\Sigma_+^\infty(A\times B) \simeq \mathbb S\vee \Sigma^\infty A\vee \Sigma^\infty B\vee \Sigma^\infty(A\wedge
	B),
\end{equation}
implying that for any generalized homology theory $h$,
\begin{equation}
\label{product_splitting}
	h_*(A\times B)\cong h_*(\mathrm{pt})\oplus \widetilde h_*(A)\oplus \widetilde h_*(B)\oplus \widetilde h_*(A\wedge B).
\end{equation}
Here $\widetilde h(X)$ denotes ``reduced $h$-homology'' of a space $X$, meaning the quotient $h(X)/i_*(h(\mathrm{pt}))$ induced by a choice of basepoint $i\colon \mathrm{pt}\to X$. Thus for example $\widetilde\Omega_*^\Spin(X)$ denotes reduced spin bordism, etc.

Apply~\eqref{product_splitting} for $h = \Omega_*^\Spin$, $A = B\Spin$, and $B = B\String$:
\begin{equation}
\label{wherecolor}
\begin{aligned}
	\Omega_*^{\String\text{-}\Spin^2} &\cong\Omega_*^\Spin(B\Spin\times B\String)\\
 &\cong \Omega_*^\Spin \oplus
		\textcolor{BrickRed}{\widetilde\Omega_*^\Spin(B\Spin)} \oplus
		\textcolor{MidnightBlue}{\widetilde\Omega_*^\Spin(B\String)} \oplus
		\textcolor{Green}{\widetilde\Omega_*^\Spin(B\Spin\wedge B\String)}.
\end{aligned}
\end{equation}
\end{subequations}
The colors in~\eqref{wherecolor} indicating the pieces of this direct-sum decomposition correspond to the colors in
the theorem statement displaying which pieces of the bordism groups come from which summands in~\eqref{wherecolor}.

The final step is to determine the four summands in~\eqref{wherecolor}.
\begin{itemize}
	\item $\Omega_*^\Spin$ was calculated by Milnor~\cite[\S 3]{Mil65} and Anderson-Brown-Peterson~\cite{ABP67}.
	\item $\textcolor{BrickRed}{\widetilde\Omega_*^\Spin(B\Spin)}$ was calculated by Francis~\cite[\S 2.2]{Fra11}.
	\item $\textcolor{MidnightBlue}{\widetilde\Omega_*^\Spin(B\String)}$ was computed by Davis~\cite{Dav83} at $p = 2$. At odd primes, these groups are easy to calculate in the range we need: because $B\String$ is $7$-connected,
	$\textcolor{MidnightBlue}{\widetilde\Omega_k^\Spin(B\String)}$ vanishes for $k < 8$; for $8\le k \le 11$, use
	the Atiyah-Hirzebruch spectral sequence. Work of Stong~\cite{Sto63} and Giambalvo~\cite{Gia69} implies that in degrees $11$ and
	below, $\widetilde H^*(B\String; \Z)$ consists of a single summand isomorphic to $\Z$ in
	degree $8$, and the remaining groups vanish. This suffices to collapse the Atiyah-Hirzebruch spectral sequence
	into the blue groups in the theorem statement.
	\item For $A= \Z$ or $\Z_2$,  $\widetilde H^k(B\Spin; A)$ vanishes for $k <4$, and $H^k(B\String; A)$ vanishes
	for $k < 8$, so by the Künneth formula, $\widetilde H^k(B\Spin\wedge B\String; A)$ vanishes for $k < 12$.
	Therefore the Atiyah-Hirzebruch spectral sequence for $\textcolor{Green}{\widetilde\Omega_*^\Spin(B\Spin\wedge
	B\String)}$ vanishes in degrees $11$ and below.
	\qedhere
\end{itemize}
\end{proof}
\begin{rem}[Analogy with $E_8\times E_8$]\label{susy_noswap}
The two-step simplification of $\String\text{-}\Spin(16)^2$ (first replace $\Spin(16)$ with $\Spin$, then recast as spin bordism of a space) is directly analogous to Witten's~\cite[\S 4]{Wit86} simplification of the symmetry type of the $E_8\times E_8$ heterotic string: first, there is a $15$-connected map $BE_8\to K(\Z, 4)$, so in dimensions relevant to string theory we may replace the former with the latter; then Witten recast the data of the two maps to $K(\Z, 4)$ and the twisted string structure given by the Green-Schwarz procedure as a spin structure and a single map to $K(\Z, 4)$.
\end{rem}
\begin{rem}[Analogy with $\Spin(32)$ and detecting a non-supersymmetric $0$-brane]
The same two-step procedure also works for the $\mathit{Spin}(32)/\Z_2$ heterotic string when one restricts to $\Spin(32)$-bundles, showing that the relevant twisted string bordism groups coincide with $\Omega_*^\Spin(B\String)$, which vanishes in dimension $11$. As with $\Spin(16)\times\Spin(16)$, this is only partial information towards a complete anomaly cancellation result.

However, the partial information provided by these bordism groups is already useful: combined with the Cobordism Conjecture~\cite{McNamara:2019rup}, it detects Kaidi-Ohmori-Tachikawa-Yonekura's non-supersymmetric $0$-brane~\cite{Kaidi:2023tqo}. To see this, consider $\Omega_8^\Spin(B\String)\cong\Z^3$: two of the $\Z$ summands come from $\Omega_*^\Spin(\mathrm{pt})$, and as such are generated by $\HP^2$ and the Bott manifold; the third $\Z$ summand is represented by $S^8$ with the map to $B\String$ given by the generator of $[S^8, B\String] = \pi_8(B\String) \cong \Z$. Tracing through the simplification from twisted string bordism of $B\Spin(32)$ to the spin bordism of $B\String$, we see that this $S^8$ has the $\Spin(32)$-bundle arising from the generator of $\pi_8(\Spin(32))\cong\Z$, which is detected by $p_2$.

The Cobordism Conjecture predicts that associated to this bordism class (or rather its image in the corresponding bordism group for $\mathit{Spin}(32)/\Z_2$), there is a $0$-brane in $\mathrm{Spin}(32)/\Z_2$ heterotic string theory whose link is $S^8$ with this $\Spin(32)$-bundle and twisted string structure. This is precisely the $0$-brane discovered by Kaidi-Ohmori-Tachikawa-Yonekura~\cite{Kaidi:2023tqo}. Those authors also discuss a $6$-brane in the $\mathit{Spin}(32)/\Z_2$ heterotic string, but its description uses $\pi_1(\mathit{Spin}(32)/\Z_2)\cong\Z_2$, so it is invisible to the $\Spin(32)$ computation we made here. 
\end{rem}

\subsection{Physical intuition from fivebrane anomaly inflow}\label{sec:physicalintuition}

As we have just seen, the relevant bordism groups vanish, and therefore there are no Dai-Freed anomalies (except possibly for the Sagnotti string). It is instructive to study the vanishing of anomalies more explicitly in particular examples, to better understand the physics at play. Let us recall from Section \ref{sec:global_anomalies} the structure of the anomaly theory for ten dimensional theories that feature a Green-Schwarz mechanism:\begin{equation}\label{eq:GSanomtheory2}\alpha_{\text{GS}}(Y_{11})= \int_{Y_{11}} H \wedge X_8.\end{equation}
The boundary mode of this eleven dimensional field theory gives exactly the contribution of the Green-Schwarz term to the classical action: \begin{equation}S_{\text{GS}}= \int_{Y_{10}} B_2 \wedge   X_8.\end{equation}
In this section, we consider simple backgrounds of the factorized form $Y_{11} = S^3 \times M_8$ for the anomaly theory \eqref{eq:GSanomtheory2}. We will also take one unit of three-form $H$ flux threading the sphere, so that the Green-Schwarz term gives a nontrivial contribution, and $M_8$ a spin manifold equipped with a gauge bundle $E$ such that $\frac{p_1(M_8) + c_2(E)}{2}$ is trivial in integer cohomology. Unlike more general backgrounds, these factorized ones allows for an intuitive understanding of how anomalies are cancelled, via inflow.

On these backgrounds, the eta invariant contribution to the anomaly theory (coming from the fermions) vanishes on account of the factorization property
\begin{align}
    \eta(A \times B) = \eta(A) \, \text{index}(B),
\end{align}
where $A$ is odd-dimensional. The eta invariant of fermions on $S_H^3$ vanishes modulo 1, as it is the same as the eta invariant on a three-sphere, which is the boundary of $\mathbb{R}^4$.  As a result, the anomaly theory simplifies to the Green-Schwarz term
\begin{align}\label{eq:intX8}
    \alpha(S_H^3 \times M_8) = \int_{M_8} X_8 \, .
\end{align}
If we can now show that this quantity is always an integer, Dai-Freed anomalies will vanish on all such factorized backgrounds. This result does not hinge on the the precise bordism groups computed in the preceding section.

In order to prove that eq. \eqref{eq:intX8} is always an integer, we can connect it with the anomaly inflow mechanism on a fivebrane. Specifically, $S^3_H$ is a non-trivial bordism class, and one possible boundary for it in string theory is a fivebrane. The fivebrane is a codimension four object, and it is characterized precisely by the fact that the angular $S^3$ in the transverse space is threaded by one unit of H-flux. These fivebranes are precisely D5-branes in the orientifold models and NS5-branes in the heterotic model. We will now show that $X_8$ coincides with the anomaly polynomial of such a fivebrane, up to terms which vanish when the Bianchi identity holds. In the presence of fivebranes coupling to $B_6$, the dual of $B_2$, the classical gauge variation of the effective action is compensated by the quantum anomaly of the chiral worldvolume degrees of freedom. For this inflow mechanism to work, the anomaly polynomial of a single fivebrane has to be $I_8 = X_8$ (up to terms that vanish on a twisted String manifold). To see this, notice that the bulk action receives additional worldvolume contributions of the form
\begin{eqaed}\label{eq:brane_action}
    S = S_\text{bulk} + S_\text{GS} + S_\text{wv} + \mu \int_W B_6 \, ,
\end{eqaed}
where $W$ denotes the worldvolume of the fivebrane(s) and $B_6$ the dual of $B_2$.  The bulk action $S_\text{bulk}$, which describes the ten-dimensional effective (super)gravity theory, is accompanied by the Green-Schwarz term $S_\text{GS}$ of eq.\ \eqref{eq:GSterm10d} to cancel bulk anomalies. The brane is instead described by the worldvolume DBI action $S_\text{wv}$ accompanied by the magnetic coupling to $B_6$, which is the relevant coupling in the following argument. The equation of motion for $B_6$ and the corresponding dual Bianchi identity are
\begin{eqaed}\label{eq:brane_bianchi}
    d \star dB_6 & = \mu \, \delta(W \hookrightarrow M_{10}) \, , \\
    dH_3 & = \mu \, \delta(W \hookrightarrow M_{10}) \, ,
\end{eqaed}
where $H_3 = dB_2$ and the $\delta$ is a distribution-valued four-form that describes the embedding of $W$ in spacetime. Correspondingly, the Bianchi identity for the gauge invariant field strength $H \equiv \widetilde{H}_3 = dB_2 - \omega_\text{CS}$, which ordinarily reads $dH = X_4$, also receives a new localized contribution. Because of this Bianchi identity, there is a new classical contribution to the gauge variations. Using descent, $X_8 = dX_7^{(0)}$, $\delta X_7^{(0)} = dX_6^{(1)}$, one finds
\begin{eqaed}\label{eq:classical_anomaly}
    \delta_\text{new} S_\text{GS} & = - \int_{M_{10}} dB_2 \wedge \delta X_7^{(0)} \\
    & = - \int_{M_{10}} dH_3 X_6^{(1)} \\
    & = - \mu \int_W X_6^{(1)} \, ,
\end{eqaed}
which cancels by inflow provided that the worldvolume theory of the D5-brane has an anomaly polynomial $I_8 = \mu \, X_8$~\cite{Dixon:1992if, Mourad:1997uc}. With our choice of units, the elementary charge $\mu = n_5 \in \Z$ counts the number of fivebranes.

As described in Section \ref{sec:local}, the anomaly polynomial is a sum of indices, given by the APS index theorem for each one of the anomalous degrees of freedom propagating on the fivebrane. As such, we know that $I_8$ is an integer and we can conclude from the previous discussion that $X_8$ must also be an integer. This anomaly inflow argument thus allows one to show that the anomaly \eqref{eq:intX8} always vanishes.

For the orientifold models, this mechanism can be implemented explicitly, since these theories have D5-branes whose worldvolume degrees of freedom are known. We describe this in detail in section \ref{sec:5braneorient}.

On the heterotic side, although the $SO(16)\times SO(16)$ theory is known to have NS5 branes, their worldvolume degrees of freedom are not known and so we have to resort to other arguments to prove that the anomaly \eqref{eq:intX8} vanishes. The proof can be found at the end of Section \ref{sec:so16so16_inflow}. There is, however, a more physical way of understanding why anomalies cancel in the heterotic case: one can show that the anomaly polynomial of $SO(16)\times SO(16)$ can be directly related to that the supersymmetric heterotic theories. Therefore, one use this connection to show that anomalies cancel for $SO(16)\times SO(16)$ by showing that they cancel in the supersymmetric cases. This is done explicitly in Section \ref{sec:so16so16_inflow}.

Finally, since we have proven that anomalies vanish in the heterotic case, we can reverse the anomaly inflow argument above to speculate about the worldvolume degrees of freedom of the NS5 brane. Indeed, we identify what kind of degrees of freedom give rise to the correct anomaly polynomial so as to have $X_8 = I_8$. We do so away from strong coupling effects, in the puffed-up instanton limit of the NS5 brane. This is detailed in section \ref{sec:SO16NS5}.

\subsubsection{\texorpdfstring{$\Sp(16)$ and $U(32)$}{Sp(16) and U(32)}}\label{sec:5braneorient}

Let us begin with the orientifold models. This cancellation of anomalies by inflow was first constructed for the case of $\Spin(32)/\Z_2$ in~\cite{Dixon:1992if,Mourad:1997uc}. The chiral fermions on the worldvolume of the D5-brane consist of one vector multiplet of $\Spin(32)/\Z_2$ and two gauge singlets, such that,
\begin{equation}\label{eq:inflowso32}
   (X_8)_{\Spin(32)/\Z_2}- I_{\Spin(32)/\Z_2} = -\frac{1}{24}p_1 (X_4)_{\Spin(32)/\Z_2}\, ,
\end{equation}
where $(X_8)_{\Spin(32)/\Z_2}$ and $(X_4)_{\Spin(32)/\Z_2}$ can be read off from \eqref{eq:anompolyso32}. This shows how one recovers $I_8 = X_8$ up to a term that vanishes on a twisted string manifold. 

One may wonder where the extra term in \eqref{eq:inflowso32} comes from, even if we know it to vanish on a twisted string manifold. This can be understood as follows; from the perspective of the ten-dimensional supergravity action, a D5-brane amounts to introducing a delta function localized on the brane. The D5-brane gives a localized contribution to the 10d action of the form $B_2 \wedge Y_4 \wedge \delta_4$ where $Y_4$ is some 4-form which in this case is reduces to $Y_4 = -\frac{1}{24}p_1$, expanding the A-roof genus in the Chern-Simons effective worldvolume action. Indeed, using the Bianchi identity for the $H_3$ flux, we see that this term contributes to the anomaly polynomial as:
\begin{equation}
    \int _{Z_{12}}X_4\wedge Y_4 \wedge \delta_4 = \int _ {X_8} X_4 \wedge Y _ 4 \,.
\end{equation}
Therefore, the appearance of the extra term in \eqref{eq:inflowso32} can be traced down to not properly taking into account the delta-function source that corresponds to the localized D5-brane. 

The same mechanism happens in the two non-supersymmetric orientifold models, as was found by~\cite{Dudas:2000sn} along the lines of~\cite{Mourad:1997uc,Imazato:2010qz}. The worldvolume degrees of freedom on D5-branes can be extracted from one-loop open-string amplitudes~\cite{Dudas:2001wd}, and the chiral fermions arrange in the virtual representation
\begin{eqaed}
    \left( \mathbf{\frac{N(N+1)}{2}}, \mathbf{1} \right) - \left( \mathbf{\frac{N(N-1)}{2}}, \mathbf{1} \right) - \left( \mathbf{N}, \mathbf{32} \right)
\end{eqaed}
of $SO(N) \times \Sp(16)$ (for the Sugimoto model\footnote{In this case, a single brane corresponds to $N=2$.}) or $U(N) \times U(32)$ (for the Sagnotti model). In order to compare the anomaly polynomial $I_8$ (without worldvolume gauge field) with the bulk $X_8$, one needs to decompose characteristic classes of the bulk tangent bundle in terms of the worldvolume tangent bundle $TW$ and normal bundle $N$. In detail,
\begin{eqaed}\label{eq:pontryagin_split}
    p_1(TM_{10}) & = p_1(TW) + p_1(N) \, , \\
    p_2(TM_{10}) & = p_2(TW) + p_1(TW) \, p_1(N) + p_2(N) \, , \\
    p_1(N) & = c_1(N)^2 - 2 \, c_2(N) \, , \\
    p_2(N) & = c_2(N)^2 = \chi(N)^2 \, . \\
\end{eqaed}
When the normal bundle of the worldvolume is trivial, one obtains
\begin{eqaed}
    I_8 - X_8 \propto p_1(TM_{10}) \, X_4 \, ,
\end{eqaed}
and therefore the inflow mechanism implies that $X_8$ integrates to an integer on any spin 8-manifold with $X_4 = 0$. When the normal bundle is non-trivial, there are additional contributions to the above expression, proportional to the Euler class of $N$. However, the full brane action also contains another term~\cite{Dixon:1992if, Mourad:1997uc} proportional to $B_2$ rather than $B_6$, which induces another classical variation to be canceled by inflow. As a result, the anomaly polynomial of the fivebrane worldvolume theory is not quite the above $I_8$, but has an additional contribution that cancels the normal bundle terms~\cite{Mourad:1997uc}. In more detail, adding a coupling of the type $\int_W B_2 Y_4$ to the fivebrane worldvolume action contributed a new classical variation to the effective action, which arises by descent from $\Delta I_8 = - \, (X_4 + n_5 \, \chi(N)) Y_4$. Therefore, the full anomaly polynomial of the fivebrane worldvolume ought to be $I_8 = n_5 \, X_8 - \Delta I_8$, again up to terms that vanish on twisted String backgrounds. This additional coupling can be shown to cancel the normal bundle terms in the anomaly~\cite{Mourad:1997uc} (see also~\cite{Witten:1996hc} for a discussion in the context of M-theory).

\subsubsection{\texorpdfstring{$SO(16)\times SO(16)$} {SO(16)xSO(16)}}\label{sec:so16so16_inflow}

For the heterotic model, no such result is available, since the worldvolume degrees of freedom of NS5-branes are not understood without supersymmetry or dualities at one's disposal. However, one can nonetheless express $X_8$ as an index of six-dimensional chiral fields; since index are manifestly integers, this will be enough to establish that anomalies cancel. In order to do so, let us observe that the formal difference of representations of the chiral fermions of the non-supersymmetric heterotic model can be rewritten as\footnote{To our knowledge, this was first explicitly stated in the literature in ~\cite{Schellekens:1986xh}, though we learned from Luis Álvarez-Gaumé that the authors of \cite{Alvarez-Gaume:1986ghj} were also aware of this fact.}
\begin{align}
\begin{split}
    & (\mathbf{128}, \mathbf{1}) +  (\mathbf{1}, \mathbf{128}) - (\mathbf{16}, \mathbf{16}) \\
    & = (\mathbf{128}, \mathbf{1}) +  (\mathbf{1}, \mathbf{128}) + (\mathbf{120}, \mathbf{1}) + (\mathbf{1}, \mathbf{120}) \\
    & - (\mathbf{120}, \mathbf{1}) -(\mathbf{1}, \mathbf{120}) - (\mathbf{16}, \mathbf{16}) \, ,
\end{split}
\end{align}
The matter fields in the first line after the equal correspond precisely to the decomposition of the adjoint of $\mathfrak{e}_8\oplus\mathfrak{e}_8$ into representations of the $\mathfrak{so}_{16}\oplus\mathfrak{so}_{16}$ subalgebra; they are the field content that would arise after giving a vev to an adjoint $\mathfrak{e}_8\oplus\mathfrak{e}_8$ field. Similarly, the fields in the second line are (with reversed chirality) those fields that would arise after adjoint Higgsing from the $\mathfrak{so}(32)$ algebra to its $\mathfrak{so}_{16}\oplus\mathfrak{so}_{16}$ subalgebra. What we are seeing here is that, at a formal level (as far as the chiral spectrum is concerned), the $SO(16)^2$ is equivalent to one copy of the $E_8\times E_8$ string stacked on top of a copy of the $\Spin(32)/\mathbb{Z}_2$ string, with opposite chirality, and Higgsed to a common subgroup with algebra $\mathfrak{so}_{16}\oplus\mathfrak{so}_{16}$. Therefore, we can write, at the level of anomaly polynomials, the equality
\begin{align}\label{eq:e8e8minusso32}
    P^{E_8 \times E_8}_{12}|_{SO(16)^2} - P^{\Spin(32)/\Z_2}_{12}|_{SO(16)^2} = P^{SO(16)^2}_{12}\,, 
\end{align}
where we have merely restricted to $SO(16)^2$ bundles inside of the two groups above. Since each of the supersymmetric string theories are anomaly-free by themselves, the formal linear combination will also be. This argument, which can be carried out at the level of eta invariants etc.\ and not just anomaly polynomials, is yet another proof of the fact that the $SO(16)^2$ theory is anomaly free\footnote{Even with the right global quotient.}, without relying explicitly on bordism calculations. Furthermore, in particular, this holds for the Green-Schwarz terms, which are
\begin{align}
\begin{split}
     & (X_8)_{\Spin(32)/\Z_2}|_{SO(16)^2} -(X_8)_{E_8 \times E_8}|_{SO(16)^2}  \; \\ & =  \frac{1}{24} \left( (c_{\mathbf{16},2}^{(1)})^2 + (c_{\mathbf{16},2}^{(2)})^2 + c_{\mathbf{16},2}^{(1)} \, c_{\mathbf{16},2}^{(2)} - 4 \, c_{\mathbf{16},4}^{(1)} - 4 \, c_{\mathbf{16},4}^{(2)} \right) \;    = (X_8)_{SO(16)^2} \, .\end{split}
\end{align}
It is unclear whether this connection between the non-supersymmetric $SO(16) \times SO(16)$ theory and the supersymmetric theories persists beyond a formal equality at the level of (super)gravity, or whether on the contrary it has a deeper meaning. Some previous work ~\cite{Blum:1997cs, Blum:1997gw} (see also~\cite{Faraggi:2007tj}) identified connections between supersymmetric and non-supersymmetric strings via interpolating models, which are nine-dimensional compactifications recovering either supersymmetric or non-supersymmetric strings in different decompactification limits. In particular, an interpolating model was constructed between the $SO(16) \times SO(16)$ theory and the supersymmetric $\Spin(32)/\Z_2$ theory, matching the worldsheet CFT descriptions and solitons in between the two\footnote{Since the non-supersymmetric theories have NS-NS tadpoles and would-be moduli run in the absence of (large) stabilizing fluxes~\cite{Mourad:2016xbk, Basile:2018irz} and/or spacetime warping~\cite{Dudas:2000ff, Blumenhagen:2000dc}, there may be additional subtleties in understanding the dynamics of this duality.}.

The cancellation of anomalies by fivebrane inflow for the $SO(16)\times SO(16)$ theory thus follows from that of the two supersymmetric heterotic theories. The anomaly inflow in the case of $\Spin(32)/\Z_2$ was discussed above \eqref{eq:inflowso32}. The case of $E_8 \times E_8$ is slightly more involved and we will discuss it now. The anomaly inflow of the NS5-brane was famously discussed in~\cite{Horava:1996ma} where the limit in which an instanton in $E_8 \times E_8$ becomes point-like was matched to the world volume theory of the NS5 brane at strong coupling. For our purposes, we can ignore strong coupling dynamics and focus on matching a 6-dimensional anomaly theory of chiral fermions to $(X_8)_{E_8 \times E_8}$ which can be read off from \eqref{eq:anompolye8e8}. This guarantees that $(X_8)_{E_8 \times E_8}$ is an integer and that local anomalies cancel in 10d. We now detail how this can be done.

One can show that $(X_8)_{E_8 \times E_8}$ can be decomposed as follows:
\begin{equation}\label{eq:X8E8E8}
    (X_8)_{E_8 \times E_8} =  \frac{(c^{(1)}_{\mathbf{16},2}-c^{(2)}_{\mathbf{16},2})^2}{32} + \frac{1}{24} X_4^2 + \mathcal{I}_{\text{SD}} + 2 \,\mathcal{I}_{\text{Dirac}}
\end{equation}
where $\mathcal{I}_{\text{SD}} + 2 \, \mathcal{I}_{\text{Dirac}}$ is the index of a self-dual form field and 2 fermion singlets in 8 dimensions. The index of a self-dual form field in 8 dimensions can be shown to be an integer over 8 \cite{Tachikawa:2021mby}. Indeed, it can be written in terms of the signature of the 8-manifold as follows \cite{Hsieh:2020jpj}:
\begin{equation}\label{eq:isd}
   \mathcal{I}_{\text{SD}} = -\frac{\sigma}{8}\,.
\end{equation}
On the other hand, the index of chiral fermions is always an integer. In order to simplify the first term in \eqref{eq:X8E8E8}, we can rewrite the Chern classes in an embedded $SU(2)$ subgroup of each $E_8$, which are known to be integer-valued. Therefore, on a twisted string manifold, $X_8$ reduces to:
\begin{equation}\label{eq:X8E8E82}
    X_8 =  \frac{( c^{(1)}_{\mathbf{2},2}- c^{(2)}_{\mathbf{2},2})^2}{8}  -\frac{\sigma}{8} + n\;\;\;\text{with}\;\; n \in \mathbb{Z}
\end{equation}
where $\frac{1}{2} c^{(i)}_{\mathbf{16},2} \xrightarrow{}  c^{(i)}_{\mathbf{2},2}  $ are the 2nd Chern classes in the fundamental of the SU(2) subgroup of the i-th $E_8$. The Bianchi identity $(X_4)_{E_8 \times E_8}=0$ gives us 
\begin{equation}
     c^{(2)}_{\mathbf{2},2}= -\frac{p_1}{2} -  c^{(1)}_{\mathbf{2},2}\,.
\end{equation}
Plugging this into \eqref{eq:X8E8E82}, we see that  the condition for anomalies to vanish comes down to showing that the following quantity is an integer: 
\begin{equation}\label{eq:E8E8proof}
     \frac{ ( c^{(1)}_{\mathbf{2},2})^2}{2} + \frac{ c^{(1)}_{\mathbf{2},2}  \;p_1}{4}+ \frac{p_1^2}{32}  -\frac{\sigma}{8}\,.
\end{equation}
As it happens, it was shown in~\cite{Hsieh:2020jpj} that the last two terms give (28 times) an integer. Indeed, one can show that:
\begin{equation}\label{eq:diractops}
   28 \; \mathcal{I}_{\text{Dirac}}=  \frac{p_1^2}{32}  -\frac{\sigma}{8}\, .
\end{equation}
Now, to show that the first two terms of \eqref{eq:E8E8proof} are an integer, one can note that on a twisted string manifold, $\frac{p_1}{2}$ is a characteristic vector of $H^4(X;\mathbb{R})$. This, in particular, means that:
\begin{equation}
    \frac{p_1}{2}c^{(i)}_{\mathbf{2},2} \mod{2} \;= \;(c^{(i)}_{\mathbf{2},2})^2 \mod{2}\;.
\end{equation}Therefore we have shown that on a twisted string manifold, $X_8$ is always an integer; and so there can never be an anomaly.

Given that the $X_8$ of $SO(16)\times SO(16)$ is a linear combination of those of $E_8\times E_8$ and $\mathit{Spin}(32)/\Z_2$, we can infer that $ (X_8)_{SO(16)^2}$ is an integer and so that all anomalies vanish for this non-supersymmetric theory. Nevertheless, for completeness, let us detail explicitly how $ (X_8)_{SO(16)^2}$ can be proven to be an integer. One can write $(X_8)_{SO(16)^2}$ as follows: 
\begin{align} \label{eq:X8so16}
  (X_8)_{SO(16)^2} & =   -\frac{1}{32}(c^{(1)}_{\mathbf{16}, 2}-c^{(2)}_{\mathbf{16}, 2})^2 - \mathcal{I}_{\text{SD}} - 4\, \mathcal{I}_{\text{Dirac}} \\ \; &-\frac{1}{48}(X_4)_{SO(16)^2} (c^{(1)}_{\mathbf{16}, 2}+c^{(2)}_{\mathbf{16}, 2}+ 3 p_1) + \mathcal{I}^{\mathbf{16}_{(1)}}_{\text{Dirac}}+ \mathcal{I}^{\mathbf{16}_{(2)}}_{\text{Dirac}}\nonumber
\end{align}
where $\mathcal{I}^{\mathbf{16}_{(i)}}_{\text{Dirac}}$ is the contribution of a fermion that transforms in the $\mathbf{16}$ of $SO(16)_i$, which is known to be integer-valued. Therefore, on a twisted string manifold, the cancellation of anomalies comes down to showing that the following quantity is an integer:
\begin{equation}
    -\frac{1}{32}(c^{(1)}_{\mathbf{16}, 2}-c^{(2)}_{\mathbf{16}, 2})^2 - \mathcal{I}_{\text{SD}} = -\frac{1}{32}(c^{(1)}_{\mathbf{16}, 2}-c^{(2)}_{\mathbf{16}, 2})^2 + \frac{\sigma}{8}
\end{equation}
Given that one can put the 2nd Chern classes in the SU(2) subgroup of $SO(16)$ as $c^{(i)}_{\mathbf{16}, 2} \to 2  c ^{(i)}_{\mathbf{2}, 2}$; the proof goes exactly as in the $E_8 \times E_8$ case.

One can sometimes read-off the chiral field content of a theory from the anomaly polynomial. For instance, for the $E_8 \times E_8$ case, the anomaly polynomial \eqref{eq:X8E8E8} suggests that the chiral field content of the NS5 brane is a self-dual form field and 2 fermion singlets. There are no chiral fields charged under the gauge group, since all the gauge-dependent parts of \eqref{eq:X8E8E8} are in the factorized piece. As it happens, this exactly the chiral field content of a 6d $(1,0)$ tensor multiplet, which is precisely the worldvolume field content of the NS5 brane in $E_8 \times E_8$ string theory. This answer is essentially determined by anomalies together with supersymmetry. In the non-supersymmetric case of the $SO(16)^2$ string, reading off the chiral field content from \eqref{eq:X8so16} in the same way suggests that the chiral field content of the $SO(16)^2$ NS5 brane is:
\begin{itemize}
    \item Four fermion singlets,
    \item A fermion transforming in the $(\mathbf{16},\mathbf{1})\oplus(\mathbf{1},\mathbf{16})$ of $SO(16)$,
    \item A self-dual 2-form field.
\end{itemize}
There are some subtleties in assessing whether or not these are truly the chiral degrees of freedom propagating on this non-supersymmetric brane. First of all, there is no supersymmetry to constrain the worldvolume theory of the NS5 brane which can therefore carry any kind of chiral degrees of freedom. As we have seen from \eqref{eq:diractops} and \eqref{eq:isd}, indices can sometimes be exchanged for one another and yet give the same integer. This means that the $X_8$ does not completely fix the worldvolume content of the NS5 brane, and that any chiral field content with the same anomaly as the one proposed above remains a possibility. Another reason why we cannot be sure that \eqref{eq:X8so16} correctly describes the degrees of freedom propagating on the NS5 brane is that we cannot be sure that an NS5 brane (understood as a small instanton where the full spacetime gauge group symmetry gets restored) exists to begin with. Unlike in the supersymmetric case, in general we expect that the size modulus of the instanton, being non-supersymmetric, receives a potential due to quantum effects that may lead to the small instanton limit being obstructed. The study of the strong coupling effects near the small instanton limit is beyond the validity of effective field theory, and thus beyond the scope of this paper (although it may be amenable to a version of the constructions in \cite{Kaidi:2023tqo}), but we point out that, if the limit does exist and the small instanton transition does survive, the transition point would be a natural place to look for a non-supersymmetric interacting CFT, a cousin of the $E_8$  SCFT. It would be interesting to explore this further. On the other hand, studying the anomaly inflow on the worldvolume of the puffed-up NS5 brane instanton is accessible within the effective field theory (see fig.\ref{fig:puffed-up_brane}). We do this explicitly in the next section.

\begin{figure}[ht!]
    \centering
    \includegraphics[scale=0.25]{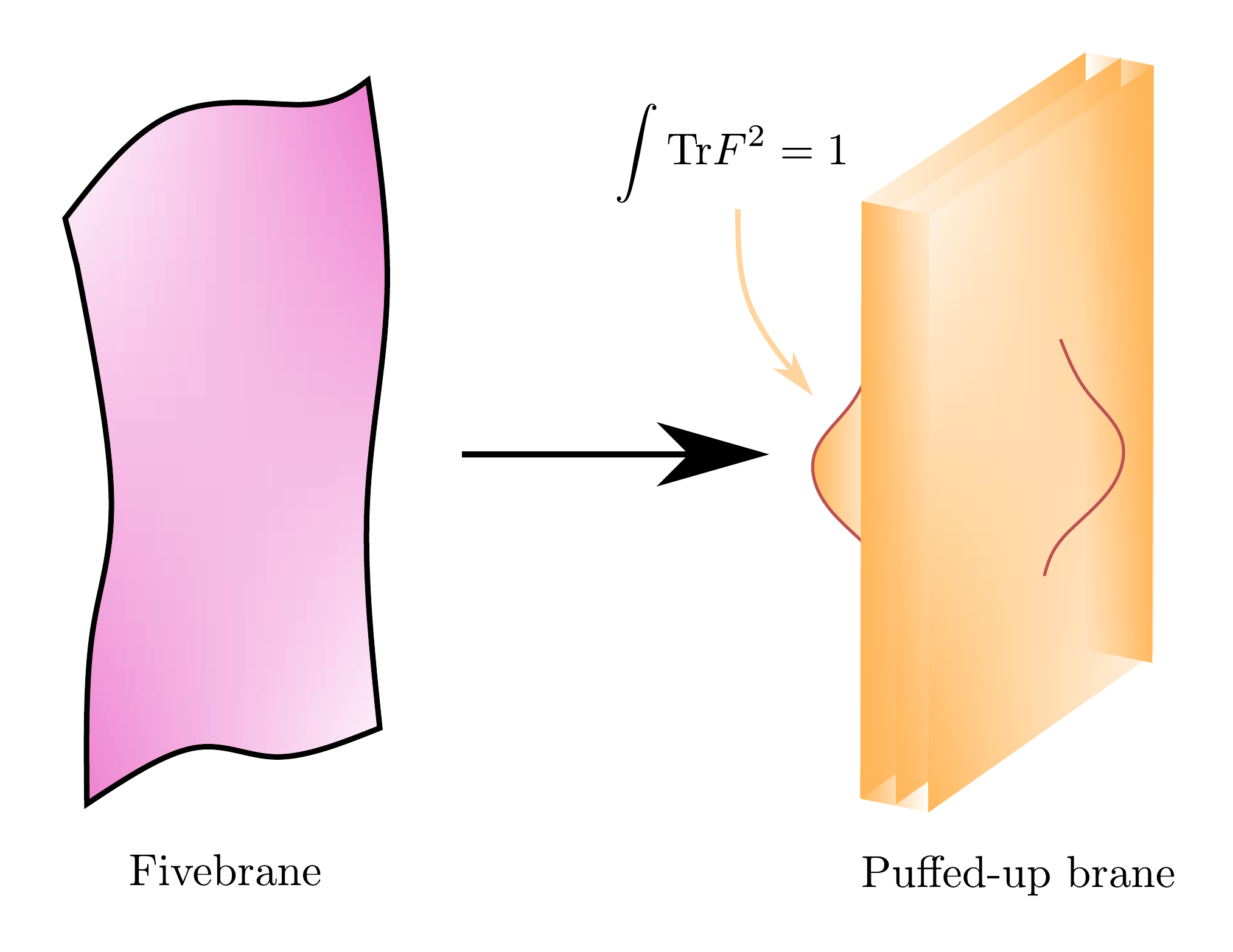}
    \caption{A sketch of a fivebrane puffing up into an instanton which can be described within the effective field theory.}
    \label{fig:puffed-up_brane}
\end{figure}

\subsubsection{Anomaly inflow on puffed-up fivebrane instantons}\label{sec:SO16NS5}

The above result shows that there are no Dai-Freed anomalies on factorized backgrounds of the form $S^3_H \times M_8$, since $X_8$ integrates to an integer. Anomaly inflow on fivebranes dictates that $X_8$ be the anomaly polynomial associated to the worldvolume theory on a single fivebrane, possibly up to terms that vanish when the Bianchi identity is satisfied. As explained above, without direct access to the relevant degrees of freedom on the fivebrane worldvolume, studying the anomaly inflow on the point-like NS5 brane is impossible. Luckily, one can still examine the anomaly inflow on the worldvolume of puffed-up fivebrane instantons, which can be described in the low-energy approximation. Puffing-up the fivebrane corresponds to delocalizing it along its transverse dimensions, making it look like a four dimensional gauge instanton.

Introducing an instanton Higgses one of the $SO(16)$ factors, say the first $SO(16)^{(1)}$, according to $SO(16) \to SU(2)^{(a)} \times SU(2)^{(b)} \times SO(12)$, so that the vector and spinor representations branch into

\begin{subequations}
\begin{align}\mathbf{16}&= (\mathbf{1^{(a)},1^{(b)},12})+ (\mathbf{2^{(a)},2^{(b)},1}) \, ,\\
\mathbf{128}&=(\mathbf{2^{(a)},1^{(b)},32})+ (\mathbf{1^{(a)},2^{(b)},\bar{32}}) \, .\end{align}\end{subequations}

If the instanton bundle only involves $SU(2)^{(a)}$, there an $SU(2)^{(b)}\times SO(12)\times SO(16)^{(2)}$ unbroken symmetry and the background has fermion zero modes (fzm) arising from the representations $(\mathbf{8_s,16^{(1)},16^{(2)}}) $ and $ (\mathbf{8_c,128^{(1)},1^{(2)}})$ of the spacetime isometries and the original gauge group. As a result, one has
\begin{equation}
    \begin{cases}
        1 \text{ fzm} \; \text{in the rep}\; (\mathbf{8_s,2^{(b)},1,16^{(2)}}) \, ,\\
        1 \text{ fzm} \;  \text{in the rep}\; (\mathbf{8_c,1^{(b)},32,1^{(2)}})\,.\\
    \end{cases}
\end{equation}
The two types of fermion zero modes have different chirality, and thus the corresponding worldvolume anomaly polynomial reads\begin{equation}
    P_8 = \frac{1}{2}[\hat A(R)(\text{ch}(F)_{(\mathbf{2^{(b)},1,16^{(2)}})}- \text{ch}(F)_{(\mathbf{1^{(b)},32,1^{(2)}})})]_8 \, ,
\end{equation} 
which evaluates to
\begin{equation}\begin{gathered}\label{eq:P8SO16}
    P_8= \frac{1}{24} \big[-2 p_1 c_{\mathbf{12},2}+p_1 (c^{(2)}_{\mathbf{16},2}+8 c^{(b)}_{\mathbf{2},2})-c_{\mathbf{12},2}^2-4 c_{\mathbf{12},4}\\+2 (c^{(2)}_{\mathbf{16},2}+2 c^{(b)}_{\mathbf{2},2}) (c^{(2)}_{\mathbf{16},2}+4 c^{(b)}_{\mathbf{2},2})-4 c^{(2)}_{\mathbf{16},4}\big] \, .
\end{gathered}\end{equation}
The next step is to evaluate $X_8$ on this background. This amounts to decomposing characteristic classes according to the branching rules, and one finds
\begin{equation}\label{eq:rulesSO161}
    \begin{gathered}
        c^{(1)}_{\mathbf{16},2}\to c_{\mathbf{12},2}+2 c^{(a)}_{\mathbf{2},2}+2 c^{(b)}_{\mathbf{2},2} \, , \\
        c^{(1)}_{\mathbf{16},4}\to 2 c_{\mathbf{12},2} c^{(a)}_{\mathbf{2},2}+2 c_{\mathbf{12},2} c^{(b)}_{\mathbf{2},2}+c_{\mathbf{12},4}-2 c^{(a)}_{\mathbf{2},2} c^{(b)}_{\mathbf{2},2}+(c^{(a)}_{\mathbf{2},2})^2+(c^{(b)}_{\mathbf{2},2})^2 \, .
    \end{gathered}
\end{equation}
Finally, $(c^{(a)}_{\mathbf{2},2})^2$ should be replaced by zero, since it is proportional to the square of the worldvolume current $\delta(W)$ of the fivebrane. All in all, when the dust settles one arrives at
\begin{equation}\label{X8P8}
    X_8 - P_8 = \frac{1}{24} \left(2 c_{\mathbf{12}, 2}-c^{(2)}_{\mathbf{16}, 2}- 8 c^{(b)}_{\mathbf{2},2}\right) \left(c_{\mathbf{12}, 2}+ 2 c^{(b)}_{\mathbf{2},2}+c^{(2)}_{\mathbf{16}, 2}+p_1\right) \, ,
\end{equation}
where the second factor corresponds to $X_4$ for the unbroken piece of the gauge group. The inflow therefore works when $X_8$ and $P_8$ are equal on manifolds where the Bianchi identity holds. A similar argument works for more general choices of instanton bundles.

For the orientifold models, one expects the small limit of the ``fat'' fivebrane instantons to yield the worldvolume degrees of freedom of D5-branes. This is a nice crosscheck that we detail now. For the Sugimoto model (the calculation is identical in the Sagnotti model), the anomaly polynomial $P_8$ associated to the fermion zero modes of the instanton is
\begin{equation}\label{eq:P8SUgi}
    P_8 = [\hat A(R) \text{ch}(F)_{\mathbf{30}}]_8 = \frac{1}{192} \left(8 p_1 c_{\mathbf{30},2 }-4 \left(8 c_{ \mathbf{30},4}+p_2\right)+16 c_{\mathbf{30},2}^2+7 p_1^2\right) \, ,
\end{equation}
since under the branching $\Sp(16) \to \SU(2) \times \USp(30)$ the adjoint representation, containing the gauginos, decomposes according to
\[\mathbf{495}=(\mathbf{2,30})+(\mathbf{1,434})+(\mathbf{1,1})\]
where the only charged contribution comes from the first term on the right-hand side. In the small limit, the $SU(2)$ Chern classes vanish and the remaining ones are enhanced to $\Sp(16)$ classes, ending up with
\begin{equation}
     P_8^\text{small} =  \frac{1}{192} \left(8 p_1 c_{\mathbf{32},2 }-4 \left(8 c_{ \mathbf{32},4}+p_2\right)+16 c_{\mathbf{32},2}^2+7 p_1^2\right) = X_8 + \frac{1}{24} \, X_4 \, p_1 ,
\end{equation}
Thus reproducing the anomaly polynomial of a D5-brane worldvolume up to terms that vanish on the allowed backgrounds.

\section{Anomalies and bordism for the swap \texorpdfstring{$\Z_2$}{Z2} action}\label{sec:swapping}
\subsection{Overview and the bordism computation}
The $E_8\times E_8$ heterotic string theory has a $\Z_2$ symmetry given by swapping the two copies of $E_8$, so it
is possible to expand the gauge group of the theory to $(E_8\times E_8)\rtimes\Z_2$. To our knowledge, this fact
first appears in~\cite[\S I]{McI99} (see also~\cite[\S 2.1.1]{dBDH00}). The question of anomaly cancellation for
this string theory is completely different in the absence versus in the presence of this extra $\Z_2$: without it,
the anomaly is known to vanish, as Witten~\cite[\S 4]{Wit86} showed it is characterized by a bordism invariant
$\Omega_{11}^\Spin(BE_8)\to\C^\times$, and Stong~\cite{Sto86} showed $\Omega_{11}^\Spin(BE_8)\cong 0$. But with the
$\Z_2$ swapping symmetry turned on, the relevant bordism group has order $64$~\cite[Theorem 2.62]{Deb23} courtesy
of a harder computation; even though we cannot determine this group exactly, we will show that the anomaly vanishes, in accordance with the results in \cite{Yam23b} obtained from a worldsheet perspective.

In this section, we discuss a closely analogous story for the $\Spin(16)\times\Spin(16)$ non-supersymmetric
heterotic string. The gauge group $\Spin(16)\times_{\Z_2}\Spin(16)$ (where the diagonal $\Z_2$ we quotient by
corresponds to either of the subgroups in each $\Spin(16)$ whose quotient is \emph{not} $\mathit{SO}(16)$) admits a $\Z_2$
automorphism switching the two $\Spin(16)$ factors, enlarging the gauge group of this theory to
$(\Spin(16)\times_{\Z_2}\Spin(16))\rtimes\Z_2$; see~\cite[\S III]{McI99}.

In this paper, we chose to work with $\Spin(16)\times\Spin(16)$, which simplifies the bordism computations at the
expense of applying to only some backgrounds. The $\Z_2$ symmetry enlarges the gauge group to $G_{16,16} :=
(\Spin(16)\times\Spin(16))\rtimes\Z_2$. The Green-Schwarz mechanism is analogous: if $x\in H^*(B\Spin(16); A)$ for
some coefficient group $A$, let $x^L$ and $x^R$ denote the copies of $x$ in $H^*(B(\Spin(16)\times\Spin(16)); A)$
coming from the first, resp.\ second copies of $\Spin(16)$ via the Künneth formula. Then the class $\tfrac 12p_1^L
+ \tfrac 12 p_1^R$, which was the characteristic class of the Green-Schwarz mechanism in the absence of the $\Z_2$
symmetry, descends through the Serre spectral sequence for the fibration
\begin{equation}\label{G1616fib}\begin{tikzcd}
	{B(\Spin(16)\times\Spin(16))} & {BG_{16,16}} \\
	& {B\Z_2}
	\arrow[from=1-1, to=1-2]
	\arrow[from=1-2, to=2-2]
\end{tikzcd}\end{equation}
to define a class in $H^*(BG_{16,16};\Z)$, and the Green-Schwarz mechanism asks, on a spin manifold $M$ with a
principal $G_{16,16}$-bundle $P\to M$, for a trivialization of
\begin{equation}
\label{Z2_GS}
	\tfrac 12 p_1(M) - (\tfrac 12 p_1^L + \tfrac 12 p_1^R)(P).
\end{equation}
Let $\mathbb G_{16,16}$ denote the Lie $2$-group corresponding to this data, i.e.\ the string cover of $\Spin\times
G_{16,16}$ corresponding to the class~\eqref{Z2_GS}. Quotienting by $G_{16,16}$ defines a map to $\Spin$ and
therefore a tangential structure in the usual way; a $\mathbb G_{16,16}$-structure on a vector bundle $E\to M$ is a
spin structure on $E$, a double cover $\pi\colon M'\to M$, a pair of rank-$16$ spin vector bundles $V^L$ and $V^R$
on $M'$ identified under the deck transformation of $M'$, and a trivialization of $\tfrac 12 p_1(E) - (\tfrac 12
p_1(V^L) - \tfrac 12 p_1(V^R))$ (the class $\tfrac 12 p_1(V^L) + \tfrac 12 p_1(V^R)$ descends from $M'$ to $M$). If
the double cover $M'\to M$ is trivial, this is equivalent to a $\Spin\text{-}\Spin(16)^2$ structure as defined in \S\ref{SO16_no_swap}.
\begin{thm}
\begin{alignat*}{2}
	\Omega_0^{\mathbb G_{16,16}} &\cong \Z \qquad\qquad &\Omega_6^{\mathbb G_{16,16}} &\cong \Z_2\\
	\Omega_1^{\mathbb G_{16,16}} &\cong (\Z_2)^{\oplus 2} &\Omega_7^{\mathbb G_{16,16}} &\cong \Z_{16}\\
	\Omega_2^{\mathbb G_{16,16}} &\cong (\Z_2)^{\oplus 2} &\Omega_8^{\mathbb G_{16,16}} &\cong \Z^{\oplus 3} \oplus
		(\Z_2)^{\oplus i}\\
	\Omega_3^{\mathbb G_{16,16}} &\cong \Z_8 &\Omega_9^{\mathbb G_{16,16}} &\cong (\Z_2)^{\oplus j}\\
	\Omega_4^{\mathbb G_{16,16}} &\cong \Z\oplus\Z_2 &\Omega_{10}^{\mathbb G_{16,16}} &\cong (\Z_2)^{\oplus k}\\
	\Omega_5^{\mathbb G_{16,16}} &\cong 0 &\Omega_{11}^{\mathbb G_{16,16}} &\cong A,
\end{alignat*}
where either $i = 1$, $j = 4$, and $k = 4$, or $i = 2$, $j= 6$, and $k = 5$, and $A$ is an abelian group of order
$64$ isomorphic to one of $\Z_8\oplus\Z_8$, $\Z_{16}\oplus\Z_4$, $\Z_{32}\oplus\Z_2$, or $\Z_{64}$.
\end{thm}
The fact that $\Omega_{11}^{\G_{16,16}}\ne 0$ implies that the $\Spin(16)\times\Spin(16)$ heterotic theory with its $\Z_2$ swapping symmetry could have an anomaly; we will nevertheless be able to cancel it later in this Section.
\begin{proof}
The proof is nearly identical to the analogous calculation for the $E_8\times E_8$ heterotic string, which is done
in~\cite[\S 2.2, \S 2.3]{Deb23}; therefore we will be succinct and direct the reader there for the details.

Let $V\to B\Spin(16)\times B\Spin(16)$ be the direct sum of the tautological vector bundles on the two factors. The $\Z_2$ swapping action on $B\Spin(16)\times B\Spin(16)$ lifts to make $V$ into a $\Z_2$-equivariant vector bundle, so $V$ descends to a vector bundle we will also call $V$ over $BG_{16,16}$. Since the action of $\Z_2$ is compatible with the spin structures on the two tautological bundles, $V\to BG_{16,16}$ is spin, so $w_1(V) = 0$ and $w_2(V) = 0$; and essentially by definition, $\tfrac 12 p_1(V) = \tfrac 12 p_1^L + \tfrac 12 p_1^R$. Therefore just as for the other theories we studied, there is an isomorphism
\begin{equation}
    \Omega_*^{\G_{16,16}}\overset\cong\to \Omega_*^\String((BG_{16,16})^{V - 32}).
\end{equation}
This is the biggest difference between the computations for the $\Spin(16)\times\Spin(16)$ and $E_8\times E_8$ theories: see~\cite[Lemma 2.2]{Deb23}. Much of the theory developed in~\cite[\S 2]{Deb23} and in~\cite{DY23} and applied to the $E_8\times E_8$ theory in \textit{loc.\ cit.} can therefore be avoided for the $\Spin(16)\times\Spin(16)$ case; nevertheless, the calculation is pretty similar.

First we must establish the absence of $p$-torsion for primes $p > 3$. This is analogous to the other twisted
string bordism computations in this paper, and we do not go into detail.

At $p = 3$, we follow~\cite[\S 3.2]{DY23}. First we need $H^*(BG_{16,16};\Z_3)$; the Serre spectral sequence for $\Z_3$ cohomology and the fibration~\eqref{G1616fib} collapses to an isomorphism
\begin{equation}
    H^*(BG_{16,16};\Z_3) \overset\cong\longrightarrow H^*(B(\Spin(16)\times \Spin(16));\Z_3)^{\Z_2}.
\end{equation}
In the degrees relevant to us, $H^*(B\Spin(16);\Z_3)$ is generated by the Pontrjagin classes $p_1$ and $p_2$ with no relations in degrees $11$ and below, so we obtain the following additive basis for $H^*(BG_{16,16};\Z_3)$ in degrees $11$ and below: $1$, $p_1^L + p_1^R$, $(p_1^L)^2 + (p_1^R)^2$, $p_2^L = p_2^R$, and $p_1^Lp_1^R$. Using this, we determine the $\cA^{\tmf}$-module structure on $H^*((BG_{16,6})^{V - 32};\Z_3)$ using~\cite[Corollary 2.37]{DY23}: if $U$ denotes the Thom class, $\beta(U) = 0$ and $\cP^1(U) = -U(p_1^L + p_1^R)$ (as $\tfrac 12 x = -x$ in a $\Z_3$-vector space). Using this and the Cartan formula, we find an $\cA^{\tmf}$-module isomorphism
\begin{equation}
\label{swap_SO16_mod3}
    H^*((BG_{16,16})^{V -32};\Z_3) \cong \textcolor{BrickRed}{N_3} \oplus \textcolor{Green}{\Sigma^8 N_3} \oplus \textcolor{MidnightBlue}{\Sigma^8 N_3} \oplus P,
\end{equation}
where $N_3$ is as in \cref{N3_defn} and $P$ is concentrated in degrees $12$ and above, and will be irrelevant for us. We draw~\eqref{swap_SO16_mod3} in \cref{SO_16_swap_3primary_figure}, left. Using the calculation of $\Ext_{\cA^{\tmf}}(N_3)$ from \cref{N2_fig}, we can draw the $E_2$-page of the Adams spectral sequence in \cref{SO_16_swap_3primary_figure}, right; it collapses to show there is no $3$-torsion in degrees $11$ and below.
\begin{figure}[!htbp]
\centering
\begin{subfigure}[c]{0.35\textwidth}
  \begin{tikzpicture}[scale=0.6, every node/.style = {font=\tiny}]
    \foreach \y in {0, 4, ..., 16} {
      \node at (-2, \y/2) {$\y$};
    }
    \begin{scope}[BrickRed]
      \tikzpt{0}{0}{$U$}{};
      \tikzpt{0}{2}{}{};
	  \tikzpt{0}{4}{}{};
      \PoneL(0, 0);
      \PoneL(0, 2);
    \end{scope}
    \begin{scope}[Green]
      \tikzptB{2}{4}{$U(p_2^L + p_2^R)$}{};
      \tikzpt{2}{6}{}{};
      \tikzpt{2}{8}{}{};
      \PoneL(2, 4);
	  \PoneL(2, 6);
    \end{scope}
    \begin{scope}[MidnightBlue]
      \tikzptB{4}{4}{$Up_1^Lp_1^R$}{};
      \tikzpt{4}{6}{}{};
      \tikzpt{4}{8}{}{};
      \PoneL(4, 4);
	  \PoneL(4, 6);
    \end{scope}
  \end{tikzpicture}
\end{subfigure}
\begin{subfigure}[!htbp]{0.6\textwidth}
\includegraphics[width=\textwidth]{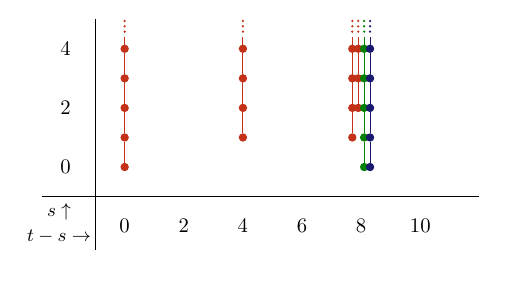}
\end{subfigure}
\caption{Left: the $\cA^{\tmf}$-module structure on $H^*((BG_{16,16})^{V-32};\Z_3)$ in low degrees; the pictured submodule contains
all elements in degrees $11$ and below. Right: the $E_2$-page of the Adams spectral sequence computing
$\tmf_*((BG_{16,16})^{V-32})_3^\wedge$.}
\label{SO_16_swap_3primary_figure}
\end{figure}
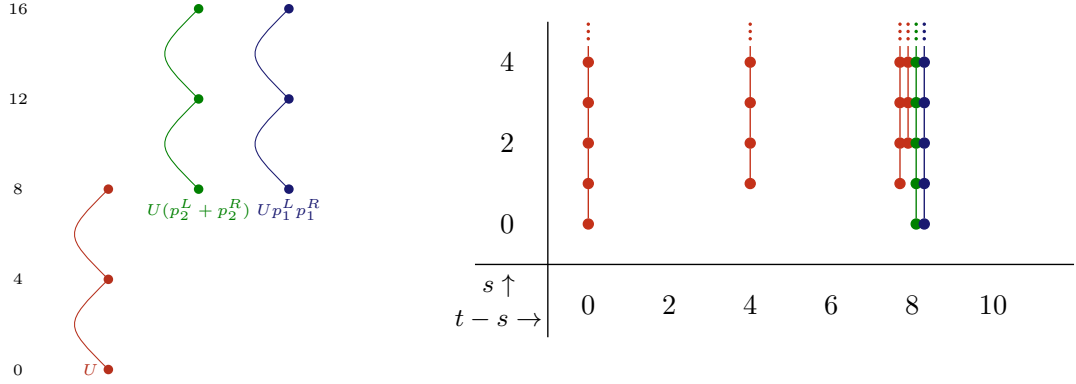

Finally $p = 2$. First, we need $H^*(BG_{16,16};\Z_2)$; Evens' generalization~\cite{Eve65} of a theorem of
Nakaoka~\cite[Theorem 3.3]{Nak61} gives us the following additive basis for these cohomology groups in degrees $13$
and below:
\begin{itemize}
	\item classes of the form $c^L + c^R$, where $c$ ranges over a basis of $H^*(B\Spin(16);\Z_2)$ in degrees $13$
	and below;
	\item the classes $w_4^Lw_4^R$, $w_6^Lw_6^R$, $w_4^Lw_k^R + w_k^Lw_4^R$ for $k = 6,7,8$, and $(w_4^L)^2w_4^R +
	w_4^L(w_4^2)^R$; and
	\item finally, we have classes of the form $x^m$, $w_4^Lw_4^R x^m$, and $w_6^Lw_6^R x^m$, where $x$ is the generator of
	$H^1(B\Z_2;\Z_2)$, pulled back by the quotient $G_{16,16}\to\Z_2$ by the normal $\Spin(16)\times\Spin(16)$
	subgroup.
\end{itemize}
Quillen's detection theorem~\cite[Proposition 3.1]{Qui71} computes the $\cA(2)$-action on these classes. Since $V$ has vanishing $w_1$ and $w_2$, but $w_4(V) = w_4^L + w_4^R$, $\Sq^1(U) = 0$,
$\Sq^2(U) = 0$, and $\Sq^4(U) = U(w_4^L + w_4^R)$. Using this, we can obtain $\cA(2)$-module structure on
$H^*((BG_{16,16})^{V - 32};\Z_2)$ by direct computation with the Cartan formula similarly to~\cite[Proposition 2.41]{Deb23}.
\begin{prop}
\label{mathcalM}
Let $\mathcal M$ be the quotient of $H^*((BG_{16,16})^{V - 32};\Z_2)$ by all elements in degrees $14$ and above.
Then $\mathcal M$ is the direct sum of the following submodules.
\begin{enumerate}
	\item $M_1$, the summand containing $U$.
	\item $\textcolor{BrickRed}{M_2} := \widetilde H^*(\RP^\infty;\Z_2)$ (modulo elements in degrees $14$ and
	above).
	\item $\textcolor{RedOrange}{M_3}$, the summand containing $U((w_4^L)^2 + (w_4^R)^2)$.
	\item $\textcolor{Goldenrod!67!black}{M_4}$, the summand containing $Uw_4^Lw_4^R$.
	\item $\textcolor{Green}{M_5}$, the summand containing $Uw_4^Lw_4^Rx$.
	\item $M_6$, the summand containing $U(w_4^Lw_6^L + w_4^Rw_6^R)$.
	\item $\textcolor{PineGreen}{M_7}$, the summand containing $U((w_4^L)^2w_4^R + w_4^L(w_4^R)^2)$.
	\item $\textcolor{MidnightBlue}{M_8}$, the summand containing $U(w_4^Lw_8^L + w_4^Rw_8^R)$.
	\item $\textcolor{Fuchsia}{M_9}$, the summand containing $U(w_4^Lw_8^R + w_8^Lw_4^R)$.
\end{enumerate}
\end{prop}
We draw this decomposition in \cref{A2_swap}.

\begin{figure}[!htbp]
\centering
 \begin{tikzpicture}[scale=0.6, every node/.style = {font=\tiny}]
    \foreach \y in {0, 1, ..., 14} {
      \node at (-2, \y) {$\y$};
    }
    \begin{scope} 
        \clip (-1.5, -1) rectangle (1.5, 13.75);
        \tikzptB{0}{0}{$U$}{};
        \foreach \y in {4, 6, 7, 8, 10, 11, 12, 13, 14} {
            \tikzpt{0}{\y}{}{};
        }
        \sqfourL(0, 0);
        \sqtwoL(0, 4);
        \sqone(0, 6);
        \sqfourL(0, 6);
        \sqfourR(0, 7);
        \sqtwoL(0, 8);
        \sqfourRtwo(0, 8);
        \sqone(0, 10);
        \sqtwoL(0, 11);
        \sqtwoR(0, 12);
        \sqone(0, 12);
        \sqfourL(0, 13);
        \sqone(0, 14); 
    \end{scope}
    \begin{scope}[BrickRed] 
        \clip (0.5, 0) rectangle (3.5, 13.75);
        \tikzptB{2}{1}{$Ux$}{};
        \foreach \y in {2, ..., 14} {
            \tikzpt{2}{\y}{}{};
        }
        \foreach \y in {1, 3, ..., 13} {
            \sqone(2, \y);
        }
        \foreach \y in {2, 6, 10, 14} {
            \sqtwoL(2, \y);
        }
        \foreach \y in {3, 7, 11} {
            \sqtwoR(2, \y);
        }
        \sqfourL(2, 4);
        \sqfourLtwo(2, 5);
        \sqfourR(2, 6);
        \sqfourRtwo(2, 7);

        \sqfourL(2, 12);
        \sqfourLtwo(2, 13);
        \sqfourR(2, 14);
    \end{scope}
    \begin{scope}[RedOrange]
        \clip (2.5, 0) rectangle (7, 13.75);
        \tikzptB{4.5}{8}{$U((w_4^L)^2 + (w_4^R)^2)$}{};
        \foreach \y in {12, 14} {
            \tikzpt{4.5}{\y}{}{};
        }
        \sqfourL(4.5, 8);
        \sqfourL(4.5, 14);
        \sqtwoL(4.5, 12);
        \sqone(4.5, 14);
    \end{scope}
    \begin{scope}[Goldenrod!67!black]
        \clip (4.25, 0) rectangle (9, 13.75);
        \tikzpt{6.75}{8}{$Uw_4^Lw_4^R$}{};
        \foreach \y in {10, 11, ..., 14} {
            \tikzpt{6.75}{\y}{}{};
        }
        \sqtwoL(6.75, 8);
        \sqfourR(6.75, 8);
        \sqone(6.75, 10); 
        \sqtwoL(6.75, 11);
        \sqone(6.75, 12);
        \sqtwoR(6.75, 12);
    \end{scope}
    \begin{scope}[Green]
        \clip (6, 0) rectangle (11, 13.75);
        \tikzptB{8.75}{9}{$Uw_4^Lw_4^Rx$}{};
        \foreach \y in {10, 11, ..., 14} {
            \tikzpt{8.75}{\y}{}{};
        }
        \tikzpt{10.25}{13}{}{};
        \tikzpt{10.25}{14}{}{};
        \sqone(8.75, 9);
        \sqfourL(8.75, 9);
        \sqtwoL(8.75, 10);
        \sqfourLtwo(8.75, 10);
        \sqone(8.75, 11);
        \sqtwoCR(8.75, 11);
        \sqfourR(8.75, 11);
        \sqone(8.75, 13);
        \sqone(10.25, 13);
        \sqfourL(8.75, 13);
        \sqfourRtwo(8.75, 12);
        \sqfourL(10.25, 13);
        \sqtwoL(10.25, 14);
    \end{scope}
    \begin{scope}
        \clip (6, 0) rectangle (16, 13.75);
        \tikzptB{11.5}{10}{$U(w_4^Lw_6^L + w_4^Rw_6^R)$}{};
        \tikzpt{11.5}{11}{}{};
        \tikzpt{11.5}{12}{}{};
        \tikzpt{11.5}{13}{}{};
        \tikzpt{11.5}{14}{}{};
                
        \sqone(11.5, 10);
        \sqtwoL(11.5, 10);
        \sqtwoR(11.5, 11);
        \sqtwoL(11.5, 12);
        \sqfourRtwo(11.5, 12);
        \sqfourL(11.5, 10);
        \sqfourR(11.5, 11);
        \sqone(11.5, 13);
        \sqfourLtwo(11.5, 13);
    \end{scope}
    \begin{scope}[PineGreen]
            \clip (6, 0) rectangle (16, 13.75);
        \tikzpt{13.5}{12}{$U\alpha$}{};
        \tikzpt{13.5}{14}{}{};
        \sqtwoL(13.5, 12);
    \end{scope}
    \begin{scope}[MidnightBlue]
            \clip (6, 0) rectangle (18, 13.75);
        \tikzptB{15}{12}{$U(w_4^Lw_8^L + w_4^Rw_8^R)$}{};
        \tikzpt{15}{14}{}{};
        \sqtwoL(15, 12);
    \end{scope}
    \begin{scope}[Fuchsia]
            \clip (6, 0) rectangle (22, 13.75);
        \tikzptR{16.5}{12}{$U(w_4^Lw_8^R + w_8^Lw_4^R)$}{};
        \tikzpt{16.5}{14}{}{};
        \sqtwoL(16.5, 12);
    \end{scope}
  \end{tikzpicture}
\caption{The $\cA(2)$-module structure on $H^*(B((\Spin(16)\times\Spin(16))\rtimes\Z_2)^{V - 32};\Z_2)$ in low degrees. The figure includes all classes in degrees $13$ and below.
Here $\alpha := (w_4^L)^2w_4^R + w_4^L(w_4^R)^2)$.}
\label{A2_swap}
\end{figure}
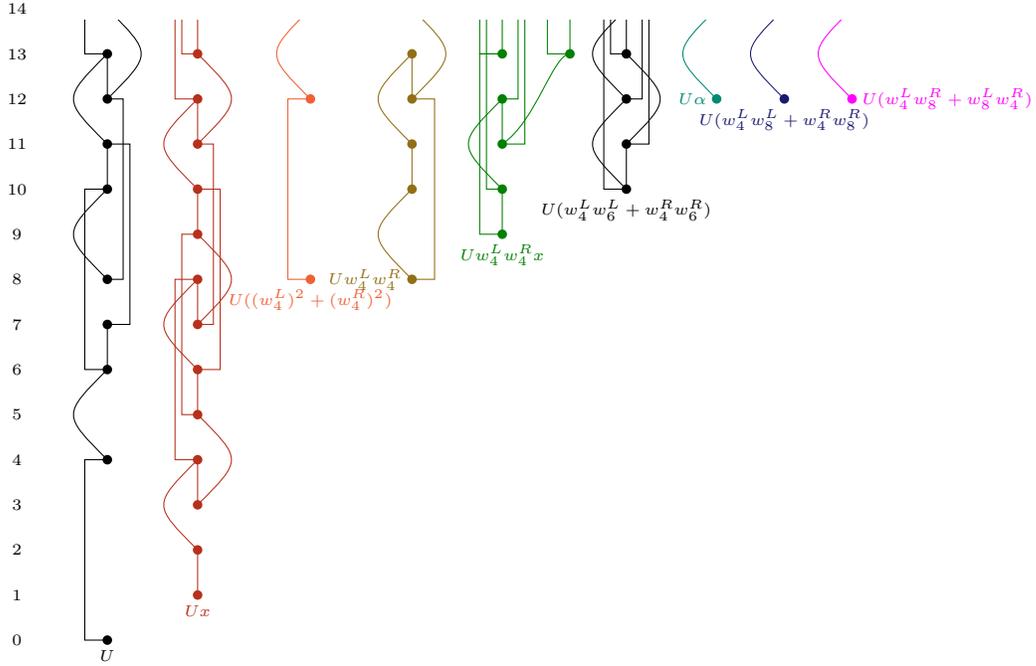

The next step is to split off some of these summands in a manner similar to~\cite[Corollary 2.36]{Deb23}. Morally this is exactly the same simplification we used in \cref{noswap} and discussed further in \cref{susy_noswap}, but the details are a little more complicated.
\begin{defn}
Let $\xi\colon B\G_{16,16'}\to B\mathit O$ be the tangential structure defined analogously to $B\G_{16,16}$, but with $\Spin$ in place of $\Spin(16)$.
\end{defn}
\begin{lem}
The map $\Spin(16)\hookrightarrow\Spin$ induces a map $\Omega_k^{\G_{16,16}}\to\Omega_k^{\G_{16,16}'}$ which is an isomorphism for $k\le 14$.
\end{lem}
This means that, for our string-theoretic applications, it does not matter whether we use $B\G_{16,16}$ or $B\G_{16,16}'$.
\begin{proof}
We want to show that the map $\mathit{MT\G}_{16,16}\to\mathit{MT\G}'_{16,16}$ of bordism spectra is an isomorphism on $\pi_k$ for $k\le 14$. By the Whitehead theorem we may equivalently use $H^k(\text{--};\Z)$, and by the Thom isomorphism, it suffices to show the map $B\G_{16,16}\to B\G_{16,16}'$ is an isomorphism on $\Z$-cohomology in degrees $14$ and below. The cohomology rings of these spaces can be computed in two steps: first the Serre spectral sequence for the fibration $B(\Spin(16)\times\Spin(16))\to BG_{16,16}\to B\Z_2$, then the Serre spectral sequence for the fibration $B^2\U(1)\to B\G_{16,16}\to BG_{16,16}$; and analogously for $B\G_{16,16}'$ with $\Spin$ in place of $\Spin(16)$. For each of these two steps, the map $\Spin(16)\to\Spin$ induces a map of Serre spectral sequences. and because $H^*(B\Spin;\Z)\to H^*(B\Spin(16);\Z)$ is an isomorphism in degrees $15$ and below, we learn that at each of the two steps, the two spectral sequences are isomorphic in degrees $14$ and below, which implies the map $B\G_{16,16}\to B\G_{16,16}'$ induces an isomorphism on cohomology in degrees $14$ and below.
\end{proof}
\begin{prop}
\label{spin_split}
There is a spectrum $\cQ$ and a splitting
\begin{equation}
\label{string_Z2_spin16_split}
	\mathit{MT\mathbb{G}}_{16,16}'\overset{\simeq}{\longrightarrow} \MTSpin\vee \mathcal Q,
\end{equation}
such that the pullback map on cohomology corresponding to the projection $\mathit{MT\G}_{16,16}'\to\cQ$ is a map
\begin{equation}
    H^*(\cQ;\Z_2)\cong \cA\otimes_{\cA(2)} \mathcal L \longrightarrow
        H^*(\mathit{MT\G}_{16,16}';\Z_2)\cong\cA\otimes_{\cA(2)} H^*((BG_{16,16})^{V-32};\Z_2)
\end{equation}
given by the inclusion of an $\cA(2)$-module $\mathcal L\hookrightarrow H^*((BG_{16,16})^{V-32};\Z_2)$, followed by applying $\cA\otimes_{\cA(2)}\text{--}$; the quotient of $\mathcal L$ by all classes in degrees $14$ and above is isomorphic to
\begin{equation}
\textcolor{BrickRed}{M_2}\oplus \textcolor{RedOrange}{M_3} \oplus \textcolor{Goldenrod!67!black}{M_4}\oplus \textcolor{Green}{M_5}\oplus \textcolor{PineGreen}{M_7}\oplus \textcolor{MidnightBlue}{M_8} \oplus \textcolor{Fuchsia}{M_9}.
\end{equation}
\end{prop}
\begin{proof}
The idea is the same as~\cite[Corollary 2.36]{Deb23}: show that a spin structure induces a $\G_{16,16}'$-structure, such that forgetting back down to $B\Spin$ recovers the original spin structure.

Any spin vector bundle $E\to M$ has a
canonical $\mathbb G_{16,16}'$-structure with a trivial double cover $M' := M\amalg M$, $V^L$ equal to $E$ on one
copy of $M$ inside $M'$ and equal to $0$ on the other copy of $M$, and $V^R$ the image of $V^L$ under the deck
transformation, as $\tfrac 12 p_1(V^L) + \tfrac 12 p_1(V^R)$ (descended to $M$) is canonically identified with
$\tfrac 12 p_1(E)$. Composing with the forgetful map $B\mathbb G_{16,16}'\to B\Spin$ gives a map $B\Spin\to B\Spin$
homotopy equivalent to the identity and therefore maps of spectra $\MTSpin\to
\mathit{MT\mathbb{G}}_{16,16}'\to\MTSpin$, yielding the splitting as promised.

To see the statement on cohomology, one can look at the edge morphism in the Serre spectral sequence for $B^2\U(1)\to B\G_{16,16}'\to BG_{16,16}'$.
\end{proof}
As we already know spin bordism groups in the dimensions we need, we focus on computing $\pi_*(\cQ)_2^\wedge$. Because the cohomology of $\cQ$ is of the form $\cA\otimes_{\cA(2)}\mathcal L$, the change-of-rings theorem simplifies the Adams spectral sequence for $\cQ$ to the form
\begin{equation}
\label{Q_sseq}
	E_2^{s,t} = \Ext_{\cA(2)}^{s,t}(\mathcal L, \Z_2) \Longrightarrow \pi_{t-s}(\mathcal Q)_2^\wedge;
\end{equation}
we will then add on the summands coming from $\Omega_*^\Spin$ to obtain the groups in the theorem statement. The first
thing we need is $\Ext_{\cA(2)}$ of $\textcolor{BrickRed}{M_2}$, $\textcolor{RedOrange}{M_3}$, $\textcolor{Goldenrod!67!black}{M_4}$,
$\textcolor{Green}{M_5}$, $\textcolor{PineGreen}{M_7}$, $\textcolor{MidnightBlue}{M_8}$, and
$\textcolor{Fuchsia}{M_9}$.
\begin{enumerate}
	\item Davis-Mahowald~\cite[Table 3.2]{DM78} compute $\Ext(\textcolor{BrickRed}{M_2})$.
        \item In degrees $14$ and below, $\textcolor{RedOrange}{M_3}$ is isomorphic to $\Sigma^8 \cA(2)\otimes_{\cA(1)}\Z_2$ (meaning the quotients of these modules by their submodules of elements in degrees $15$ and above are isomorphic). Therefore when $t-s\le 14$, there is an isomorphism
        \begin{subequations}
        \begin{equation}
            \Ext_{\cA(2)}^{s,t}(\textcolor{RedOrange}{M_3}, \Z_2) \cong \Ext_{\cA(2)}^{s,t}(\cA(2)\otimes_{\cA(1)}\Z_2, \Z_2),
        \end{equation}
        and the change-of-rings theorem (see, e.g., \cite[\S 4.5]{BC18}) implies that in all degrees,
        \begin{equation}
            \Ext_{\cA(2)}(\cA(2)\otimes_{\cA(1)}\Z_2, \Z_2)\cong \Ext_{\cA(1)}(\Z_2, \Z_2).
        \end{equation}
        Liulevicius~\cite[Theorem 3]{Liu62} first calculated the algebra $\Ext_{\cA(1)}(\Z_2, \Z_2)$.
        \end{subequations}
	\item As an $\Ext(\Z_2)$-module, $\Ext(\textcolor{Goldenrod!67!black}{M_4})\cong\Z_2[h_0]$ with
	$h_0\in\Ext^{1,1}$~\cite[(2.43)]{Deb23}.
	\item $\Ext(\textcolor{Green}{M_5})$ is computed in~\cite[Figure 2]{Deb23}.
	\item Finally, for $\textcolor{PineGreen}{M_7}$, $\textcolor{MidnightBlue}{M_8}$, and $\textcolor{Fuchsia}{M_9}$,
	we only need to know their Ext groups in degrees $12$ and below. For $i = 7,8,9$, there is a surjective map
	$M_i\to\Sigma^{12}\Z_2$ whose kernel is concentrated in degrees $14$ and above, so (e.g.\ using the long exact
	sequence in Ext associated to a short exact sequence of $\cA(2)$-modules~\cite[\S 4.6]{BC18}) for $t-s \le 12$,
	$\Ext$ of each of these modules is isomorphic to $\Ext(\Sigma^{12} \Z_2)$, which was computed by May
	(unpublished) and Shimada-Iwai~\cite[\S 8]{SI67}.
\end{enumerate}
These assemble into a description of the $E_2$-page of~\eqref{Q_sseq} (compare~\cite[Proposition 2.46]{Deb23}).
\begin{prop}
The $E_2$-page of the Adams spectral sequence for $\mathcal Q$ in degrees $t-s\le 12$ is as displayed in
\cref{heterotic_SS}. In this range, the $E_2$-page is generated as an $\Ext_{\cA(2)}(\Z_2)$-module by ten elements:
\begin{itemize}
	\item $p_1\in\Ext^{0,1}$, $p_3\in\Ext^{0,3}$, $p_7\in\Ext^{0,7}$, and $b\in \Ext^{2,10}$, coming from
	$\Ext(\textcolor{BrickRed}{M_2})$;
        \item $a_1\in\Ext^{0,8}$ and $a_3\in\Ext^{3,15}$, coming from $\Ext(\textcolor{RedOrange}{M_3})$.
	\item $a_2\in\Ext^{0,8}$, coming from $\Ext(\textcolor{Goldenrod!67!black}{M_4})$.
	\item $c\in\Ext^{0,9}$ and $d\in\Ext^{0,11}$, coming from $\Ext(\textcolor{Green}{M_5})$.
	\item $e\in\Ext^{0,12}$, coming from $\Ext(\textcolor{PineGreen}{M_7})$.
	\item $f\in\Ext^{0,12}$, coming from $\Ext(\textcolor{MidnightBlue}{M_8})$.
	\item $g\in\Ext^{0,12}$, coming from $\Ext(\textcolor{Fuchsia}{M_9})$.
\end{itemize}
\end{prop}

\begin{figure}[!htbp]
\centering
\includegraphics[width=\textwidth]{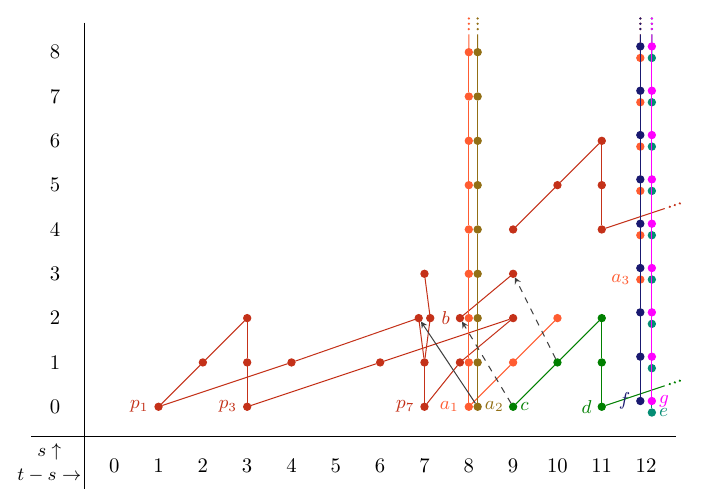}
\caption{The $E_2$-page of the Adams spectral sequence computing
$\tmf_*((BG_{16,16})^{V-32})_2^\wedge$. In \cref{hetdiffs} we show that $d_2(a_2) = h_2^2p_1$ and that many other differentials vanish. We do not know the values of $d_2(c)$ or $d_2(h_1c)$, which is why those differentials are denoted with dotted lines.}
\label{heterotic_SS}
\end{figure}

The next step is to evaluate the differentials. Unlike the other Adams spectral sequences we considered in this
paper, there are several differentials to address, even after using that differentials commute with the action of
$h_0$, $h_1$, and $h_2$:
\begin{itemize}
	\item $d_2$ on $a_1$, $a_2$, $a_3$, $c$, $d$, $e$, $f$, and $g$,
	\item $d_3$ on $a_1$, $a_2$, and $a_3$, and
	\item $d_4$, $d_5$, and $d_6$ on $e$, $f$, and $g$.
\end{itemize}
The argument is nearly the same as in~\cite[Lemmas 2.47, 2.50, and 2.56]{Deb23}.
\begin{lem}
\label{hetdiffs}
$d_2(a_2) = h_2^2 p_1$, and all differentials vanish on $a_1$, $a_3$, $e$, $f$, and $g$.
\end{lem}
\begin{proof}
If $\xi'$ denotes the
tangential structure identical to $\mathbb G_{16,16}$ except with $K(\Z, 4)$ in place of $B\Spin(16)$, then the
class $\tfrac 12 p_1$, interpreted as a map $B\Spin(16)\to K(\Z, 4)$, induces a map of tangential structures from
$\mathbb G_{16,16}$-structure to $\xi'$-structure, hence also a
map of Thom spectra, hence a map of Adams spectral sequences. The $E_2$-page for $\Omega_*^{\xi'}$ is computed
in~\cite[Figure 3]{Deb23} in the range $t-s\le 12$, and looks very similar to our $E_2$-page in
\cref{heterotic_SS}; using the comparison map between these two spectral sequences, we conclude the differentials
in the lemma statement.
\end{proof}
The comparison map would also tell us $d_2(c)$, except that the fate of this differential in $\xi'$-bordism is not
known.

Lastly, we address the class $d\in E_2^{0,11}$. Since $d$ has topological degree $11$, its fate affects the size of $\Omega_{11}^{\G_{16,16}}$, hence the possible anomaly theories for the $\Spin(16)\times\Spin(16)$ theory.
\begin{defn}
\label{y11_defn}
Embedding each $S^k\hookrightarrow \R^{k+1}$ and using the notation $(\vec x, \vec y, \vec z)$ for a vector in $\R^5\times\R^5\times\R^4$, let $\Z_2$ act on $S^4\times S^4\times S^3$ by the involution
\begin{equation}
\label{y11_involution}
    (\vec x, \vec y, \vec z)\longmapsto (-\vec y, -\vec x, -\vec z).
\end{equation}
This action is free on $S^4\times S^4\times S^3$; let $Y_{11}$ denote the quotient.
\end{defn}
$Y_{11}$ is an $(S^4\times S^4)$-bundle over $\RP^3$.
\begin{lem}
$Y_{11}$ has a spin structure.
\end{lem}
\begin{proof}
To prove this, we will stably split the tangent bundle of $Y_{11}$. This is a standard technique; for more examples from a similar perspective, see~\cite[\S 5.2, \S 5.5.2]{Freed:2019sco}, \cite[Examples 14.51 and 14.54; Lemma 14.56; Propositions 14.74, 14.83, and 14.101]{Debray:2023yrs}, and~\cite[Lemma 2.68]{Deb23}.

Recall that, since the normal bundle to $S^k\hookrightarrow\R^{k+1}$ is trivialized by the unit outward normal vector field $\vec v$, there is an isomorphism $\phi\colon TS^k\oplus\underline\R\cong\underline\R^{k+1}$; since $\vec v$ is $\O(k+1)$-invariant, $\phi$ promotes to an isomorphism of $\O(k+1)$-equivariant vector bundles, where $\O(k+1)$ acts trivially on the normal bundle and via the defining representation on $\underline\R^{k+1}$.

Applying this thrice, we have an isomorphism of vector bundles
\begin{equation}
\label{S4S4S3}
    T(S^4\times S^4\times S^3)\oplus\underline\R^3\overset\cong\longrightarrow \underline\R^5\oplus\underline\R^5\oplus\underline\R^4.
\end{equation}
The $\Z_2$-action on $S^4\times S^4\times S^3$ we used to define in $Y_{11}$ in \cref{y11_defn} extends to a linear action on $\R^5\times\R^5\times\R^4$, upgrading~\eqref{S4S4S3} to an isomorphism of $\Z_2$-equivariant vector bundles. In a little more detail:
\begin{itemize}
    \item $\Z_2$ acts on $T(S^4\times S^4\times S^3)$ as the derivative of the involution~\eqref{y11_involution}.
    \item $\Z_2$ acts on $\underline\R^5\oplus\underline\R^5\oplus\underline\R^4$ as the $\Z_2$-representation described by the same formula~\eqref{y11_involution}.
    \item $\Z_2$ acts on the normal $\underline\R^3$ by inverting and swapping the first two coordinates, and inverting the third: $(x, y, z)\mapsto (-y, -x, -z)$.\footnote{In particular, unlike most of the standard examples of the stable splitting technique, the normal bundle is \emph{not} equivariantly trivial. This is because the image of the $\Z_2$-representation in $\O(14)$ is not contained in the subgroup $\O(5)\times\O(5)\times\O(4)$.}
\end{itemize}
The isomorphism~\eqref{S4S4S3} of $\Z_2$-equivariant vector bundles descends through the quotient by $\Z_2$ to an isomorphism of vector bundles on $Y_{11}$; trivial bundles made equivariant by a $\Z_2$-representation descend to vector bundles associated to that representation and the principal $\Z_2$-bundle $\pi\colon S^4\times S^4\times S^3\to Y_{11}$.

In particular, if $\sigma_\pi\to Y_{11}$ denotes the line bundle associated to $\pi$ and the sign representation $\sigma$ of $\Z_2$ on $\R$ and $\R$ denotes the trivial representation, then the $\Z_2$-representation $(x, y)\mapsto (-y, -x)$ on $\R^2$ is isomorphic to $\sigma\oplus\R$. Using this, we obtain an isomorphism of vector bundles
\begin{subequations}
\begin{equation}
    TY_{11}\oplus \sigma_\pi\oplus\underline\R^2 \overset\cong\longrightarrow \sigma_\pi^{\oplus 5}\oplus\underline\R^5 \oplus \sigma_\pi^{\oplus 4}.
\end{equation}
Therefore we have an isomorphism of \emph{virtual} vector bundles
\begin{equation}
\label{y11_TM}
    TY_{11} \underset{\text{virt.}}{\overset\cong\longrightarrow} \sigma_\pi^{\oplus 8} + \underline\R^3.
\end{equation}
\end{subequations}
For any vector bundle $V$, $V^{\oplus 4}$ is spin, as can be verified with the Whitney sum formula, and the existence of a spin structure is an invariant of the virtual equivalence class of a vector bundle, so we can conclude.
\end{proof}
\begin{prop}
\label{y11_is_gen}
$Y_{11}$ admits a $\G_{16,16}$-structure such that the bordism invariant
\begin{equation}
    \int_{Y_{11}} w_4^Lw_4^R x^3 = 1.
\end{equation}
\end{prop}
\begin{proof}
The following data describes a $\G_{16,16}$-structure on $Y_{11}$: identify $S^4 = \HP^1$ and consider the tautological quaternionic line bundle $L \to \HP^1$ on the first $S^4$ factor, and $L^* :=\Hom_\H(L, \underline\H)$ on the second $S^4$ factor. These have associated $\Sp(1) = \Spin(3)$ bundles; inflate via $\Spin(3) \hookrightarrow \Spin(16)$ to obtain a $(\Spin(16)\times\Spin(16))$-bundle on $S^4\times S^4\times S^3$. The two $\Spin(16)$-bundles are switched when one applies the involution~\eqref{y11_involution}, so on the quotient $Y_{11}$, we obtain a principal $G_{16,16}$-bundle $P\to Y_{11}$.

To verify the claim in the first sentence of our proof, we need to check that a spin structure on $Y_{11}$ and the principal $G_{16,16}$-bundle $P\to Y_{11}$ satisfy the Green-Schwarz condition $\tfrac 12 p_1(TY_{11}) + \tfrac 12 p_1(V^L) + \tfrac 12 p_1(V^R) = 0$. In fact, the two parts of this expression vanish separately.
\begin{itemize}
    \item In~\eqref{y11_TM}, we learned that $TY_{11}$ is virtually equivalent to $\sigma_\pi^{\oplus 8}\oplus\underline\R^3$. This bundle turns out to admit a string structure, meaning $\tfrac 12 p_1(TY_{11}) = 0$. It suffices to prove that $\sigma^{\oplus 8}\to B\Z_2$ admits a string structure, where $\sigma$ is the tautological line bundle. To see this, recall that $\sigma^{\oplus 4}$ (like the sum of $4$ copies of any vector bundle) is spin, so $\sigma^{\oplus 8}\cong\sigma^{\oplus 4}\oplus \sigma^{\oplus 4}$ factors $\sigma^{\oplus 8}$ as the direct sum of two spin vector bundles. Then use the Whitney sum formula for $\tfrac 12 p_1$ of a direct sum of spin vector bundles~\cite[Lemma 1.6]{Deb23} to conclude that in $H^4(B\Z_2;\Z)$,
    \begin{equation}
        \frac 12 p_1(\sigma^{\oplus 8})  = 2\cdot \frac 12 p_1(\sigma^{\oplus 4}).
    \end{equation}
    Maschke's theorem implies that for $k \ge 1$, multiplication by $2$ kills all elements in $H^k(B\Z_2;\Z)$, so $\tfrac 12 p_1(\sigma^{\oplus 8}) = 0$.
    \item The bundles $L$ and $L^*$ over $S^4$ have inverse values of $p_1$, hence also of $\tfrac 12 p_1$ (since $H^4(S^4;\Z)$ is torsion-free, the latter follows from the former). Therefore when we descend from $S^4\times S^4\times S^3$ to $Y_{11}$, the class $\tfrac 12 p_1(V^L) + \tfrac 12 p_1(V^R)$ is $0$.
\end{itemize}
Finally, we need to verify $\int_{Y_{11}} w_4^Lw_4^R x^3 = 1$. Since $H^{11}(Y_{11};\Z_2)\cong\Z_2$, it suffices to show that the pullback of $w_4^Lw_4^Rx^3\in H^{11}(BG_{16,16};\Z_2)$ along the classifying map $f_P\colon Y\to BG_{16,16}$ for $P\to Y_{11}$ is nonzero. To do this, first factor $f_P$ into the following diagram of three fibrations:
\begin{equation}
\label{triple_fibration}
\begin{tikzcd}
	{S^4\times S^4} & {S^4\times S^4} & {B\Spin(16)\times B\Spin(16)} \\
	& {Y_{11}} & {\mathcal X} & {BG_{16,16}} \\
	& {\RP^3} & {B\Z_2} & {B\Z_2.}
	\arrow[from=3-2, to=3-3]
	\arrow[from=2-2, to=3-2]
	\arrow[from=1-1, to=2-2]
	\arrow[from=2-2, to=2-3]
	\arrow[from=2-3, to=3-3]
	\arrow[from=2-4, to=3-4]
	\arrow[from=1-3, to=2-4]
	\arrow[from=1-2, to=2-3]
	\arrow[Rightarrow, no head, from=3-3, to=3-4]
	\arrow[from=2-3, to=2-4]
	\arrow[Rightarrow, no head, from=1-1, to=1-2]
	\arrow["j", from=1-2, to=1-3]
	\arrow["\lrcorner"{anchor=center, pos=0.125}, draw=none, from=2-2, to=3-3]
	\arrow["{f_P}"'{pos=0.3}, curve={height=12pt}, dashed, from=2-2, to=2-4]
\end{tikzcd}\end{equation}
Here $\mathcal X$ is the $S^4\times S^4$-bundle over $B\Z_2$ defined analogously to $Y_{11}$ but using $S^\infty = E\Z_2$ instead of $S^3$. The map $j\colon S^4\times S^4\to B\Spin(16)\times B\Spin(16)$ is the map classifying $L$ and $L^*$.

The diagram~\eqref{triple_fibration} induces maps between the Serre spectral sequences of the three fibrations; using it, one can compute the pullback of $w_4^Lw_4^Rx^3$ to $Y_{11}$ and see that it is nonzero, as promised.
\end{proof}
\begin{cor}
In the Adams spectral sequence in \cref{heterotic_SS}, $d$ survives to the $E_\infty$-page; in particular, $d_2(d) = 0$.
\end{cor}
\begin{proof}
We reuse the strategy from \cref{gens_diffs}: since $d$ is in filtration $0$, it corresponds to some characteristic class $c\in H^{11}(BG_{16,16};\Z_2)$, and $d$ survives to the $E_\infty$-page if and only if there is some closed $11$-dimensional $\G_{16,16}$-manifold $M$ such that $\int_M c = 1$. By inspection of \cref{A2_swap}, $c = w_4^Lw_4^Rx^3$, so by \cref{y11_is_gen} we can take $M = Y_{11}$.
\end{proof}
The last step in this calculation is to address extensions. The argument is nearly identical to~\cite[Lemma 2.59 and Proposition 2.60]{Deb23}, though one now has the extra classes $h_0^k a_1$ for $k\ge 0$, $h_1a_1$, and $h_1^2 a_1$ in degrees $8$, $9$, and $10$ respectively which were not present in the $E_8\times E_8$ spectral sequence. Fortunately, this new ambiguity is fully resolved by applying the ``$2\eta = 0$ trick'' to classes of the form $h_1x$ in standard ways, for example as in~\cite[Corollary F.16(2)]{KPMT20}, \cite[(5.47)]{Deb21}, \cite[Lemmas 14.29 and 14.33]{Debray:2023yrs}, and \cite[Lemma 2.59]{Deb23}, and one learns that there are no hidden extensions in degrees $10$ and below. Unfortunately, just like in the $E_8\times E_8$ case~\cite[Theorem 2.62]{Deb23}, we have not ruled out the possibility of a hidden extension in $\Omega_{11}^{\G_{16,16}}$.
\end{proof}
The generators described in~\cite[\S 2.2.1, \S 2.2.2]{Deb23} for $\Omega_*^{\xi^{\mathrm{het}}}$ pull back to
generate most of the corresponding $\mathbb G_{16,16}$ bordism groups: the difference between a $\G_{16,16}$-structure and a $\xi^{\mathrm{het}}$-structure is that in the latter, $\Spin(16)$ is replaced by $E_8$, so to support our claim that the generators there pull back to generators of $\G_{16,16}$-bordism, we must argue that the $(E_8\times E_8)\rtimes\Z_2$-bundles are induced from $G_{16,16}$-bundles. As usual we may replace $BE_8$ with $K(\Z, 4)$, so this amounts to checking that for the generating manifolds in~\cite[\S 2.2.1, \S 2.2.2]{Deb23}, the degree-$4$ classes entering the Green-Schwarz mechanism can be written as $\tfrac 12 p_1(V)$ for some rank-$16$ spin vector bundle $V$. By adding trivial summands, we may use lower-rank spin vector bundles.

By inspection of the list of generators in~\cite[\S 2.2.1, \S 2.2.2]{Deb23}, it suffices to show this for $\HP^2$: the rest of the list of generators there either have degree-$4$ classes equal to $0$, have their $\xi^{\mathrm{het}}$-structure induced from a spin structure (so that we may use the tangent bundle to define the $\G_{16,16}$-structure as in the proof of \cref{spin_split}), or are products of manifolds otherwise accounted for. For $\HP^2$, the degree-$4$ classes come from the tautological quaternionic line bundle, hence define a $\G_{16,16}$-structure.

Thus the list of generators in~\cite[\S 2.2.1, \S 2.2.2]{Deb23} accounts for most of the generators of the $\G_{16,16}$-bordism groups we have computed. A few manifolds are as yet unaccounted for.
\begin{enumerate}
    \item There is an $8$-dimensional $\G_{16, 16}$-manifold $Y_8$ generating a $\Z$ and whose image in the Adams $E_\infty$-page is $a_1$. The $\Z_2$ summands lifting $h_1a_1$ and $h_1^2a_1$ are also unaccounted for, and can be generated by $Y_8\times S_{\mathit{nb}}^1$, resp.\ $Y_8\times S_{\mathit{nb}}^1\times S_{\mathit{nb}}^1$.
    \item Depending on the fate of $d_2(c)$, there may be a $\Z_2$ summand in $\Omega_9^{\G_{16,16}}$ whose generator lifts the class $c\in E_\infty^{9,0}$. In~\cite[\S 2.2.1]{Deb23} no generator was provided and we also do not know what manifold this would be.
    \item A generator lifting the class $d\in E_\infty^{0,11}$ was left as an open question in~\cite[\S 2.2.1(11)]{Deb23}. Thanks to \cref{y11_is_gen}, we can choose $Y_{11}$.
\end{enumerate}
Thus we have found generators for all classes except for $a_1$, $c$, and $h_1c$; $c_1$ and $hc_1$ may or may not be trivial, depending on the value of an Adams differential.

\subsection{Cancelling the anomaly}
Now that we have the generators of $\Omega_{11}^{\mathbb G_{16,16}}$ in hand, we proceed to calculate the partition function of the anomaly theory on these generators and show that it is trivial. We are able to do this without knowing the isomorphism type of $\Omega_{11}^{\mathbb G_{16,16}}$, similarly to Freed-Hopkins' approach in~\cite{Freed:2019sco}.
\begin{thm}
\label{non_susy_swap_anomaly}
Let $\alpha$ denote the anomaly field theory for the $\Spin(16)^2$ heterotic string on $\mathbb G_{16,16}$-manifolds. Then $\alpha$ is isomorphic to the trivial theory.
\end{thm}
\begin{proof}
Recall that $\alpha\cong \alpha_f\otimes\alpha_{X_8}$, where $\alpha_f$ is the anomaly of the fermionic fields and $\alpha_{X_8}$ is the anomaly coming from the term $-\int B_2\wedge X_8$~\eqref{eq:GSterm10d} that the Green-Schwarz mechanism adds to the action, as we discussed in \S\ref{sec:local}.

We will calculate $\alpha$ on a generating set for $\Omega_{11}^{\mathbb G_{16,16}}$.  Based on \cite{Debray:2023yrs} and the above discussion, the two generators are
\begin{equation}B\times \mathbb{RP}^3\quad\text{and}\quad Y_{11} \label{e323}\end{equation}
with $\mathbb G_{16,16}$-structures described in the previous subsection, corresponding physically to turning on appropriate gauge bundles. Here $B$ is a \emph{Bott manifold}, i.e.\ a closed spin $8$-manifold satisfying $\widehat A(B) = 1$, and indeed any choice of $B$ that admits a string structure may be used in this computation. We will use the Bott manifold constructed by Freed-Hopkins in~\cite[\S 5.3]{Freed:2019sco}; those authors show $\tfrac 12 p_1(B) = 0$, so $B$ is string, and that $p_2 = -1440 b$, where $b$ is a generator of $H^8(B;\Z)\cong\Z$. Any other Bott manifold is cobordant to this one, so we will not botter studying them all.

The first generator we evaluate $\alpha$ on is $B\times\RP^3$. As discussed in~\cite[\S 2.2.1]{Deb23}, this generator has $G_{16,16}$-bundle induced from the principal $\Z/2$-bundle $S^3\times B\to\RP^3\times B$ and any inclusion $\Z_2\hookrightarrow G_{16,16}$ complementary to the normal $\Spin(16)^2$ subgroup. From a physics point of view, this means the $\Spin(16)$ gauge bundles are trivial: the $\Z_2$ symmetry switches two copies of the trivial bundle. This implies $X_8 = 0$, so $\alpha_{X_8}$ is trivial. For $\alpha_f$, we must calculate the $\eta$-invariants of the spinor bundles associated to the gauge bundles. We will first dimensionally reduce our theory on $B$, to obtain a 2d effective theory, and study the corresponding anomaly, which is the dimensional reduction of $\alpha_f$, on $\RP^3$. As the defining property of a Bott manifold is that the Dirac index is 1, and the gauge bundle is switched off in our example, the 2d spectrum is identical to the ten-dimensional one, so showing that the anomaly on $\RP^3$ is trivial will imply $\alpha_f(B\times\RP^3) = 1$.

We need to know the gauge bundle on $\RP^3$,\footnote{In general keeping track of tangential structures on dimensional reductions can be complicated (see, e.g., \cite[\S 9]{SP18}, but because $B$ has a string structure and the tangential structure of the theory is a twisted string structure, we do not need to worry about this detail.} but because the gauge bundle is trivial on $B$, we can describe the $G_{16,16}$-bundle on $\RP^3$ as induced from the $\Z_2$-bundle $S^3\to\RP^3$. Thus we should see how the $G_{16,16}$-representations describing the fermions branch when we restrict to $\Z_2$.
\begin{itemize}
\item The 10d fermions in the $(\mathbf{128},\mathbf{1})\oplus(\mathbf{1},\mathbf{128})$ give a total of 128 2d fermions transforming as singlets of the swap, and another 128 transforming in the sign representation.
\item For the 10d fermions in the $(\mathbf{16},\mathbf{16})$, the swap is implemented via a matrix with sixteen blocks each having eigenvalues $\pm1$, again giving 128 fermions in each of the trivial and sign representations of the swap $\mathbb{Z}_2$. Since the 10d fermions have opposite chirality to those in the previous point, the resulting 2d fermions also come in opposite chirality.
\end{itemize}
With these matter assignments, we obtain a total of $128$ $\mathbb{Z}_2$ charged fermions of each chirality, which collectively are anomaly-free (and therefore, gravitational anomalies cancel). Therefore, there is no anomaly under the swap on any background, such as $\mathbb{RP}^3$: the $\eta$-invariants all cancel out. Thus $\alpha_f(B\times\RP^3)$ vanishes and the overall anomaly $\alpha_f\otimes\alpha_{X_8}$ vanishes on $B\times\RP^3$.

For $Y_{11}$, which is an $(S^4\times S^4)$-bundle over $\RP^3$, we perform a twisted compactification on $S^4\times S^4$ and study the anomaly of the resulting 2d theory on $\RP^3$. Because $Y_{11}$ is not a product, we must take a little more care with this procedure, but it is not so difficult to show that the assignment from a string $3$-manifold $N$ with principal $\Z/2$-bundle $P\to N$ to the manifold
\begin{equation}
    \kappa(N) := (S^4\times S^4)\times_{\Z_2} N,
\end{equation}
where the two copies of $S^4$ are given the same $\Z_2$-action and $\Spin(16)$-bundles as we used in the construction of $Y_{11}$, produces a $\G_{16,16}$-manifold for all $N$ and is compatible with bordism, allowing $\kappa$ to define a functor of bordism categories and therefore a twisted compactification as promised.

The covering $S^4$ has $\Spin(16)^2$ bundles characterized by a second Chern class
\begin{equation}
    c_{2}^{SO(16),i}=(-1)^i( b_1+b_2),
\end{equation}
where $b_1,b_2$ are the volume forms of both $S^4$ factors. Now, rather than explicitly computing the dimensional reductions of $\alpha_f$ and $\alpha_{X_8}$ on $\RP^3$, we take advantage of the fact that $\alpha$ is a deformation invariant, so we may deform our 2d theory into something where the value of the anomaly on $\RP^3$ is more obviously trivial.\footnote{The $SO(16)\times SO(16)$ string is non-supersymmetric, and therefore the deformations we have just outlined in the previous paragraph may be obstructed dynamically; for instance, there may be a potential obstructing the small instanton limit. However, since we only wish to compute the anomaly, we may ignore such effects; the only ingredient we really need is the fact, proven in Section \ref{sec:so16so16_inflow}, that in the small instanton limit the anomaly becomes symmetric between both $\Spin(16)$ factors.} Specifically, 
we can take a limit in moduli space where the instantons become singular and pointlike, turning into a non-supersymmetric version of the heterotic NS5-brane; as explained in Section \ref{sec:so16so16_inflow}, the resulting theory becomes symmetric between the two $\Spin(16)$ factors, implying that, just like in the supersymmetric heterotic string theories, small instantons of both gauge factors are identified. After deforming in this way the gauge bundle on both $\Spin(16)$ factors, we are left with a single pointlike NS5 and a single anti-NS5 in each sphere, which annihilate, leading to a trivial and therefore anomaly-free configuration for the compactified theory, and implying that $\alpha(Y_{11}) = 1$.
\end{proof}
As a bonus, we can answer a question of~\cite{Deb23}, giving a bordism-theoretic argument for the analogous anomaly cancellation question for the $E_8\times E_8$ heterotic string. This anomaly cancellation result was first established by Tachikawa-Yamashita in~\cite{Tachikawa:2021mby} by a different argument.
\begin{cor}
\label{susy_swap_anomaly}
The anomaly field theory $\alpha$ for the $E_8\times E_8$ heterotic string theory taking into account the $\Z_2$ swap symmetry is trivial.
\end{cor}
\begin{proof}
The argument for $Y_{11}$ also works in the supersymmetric $E_8\times E_8$ theory, since the instantons may also be embedded in $E_8$. In the supersymmetric case, the pointlike limit of the instanton is the ordinary, supersymmetric heterotic NS5-brane, as illustrated in \cite{Horava:1996ma}, so the $E_8^2$ anomaly vanishes on $Y_{11}$.

For $B\times\RP^3$, dimensional reduction leads to 248 singlet and 248 fermions (from the $E_8$ adjoints) charged under the sign representation. Since the relevant anomaly is controlled by
\begin{equation}\Omega_{3}^{\text{Spin}}(B\mathbb{Z}_2)=\mathbb{Z}_8,\end{equation}
and 248 is a multiple of 8, we conclude there is no swap anomaly either.  Finally, we already know that gravitational anomalies cancel in $B\times \mathbb{RP}^3$, since if we forget about the swap this is just an ordinary string background. 
\end{proof}
In summary, we have shown that anomalies vanish under both generators of the swap bordism group, both for Spin$(16)^2$ and $E_8\times E_8$. The supersymmetric case is covered by the worldsheet analysis in \cite{Tachikawa:2021mby}, which takes into account twists including the swap we just discussed. Thus, we recover a special case of the general anomaly cancellation result there. On the other hand, our approach covers the  non-supersymmetric $SO(16)^2$ case (for the case of geometric target spaces only).

\section{Conclusions} \label{conclus}

Our world is non-supersymmetric, and that fact alone means that non-supersymmetric corners of the string landscape warrant much more attention than they have received so far, both as a source of interesting backgrounds that might connect more directly to our universe, as well as a new trove of data to check and refine Swampland constraints. In this paper we have moved a bit in this direction by computing the bordism groups and anomalies associated to twisted string structures in the three known non-supersymmetric, tachyon free string models in ten dimensions. The results we obtained are summarized in Table \ref{tbor} for the Sugimoto and $\Spin(16)^2$ groups; for the more complicated Sagnotti 0'B model, we were just able to show that there is a potential $\Z_2$ anomaly. 

\begin{table}[!htbp]
\centering
\begin{tabular}{c| c c c c}
$k$ & $\Omega_k^{\text{String}-\Spin(16)^2} $ &$\Omega_k^{\mathbb{G}_{16,16}} $ & $\Omega_k^{\text{String}-\Sp(16)} $ & $\Omega_k^{\String\text{-}\SU(32)\ang{c_3}}$\\\hline
$0$ &$\Z$&$\Z$ & $\Z$ & $\Z$\\
$1$ &$\Z_2$& $\Z_2^2$& $\Z_2$ & $\Z_2$\\
$2$ &$\Z_2$&$\Z_2^2$ & $\Z_2$ & $\Z_2$\\
$3$ &$0$&$\Z_8$ & $0$ & $0$\\
$4$ &$\Z^2$&$\mathbb{Z}\oplus\Z_2$ & $\Z$ & $\Z$\\
$5$ &0& 0& $\Z_2$ & $\Z_2$\\
$6$ &0&$\Z_2$ & $\Z_2$ & $\Z_2$ or $\Z_2^2$\\
$7$ &0&$\Z_{16}$ & 0 & 0 or $\Z_2$ \\
$8$ &$\Z^6$&$\Z^3\oplus\Z_2^i$ & $\Z^3$ &  $\Z^3\oplus \Z_2$ or  $\Z^3\oplus \Z_2^2$\\
$9$ &$\Z_2^5$& $\Z_2^j$& $\Z_2^2$ & $\Z_2^2$\\
$10$ &$\Z_2^7$& $\Z_2^k$& $\Z_2^3$ & $\Z\oplus\Z_2$ or $\Z\oplus\Z_2^2$\\
$11$ &0& $A$ & $0$& $0$ or $\Z_2$
\end{tabular}
\caption{Twisted string bordism groups computed in this paper for the $\Spin(16)^2$ theory with and without including the swap (second and third columns), for the Sugimoto string (fourth column), and for the Sagnotti string (fifth column). In the second column, $i,j,k$ are unknown integers, and $A$ is an abelian group of order 64 (see Section \ref{sec:swapping} for details). In the fifth column, there are ambiguities due to undetermined differentials in the Adams spectral sequence; see \S\ref{ss:sagnotti} for details. In some cases, the bordism group vanishes in degree $11$, which automatically implies the corresponding theory has no anomalies; we also show the anomaly can be trivialized for the $\Z_2$ outer automorphism of the $\Spin(16)^2$ string, even though the bordism group is nonzero.
The results in this table can be further used to classify bordism classes and predict new solitonic objects in these non-supersymmetric string theories following \cite{McNamara:2019rup,Debray:2023yrs}.  }
\label{tbor}
\end{table}

From the results of the table, it is clear that both $\Spin(16)^2$ and Sugimoto models are free of global anomalies. One might have expected this from the fact that they have a consistent worldsheet description. However, there can be non-perturbative consistency conditions that are not automatically satisfied by the existence of a consistent worldsheet at one-loop, see for instance \cite{Uranga:2000xp}, where a $K$-theory tadpole is not detected by the closed string sector. It would be very interesting to determine in full generality whether existence of a consistent worldsheet is sufficient to guarantee consistence of the target spacetime. Although we have not settled the question of consistency in the Sagnotti string, we expect that it is also free of anomalies; for instance, upon circle compactification, it can be related to a ``hybrid'' type I' setup involving an $O8^+$ plane and an $\overline{O8^-}$, both of which are individually consistent \cite{Dudas:2000sn}. 

Perhaps the more  interesting result of our work is Table \ref{tbor} itself, listing the bordism groups of the $\Spin(16)^2$ and Sugimoto theories. An obvious follow-up to this paper is to use the Cobordism Conjecture \cite{McNamara:2019rup} together with the groups in Table \eq{tbor} to predict new, non-supersymmetric objects in the non-supersymmetric string theories, similarly to what has been done in type II in \cite{Debray:2023yrs}. While it is natural to expect these new branes to be non-supersymmetric, it may be worthwhile to pursue this direction in more detail.

One subtlety that must be kept in mind, when considering our results, is that we did not necessarily use the correct global form of the gauge group in our calculations. With the exception of the $SO(16)^2\rtimes\Z_2$, we focused on simply connected versions of all the groups, which immensely simplified the calculations. Since any bundle before taking a quotient is still an allowed bundle after taking the quotient, our results show that a very large class of allowed bundles in the Sugimoto and $SO(16)^2$ theories are anomaly free\footnote{The equivalence discussed at the beginning of Subsection \ref{sec:so16so16_inflow} shows that anomalies vanish for $SO(16)^2$ even when the correct global form is taken into account.}, but particularly in the Sugimoto case there may be more bundles to check if the gauge group is actually $\Sp(16)/\Z_2$. In the type I theory, we know the group is $\Spin(32)/\Z_2$ and not just $\SO(32)$ due to the presence of $K$-theory solitons transforming in a spinorial representation \cite{Witten:1998cd}. In the Sugimoto theory, the relevant $K$-theory is symplectic, and there do not seem to be any such solitonic particles \cite{Sugimoto:1999tx}, suggesting that the group might actually be $\Sp(16)/\Z_2$. It would be interesting to elucidate this point and figure out whether there really are any global anomalies beyond those studied here.

Another result of our paper is a series of arguments and checks that in any heterotic string theory, the Bianchi identity must hold at the level of integer coefficients. Furthermore, satisfying the Bianchi identity even at the level of integer coefficients is not enough to guarantee consistency of the string background; there is also a consistency condition (tadpole) that is detected by $H^3(M;\mathbb{R})$. The general consistency condition is of course that the anomaly of probe strings vanishes; more generally, it is natural to expect that all consistency conditions (tadpoles) of any quantum gravity background come from consistency of probe branes in said background.

Another limitation of our study is that, by following a (super)gravity approach, we must restrict to studying anomalies on \emph{smooth} backgrounds. String theories, both with and without spacetime supersymmetry, make sense on much larger classes of backgrounds that do not admit a geometric description, such as orbifolds, and which are only analyzed from a worldsheet perspective.  These cases are not covered by our analysis. Using modular invariance one can show that the Green-Schwarz mechanism always cancels local anomalies in any consistent worldsheet background~\cite{Schellekens:1986xh, Schellekens:1986yi}, with or without spacetime supersymmetry.  The question of whether global anomalies also cancel in these non-geometric backgrounds was addressed in~\cite{Tachikawa:2021mvw,Tachikawa:2021mby}, where \emph{all} global anomalies are shown to cancel for all gauge groups and dimensions in the ordinary supersymmetric heterotic string theories. This remarkable result rests on the validity of the Segal-Stolz-Teichner conjecture~\cite{Stolz:2011zj}, which connects deformation classes of worldsheet theories (or, more generally, two-dimensional $(0,1)$ supersymmetric QFTs) to the spectrum of (connective) \emph{topological modular forms} (TMFs)~\cite{douglas2014topological, hopkins} (see also~\cite{Gukov:2018iiq}). The physical interpretation of this more refined generalized cohomology theory is related to ``going up and down RG flows''~\cite{Gaiotto:2019asa}, and it includes the familiar string bordism deformations of the target space manifold of a sigma model as well as more exotic, ``non-geometric'' deformations. 

To construct an ordinary, spacetime-supersymmetric heterotic model, all that one needs is a $(0,1)$ SQFT. Such a QFT always has a notion of a right-moving worldsheet fermion number $F_R^w$, which is gauged by the usual GSO projection to construct a modular-invariant partition function. The original Segal-Stolz-Teichner conjecture applies precisely to $(0,1)$ SQFT's. If we wanted to make such an argument for a spacetime non-supersymmetric string theory (tachyonic or not), we face the obstacle that the GSO projection is different, and it involves \emph{additional} worldsheet symmetries. For instance, the $SO(16)^2$ theory has a ``diagonal'' modular invariant partition function, which requires a notion of a left-moving worldsheet fermion number in addition to the $(0,1)$ SQFT structure. Thus, valid $SO(16)^2$ worldsheet theories are equipped with an additional left-moving $\Z_2$ symmetry, or equivalently, they are equipped with both a spin structure and a $\Z_2$ symmetry. To repeat the argument of \cite{Tachikawa:2021mvw,Tachikawa:2021mby}, one must work with $\Z_2$-equivariant TMF; it would be very interesting to do so. 

When we started this project we were actually quite surprised that we could not find a comment on global anomalies of non-supersymmetric tachyon-free strings\footnote{Other than \cite{osti_5515225}; maybe we just missed it.} anywhere in the literature. After all, these constructions are all 25+ years old, and they have a quite distinguished role in the string landscape. In a sense, they look more like our universe than the more familiar, supersymmetric theories! Maybe the reason for this neglect is simply lack of workforce; the last 25 years have brought so much progress on so many areas that the community just had to focus on the most novel or promising ones, and simply left many important questions unanswered. The physics of non-supersymmetric string theories was a victim to this rapid progress. Despite this, recent research in this direction has yielded e.g. metastable vacua \cite{Basile:2018irz, Antonelli:2019nar}, novel end-of-the-world defects \cite{Buratti:2021fiv, Blumenhagen:2022mqw, Angius:2022aeq, Blumenhagen:2023abk, Angius:2023xtu, Huertas:2023syg, Calderon-Infante:2023ler}, and checks of Swampland constraints \cite{Basile:2020mpt, Basile:2021mkd, Basile:2022zee}. The results that we have presented in this paper are yet another step in this direction. We believe (and hope to have convinced at least some readers) that non-supersymmetric string theories constitute a very interesting arena where there seems to be an abundance of low-hanging fruit that is likely to yield novel lessons both in the Landscape and the Swampland.

\section*{Acknowledgements}

We are very thankful to Luis \'{A}lvarez-Gaum\'{e}, Carlo Angelantonj, Alberto Castellano, Markus Dierigl, Dan Freed, Christian Kneißl, Giorgio Leone, Ignacio Ruiz, Cumrun Vafa, \'{A}ngel Uranga, and Matthew Yu for many interesting discussions and insightful comments. We are particularly grateful to Yuji Tachikawa for explanations on equivariant topological modular forms. MM gratefully acknowledges the support and hospitality of KITP during the ``Bootstrapping Quantum Gravity'' Program, supported in part by the National Science Foundation under Grant No. NSF PHY-1748958, as well as the Aspen Center for Physics, which is supported by National Science Foundation grant PHY-2210452, in the context of the program ``Traversing the Particle Physics Peaks''. The travel of MM to ACP was supported by a grant of the Simons Foundation. During the early stages of this work, MM was supported by  a grant from the Simons Foundation (602883,CV) and by the NSF grant PHY-2013858, by the Maria Zambrano Scholarship (UAM) CA3/RSUE/2021-00493, and currently by the
Atraccion del Talento Fellowship 2022-T1/TIC-23956 from Comunidad de Madrid. The work of MD is supported by the FPI grant
no. FPI SEV-2016-0597-19-3 from Spanish National Research Agency from the Ministry of
Science and Innovation. MM and MD acknowledge the support of the grants CEX2020-001007-S and PID2021-123017NB-I00, funded by MCIN/AEI/10.13039/501100011033 and by ERDF A way of making Europe. Part of this work was completed while IB and AD participated in the ``Geometry and the Swampland'' conference at IFT in 2022; we thank the conference for support. MD also wishes to acknowledge the hospitality of the Arnold Sommerfeld Center for Theoretical Physics (ASC) at the Ludwig-Maximilians Universität throughout a research visit during which part of this work was completed.
\\

\textbf{v3}: We are very thankful to Yuji Tachikawa and Shota Saito for letting us know of and correcting an error in the proof of Lemma 3.59 of an earlier version of this draft. This error affects the identification of correct generators of $\text{Sp}(16)$-twisted string bordism and affects the spectral sequence and the calculation of bordism for the Sugimoto string in degrees 7,8 and 9. We have also corrected similar errors in the Sagnotti string calculation and in the $Spin(16)^2$ (removing Proposition 4.28). The results for the Sagnotti string change also in degrees 7,8,9, while those of the $Spin(16)^2$ string remain unchanged.

\bibliographystyle{JHEP}
\bibliography{anomalies.bib}

\end{document}